\documentclass[12pt]{article}
\usepackage{amsmath}
\usepackage{graphicx,psfrag,epsf}
\usepackage{enumerate}
\usepackage{natbib}
\usepackage{url} 

\newcommand{\blind}{0}

\usepackage{adjustbox}
\usepackage{enumerate}
\usepackage{bm}
\usepackage{amsfonts}
\usepackage[colorlinks,citecolor=blue,urlcolor=blue,hyperfootnotes=false]%
{hyperref}
\usepackage{amsthm}
\usepackage{amssymb}
\usepackage{graphicx}
\usepackage{float}
\usepackage{booktabs}
\usepackage{threeparttable}
\usepackage{amsmath}
\numberwithin{equation}{section}
\usepackage{lscape}
\usepackage{scrextend}
\usepackage{relsize}
\usepackage{multicol}
\usepackage{chngcntr}
\usepackage{etoolbox}
\usepackage{amssymb}
\usepackage{tabularx}
\usepackage{mathrsfs}
\usepackage[english]{babel}
\usepackage{multirow}
\usepackage{booktabs}
\usepackage[labelfont=bf]{caption}
\usepackage{float}
\usepackage{titlesec}
\usepackage{array}
\usepackage{color}
\usepackage{afterpage}
\usepackage{framed}
\usepackage{float}
\usepackage{afterpage}%
\usepackage{comment}
\usepackage{multirow}
\usepackage{accents}
\usepackage[pass]{geometry}
\usepackage{amssymb}
\usepackage{xcolor}
\usepackage{amsmath}
\usepackage{extpfeil}
\usepackage{pdflscape}
\usepackage{natbib}
\setcitestyle{open={(},close={)}}
\allowdisplaybreaks

\newtheorem{theorem}{Theorem}[section]
\newtheorem{corollary}{Corollary}[section]
\newtheorem{proposition}{Proposition}[section]
\newtheorem{lemma}{Lemma}[section]

\theoremstyle{definition} 
\newtheorem{assumption}{Assumption}[section]

\addto\extrasenglish{%
   
}
\addto\extrasenglish{%
   
}
\addto\extrasenglish{%
  
}
\addto\extrasenglish{%
  
}
\addto\extrasenglish{%
  
}
\addto\extrasenglish{%
  
}
\addto\extrasenglish{%
  
}

\hypersetup{
  linkcolor=blue,
  urlcolor=blue}

\begin{document}

\def\spacingset#1{\renewcommand{\baselinestretch}%
{#1}\small\normalsize} \spacingset{1}


\if0\blind
{
  \title{\bf Detecting Copula Structural Changes:\\
A Smooth Testing Approach
} 
  \author{Shiyao Huang\thanks{Department of Business Statistics and Econometrics, Guanghua School of Management, Peking University, Beijing, 100871, China. Email: \texttt{2401111054@gsm.pku.edu.cn}.}\hspace{.2cm}\\
    Peking University\\
    and \\
    Xiaojun Song\thanks{Corresponding author: Department of Business Statistics and Econometrics, Guanghua School of Management, Peking University, Beijing, 100871, China. Email: \texttt{sxj@gsm.pku.edu.cn}. The financial support from the National Natural Science Foundation of China [Grant Numbers 72373007, 72333001, and 72621002] 
is gratefully acknowledged. The author also gratefully acknowledges the research support from the Center for Statistical Science of Peking University.} \\
    Peking University}
  \maketitle
} \fi

\if1\blind
{
  \bigskip
  \bigskip
  \bigskip
  \begin{center}
    {\LARGE\bf 
    Detecting Copula Structural Changes:\\
A Smooth Testing Approach
}
\end{center}
  \medskip
} \fi

\bigskip
\begin{abstract}
This paper develops novel smooth tests for structural changes in the innovation copula of multivariate dynamic models. We characterize deviations from copula constancy through a collection of generalized Fourier coefficients and test their joint significance. Under the null hypothesis, estimation of the dynamic parameters and the unknown marginals has no first-order estimation effect on the proposed statistics. Consequently, the feasible tests based on estimated residuals are asymptotically equivalent to their infeasible counterparts based on the unobserved innovations, yielding a desirable oracle property. Our proposed tests are asymptotically $\chi^2$-distributed and possess nontrivial power against local alternatives that approach the null at the parametric rate. To enhance the practicability of our methods, we further develop a data-driven procedure that automatically selects the truncation orders of the basis expansions. Unlike existing methods based on empirical copula processes or kernel smoothing, our tests require neither computationally intensive bootstrap procedures nor bandwidth selection. Extensive simulations demonstrate the satisfactory empirical size and power of our proposed tests. In particular, the data-driven test delivers substantial power gains under sparse alternatives, while remaining competitive under dense alternatives. Applications to exchange rates and stock returns further illustrate the practical usefulness of the proposed methods.
\end{abstract}

\noindent%
\vfill

\newpage
\spacingset{1.8} 

\section{Introduction}
\label{Introduction}
Copulas provide a flexible framework for modeling multivariate distributions by separating the marginal behavior of individual variables from their dependence structure. Sklar's theorem \citep{sklar1959fonctions} ensures that any joint distribution can be represented through a copula and its marginal distributions, thereby allowing the two components to be specified and estimated separately. This flexibility has led to the widespread use of copula models in risk management, finance, and insurance \citep{frey2001modelling,cherubini2004copula,cherubini2011dynamic,patton2012review,mcneil2015quantitative}. For instance, in economics, copula modeling helps in the prediction of financial contagion and periods of ``boom'' or ``bust'' \citep{dewick2022copula}; in finance, it plays an indispensable role in multivariate option pricing \citep{rosenberg1999semiparametric} and portfolio Value-at-Risk \citep{hull1998value}. 

Building on the separation between marginal behavior and dependence, \citet{chen2006estimation} proposed the semiparametric copula-based multivariate dynamic (SCOMDY) model to address the nonlinearities, non-Gaussian features, and complex comovements commonly observed in multivariate economic and financial time series. The SCOMDY framework parametrically specifies the conditional mean and variance while modeling the joint distribution of the innovations using a parametric copula with unspecified marginal distributions. It is sufficiently general to accommodate widely used dynamic specifications, including ARCH, GARCH, VAR, and Markov-switching models, as well as various copula families that capture nonlinear, asymmetric, and tail dependence. This flexibility makes SCOMDY models applicable to a broad range of financial problems, such as portfolio selection, asset and option pricing, Value-at-Risk measurement, and forecasting \citep{diks2010out,min2014scomdy}.

Despite its flexibility and broad applicability, a primary concern regarding the SCOMDY model is that the innovation copula is assumed to be static. In practice, however, the dependence structure among economic and financial variables may change in response to the evolving economic conditions, including financial crises, pandemic shocks, and major policy or institutional changes \citep{rodriguez2007measuring,bartram2007euro,benkraiem2022financial}. It is therefore crucial to test for structural changes in the innovation copula over the sample period before imposing a time-invariant specification. Otherwise, an erroneously assumed time-invariant copula model may lead to biased parameter estimates and invalid empirical conclusions. \citep{chollete2009modeling,supper2020comparison}.

Testing for copula constancy has attracted considerable attention in recent decades. For example, \citet{nasri2022change} propose Kolmogorov--Smirnov (KS)-type and Cramér--von Mises (CvM)-type test statistics based on the empirical copula process, whereas \citet{lu2025adaptive} develop a kernel-based testing procedure. Both approaches are omnibus against fixed alternatives. However, the asymptotic null distributions of the tests proposed by \citet{nasri2022change} are non-pivotal, requiring computationally intensive bootstrap procedures to obtain critical values. By contrast, the test statistic proposed by \citet{lu2025adaptive} is asymptotically standard normal under the null. Nevertheless, the use of the kernel smoothing method and the associated bandwidth leads to a slower convergence rate of the test, thereby restricting the local alternatives against which the test has nontrivial power. Furthermore, both approaches may suffer power losses when the underlying copula evolves smoothly rather than changes abruptly. For related tests, see also \citet{da2004change}, \citet{busetti2011copula}, \citet{bucher2013consistent}, \citet{dehling2017testing}, and \citet{stark2022testing}.

Unlike existing approaches, we develop a smooth testing approach for structural changes in the innovation copula of SCOMDY models, in the spirit of \citet{neyman1937smooth}. The proposed tests are constructed by jointly testing whether selected generalized Fourier coefficients, which characterize departures from copula constancy, are zero. Our framework accommodates both directional and asymptotically omnibus tests. When a fixed number of Fourier coefficients are included, the resulting test is directional, with power directed toward alternatives captured by the selected basis terms. On the other hand, when the truncation orders are selected using a data-driven procedure over candidate sets whose upper bounds diverge as the sample size increases, the resulting test is asymptotically omnibus. Under the null hypothesis of copula constancy, the test statistics have asymptotic $\chi^2$ distributions and can detect local alternatives converging to the null at the parametric rate. Their critical values can therefore be obtained directly without requiring bootstrap procedures. Thus, unlike \citet{nasri2022change} and \citet{lu2025adaptive}, our proposed tests achieve asymptotic pivotality without sacrificing the parametric detection rate against local alternatives. Furthermore, under the null hypothesis, estimating the dynamic parameters and the unknown marginal distributions does not introduce any first-order estimation effect to our test statistics. Consequently, the feasible tests constructed from estimated residuals are asymptotically equivalent to their infeasible counterparts based on the latent innovations. Finally, our simulation results indicate that the proposed tests deliver substantial finite-sample power gains over existing procedures, particularly when the underlying copula changes smoothly over time. Notably, the data-driven test exhibits high sensitivity and robustness to sparse alternatives, in which changes in the copula are confined to only a subset of components. Taken together, this paper greatly extends and complements existing work on testing copula constancy.

The rest of this paper is organized as follows. In \autoref{Testing framework}, we present our smooth testing framework and the characterization on which it is based. \autoref{Asymptotic theory} systematically investigates the asymptotic theory of our proposed methods. \autoref{Data-driven selections} presents a data-driven selection procedure to automatically determine the truncation orders of the proposed smooth tests. We verify our theoretical results through Monte Carlo simulations in \autoref{Simulations} and provide two empirical applications to illustrate the usefulness of our tests in \autoref{Empirical applications}. \autoref{Concluding comments} concludes the paper. Proofs of theoretical results are presented in the Online Supplementary Material.

Throughout the paper, $\stackrel{d}{\longrightarrow}$ and $\stackrel{p}{\longrightarrow}$ denote convergence in distribution and convergence in probability, respectively. We reserve $\rightsquigarrow$ for weak convergence of random elements in the Skorokhod space $D([0,1],\mathbb{R}^p)$ endowed with the $J_1$ topology, where the dimension $p$ may vary from line to line. In addition, $\rightarrow$ is employed to denote ordinary convergence of deterministic
sequences. For any matrix $A$, we use $\|A\|=\sqrt{\operatorname{tr}(AA')}$ to denote its Frobenius norm. For any bounded function $r:[0,1]\rightarrow\mathbb{R}^p$, $p\geqslant 1$, we define its supremum norm by $\|r\|_{\infty}=\sup_{s\in[0,1]}\|r(s)\|$. For a real symmetric matrix $A$, we use $\lambda_{\min}(A)$ and $\lambda_{\max}(A)$ to denote its smallest and largest eigenvalues, respectively.

\section{Testing framework}
\label{Testing framework}

Let $\{(Y_t',X_t')'\}_{t=1}^n$ be a vector stochastic process, where $Y_t:=(Y_{t1},Y_{t2},\ldots,Y_{td})'$ is of dimension $d\geqslant 2$, and $X_t$ is a vector of predetermined or exogenous variables distinct from $Y_t$'s. Let $\mathcal{F}_{t-1}$ denote the $\sigma$-field generated by $\{Y_{t-1},Y_{t-2},\ldots;X_t,X_{t-1},\ldots\}$. Following \citet{chen2006estimation}, in this paper, we consider the following SCOMDY model:
\begin{align}
Y_t=\mu_t(\alpha_0)+V^{1/2}_t(\alpha_0)\varepsilon_t,
\label{SCOMDY}
\end{align}
where $\mu_t(\alpha_0)=(\mu_{t1}(\alpha_0),\mu_{t2}(\alpha_0),\ldots,\mu_{td}(\alpha_0))':=\mathbb{E}[Y_t|\mathcal{F}_{t-1}]$ is the conditional mean of $Y_t$ given $\mathcal{F}_{t-1}$, and $V_t(\alpha_0)=\text{diag}(V_{t1}(\alpha_0),V_{t2}(\alpha_0),\ldots,V_{td}(\alpha_0))$, in which $V_{tj}(\alpha_0)=\mathbb{E}[(Y_{tj} - \mu_{tj}(\alpha_0))^2|\mathcal{F}_{t-1}]$ denotes the conditional variance of $Y_{tj}$ given $\mathcal{F}_{t-1}$. The standardized innovation vectors $\{\varepsilon_t:=(\varepsilon_{t1},\varepsilon_{t2},\ldots,\varepsilon_{td})'\}_{t=1}^n$ are independent across $t$, and each $\varepsilon_t$ is independent of $\mathcal{F}_{t-1}$. For $j=1,2,\ldots,d$, $\{\varepsilon_{tj}\}_{t=1}^n$ are independent and identically distributed (i.i.d.) with distribution function $F_j$ and satisfy $\mathbb{E}(\varepsilon_{tj}) = 0$, $\mathbb{E}(\varepsilon_{tj}^2)=1$. The distribution function of $\varepsilon_t$ is given by $F_t(x) = C_t(F_1(x_1),F_2(x_2),\ldots,F_d(x_d))$, where $x = (x_1,x_2,\ldots,x_d)'\in\mathbb{R}^d$ and $C_t:[0,1]^d \rightarrow [0,1]$ is the copula function that maps the marginal distributions $F_1(x_1),F_2(x_2),\ldots,F_d(x_d)$ to the joint distribution $F_t(x)$. 

For analytical purposes, we adopt the rescaled-time representation $C_t(\cdot)
=
C(\cdot;t/n)$, $t=1,2,\ldots,n$, where $\{C(\cdot;\tau):\tau\in[0,1]\}$ is a copula family indexed by the rescaled time $\tau$. This article aims to test whether the unknown copula function $C(\cdot;\tau)$ exhibits structural changes. Define the time-averaged copula as $\bar{C}(\cdot) = \int_0^1C(\cdot;\tau)d\tau$. The null hypothesis of no structural change is
\begin{align}
\mathcal{H}_0:C(\cdot;\tau)=\bar{C}(\cdot)\,\,\,\text{for all}\,\,\,\tau\in[0,1],
\label{oriH_0}
\end{align}
whereas the alternative hypothesis is
\begin{align}
\mathcal{H}_1:\exists \,\mathcal{T}\subset [0,1]\,\,\,\text{with positive measure, such that}\,\,\, C(\cdot;\tau)\neq\bar{C}(\cdot)\,\,\, \text{$\forall\tau\in\mathcal{T}$.}
\label{oriH_1}
\end{align}

For $(u,\tau)\in[0,1]^d\times[0,1]$, define the function $\tilde\Delta(u;\tau) = C(u;\tau) - \bar{C}(u)$. The null hypothesis \eqref{oriH_0} is then equivalent to testing whether
\begin{align}
\tilde Q := \int_0^1\int_{[0,1]^d}\tilde \Delta^2(u;\tau)dud\tau=0.
\label{quantity tilde Q}
\end{align}
In the existing literature, \citet{lu2025adaptive} use the criterion in \eqref{quantity tilde Q} to construct a CvM-type test. Their procedure, however, requires estimating $C(\cdot;\tau)$ pointwise in $\tau$ using local smoothing methods. Consequently, the procedure depends critically on bandwidth selection, and the associated test statistic converges at a rate slower than the parametric rate $n^{-1/2}$. Furthermore, the choice of bandwidth also restricts the class of local alternatives that the test can detect.

In this paper, we adopt a criterion similar to that of \citet{lu2025adaptive} but use a completely different approach to test the null hypothesis in \eqref{oriH_0}. Suppose $C(u;\tau)$ is $d$-times differentiable with respect to $u$ such that $c(u;\tau) : = \partial^d C(u;\tau)/\partial u_1\partial u_2\ldots \partial u_d$ and $\bar{c}(u):= \partial^d \bar{C}(u)/\partial u_1\partial u_2\ldots \partial u_d$ are the copula densities of $C(u;\tau)$ and $\bar{C}(u)$, respectively. Define $\Delta(u;\tau) = c(u;\tau) - \bar{c}(u)$. The null hypothesis \eqref{oriH_0} also amounts to testing whether 
\begin{align}
Q := \int_0^1\int_{[0,1]^d} \Delta^2(u;\tau)dud\tau = 0.
\label{quantity Q}
\end{align}
Let $\{\pi_k(u)\}_{k=0}^\infty$ and $\{\psi_l(\tau)\}_{l=0}^\infty$ denote the complete orthonormal systems in function spaces $\mathcal{L}_2([0,1]^d)=\{f:\int_{[0,1]^d}|f(u)|^2du<\infty\}$ and $\mathcal{L}_2([0,1])=\{f:\int_0^1|f(\tau)|^2d\tau<\infty\}$, respectively, with $\pi_0(u)=\psi_0(\tau)\equiv 1$ and $\int_{[0,1]^d}\pi_i(u)\pi_j(u)du = \int_0^1\psi_i(\tau)\psi_j(\tau)d\tau = \delta_{ij}$, where $\delta_{ij}$ is the Kronecker delta that equals 1 if $i=j$ and 0 otherwise. Note that if $\Delta\in\mathcal{L}_2([0,1]^d\times [0,1])$, Parseval’s relation implies that
\begin{align*}
Q = \sum_{k=0}^\infty\sum_{l=0}^\infty \theta_{kl}^2,\,\,\, \text{where}\,\,\, \theta_{kl}:= \int_0^1\int_{[0,1]^d}\Delta(u;\tau)\pi_k(u)\psi_l(\tau)dud\tau.
\end{align*}
As such, testing \eqref{quantity Q} amounts to testing whether the Fourier coefficients $\theta_{kl}=0$ for $k,l\geqslant 0$. In addition, since
\begin{align*}
&\theta_{k0} = \int_0^1\int_{[0,1]^d}(c(u;\tau) - \bar{c}(u))\pi_k(u)dud\tau = \int_{[0,1]^d}(\bar{c}(u) - \bar{c}(u))\pi_k(u)du=0,\\
&\theta_{0l}=\int_0^1\int_{[0,1]^d}(c(u;\tau) - \bar{c}(u))\psi_l(\tau)dud\tau = \int_0^1(1 - 1)\psi_l(\tau)d\tau=0,
\end{align*}
and for $k,l\geqslant 1$, 
\begin{align*}
\int_0^1\int_{[0,1]^d}\bar{c}(u)\pi_k(u)\psi_l(\tau)dud\tau = \int_0^1\psi_l(\tau)d\tau\cdot \int_{[0,1]^d}\bar{c}(u)\pi_k(u)du=0,
\end{align*}
the expansion of $Q$ can be further simplified as 
\begin{align}
Q = \sum_{k=1}^\infty\sum_{l=1}^\infty\theta_{kl}^2,\,\,\,\text{where}\,\,\, \theta_{kl}= \int_0^1\int_{[0,1]^d}c(u;\tau)\pi_k(u)\psi_l(\tau)dud\tau.
\label{expansion Q}
\end{align}
Moreover, when $\Delta\notin\mathcal{L}_2([0,1]^d\times[0,1])$, we may interpret $Q$ as $Q=\infty$. Nevertheless, the Fourier coefficients $\theta_{kl}$ remain well defined and finite, provided that $\pi_k(u)$ and $\psi_l(\tau)$ are bounded on $[0,1]^d$ and $[0,1]$, respectively. 

A traditional CvM-type test can be viewed as jointly testing all coefficients $\theta_{kl}$ for $k,l\geqslant1$. However, such a procedure may suffer a loss of power because it aggregates all possible directions, even though not every coefficient $\theta_{kl}$ contains useful information about departures from the null. In sharp contrast, we propose to construct our test using coefficients that capture the main directions of departure. Note that orthonormal systems are typically constructed so that basis terms with smaller indices capture the dominant features, whereas those with larger indices correspond to the noise. We therefore consider the coefficients $\theta_{kl}$ for $k=1,2,\ldots,K$ and $l=1,2,\ldots,L$, where $K$ and $L$ may either be fixed or selected using data-driven methods. Specifically, when $K$ and $L$ are fixed, we test \eqref{oriH_0} against alternatives exhibiting deviations in certain prespecified directions. Alternatively, we can select $K$ and $L$ using data-driven methods, in which the upper bounds of the candidate sets diverge to infinity as $n\rightarrow \infty$. Thus, the data-driven methods ensure that our tests are asymptotically consistent against any fixed alternatives. In the following, we will first consider the case of fixed $K$ and $L$ in \autoref{Asymptotic theory}, and then analyze the data-driven methods in \autoref{Data-driven selections}.

Suppose that a $\sqrt{n}$-consistent estimator $\hat\alpha_n$ for $\alpha_0$ is available (e.g., using quasi-maximum likelihood estimation). We estimate each $\varepsilon_{tj}$ by $\hat{\varepsilon}_{tj} = V_{tj}^{-1/2}(\hat\alpha_n)(Y_{tj} - \mu_{tj}(\hat\alpha_n))$, $t=1,2,\ldots,n$, $j=1,2,\ldots,d$. Let $U_t = (U_{t1},U_{t2},\ldots,U_{td})'$ with $U_{tj} = F_j(\varepsilon_{tj})$. We then estimate each $U_{tj}$ by $\hat{U}_{tj} = \hat{F}_{nj}(\hat{\varepsilon}_{tj})$, where 
\begin{align*}
\hat{F}_{nj}(x_j) = \frac{1}{n}\sum_{t=1}^n1 (\hat{\varepsilon}_{tj}\leqslant x_j)
\end{align*}
is the empirical distribution function of $F_j$ and $1(\cdot)$ denotes the indicator function. Under some mild regularity conditions, it can be shown that
\begin{align*}
\mathbb{E}\left(\frac{1}{n}\sum_{t=1}^n\pi_k(U_t)\psi_l(t/n)\right)= \frac{1}{n}\sum_{t=1}^n\int_{[0,1]^d}c(u;t/n)\pi_k(u)\psi_l(t/n)du \rightarrow  \theta_{kl}.
\end{align*}
Thus, we propose to estimate each $\theta_{kl}$ via
\begin{align*}
\hat{\theta}_{kl} = \frac{1}{n}\sum_{t=1}^n \pi_k(\hat{U}_t)\psi_l(t/n),
\end{align*}
where $\hat{U}_t=(\hat{U}_{t1},\hat{U}_{t2},\ldots,\hat{U}_{td})'$. For notational brevity, let $\pi(u) = (\pi_1(u),\pi_2(u),\ldots,\pi_K(u))'$ and  $\hat{\theta}_{n,l} = n^{-1}\sum_{t=1}^n\pi(\hat{U}_t)\psi_l(t/n)$. We construct the building block of our test as $\hat{\theta}_n(K,L) = (\hat\theta_{n,1}',\hat\theta_{n,2}',\ldots,\hat\theta_{n,L}')'$. One may wonder whether replacing the unobservable $U_t$ with $\hat{U}_t$ introduces additional estimation uncertainty, causing the asymptotic behavior of $\hat{\theta}_{n}(K,L)$ to differ from that of its infeasible counterpart. Interestingly, we show that under the null hypothesis, neither the estimation of the dynamic parameters nor that of the unknown marginal distributions introduces a first-order estimation effect. In other words,
\begin{align*}
\hat{\theta}_{kl} = \frac{1}{n}\sum_{t=1}^n \pi_k(U_t)\psi_l(t/n) + o_p(n^{-1/2}).
\end{align*}
Such an oracle equivalence is a highly desirable property in econometrics, as it obviates the need for first-order corrections arising from nuisance estimation. By contrast, the weak limit of the empirical copula process generally contains additional first-order terms induced by estimating the unknown marginals; see Lemma~1(iii) of \citet{nasri2022change} for example. In the next section, we will formally establish the asymptotic properties of $\hat{\theta}_n(K,L)$, and construct the final test statistic based on it.

\section{Asymptotic theory}
\label{Asymptotic theory}
In this section, we establish the asymptotic properties of $\hat{\theta}_n(K,L)$ under the null hypothesis $\mathcal{H}_0$, the fixed alternatives $\mathcal{H}_1$, and a sequence of local alternatives $\mathcal{H}_{1L}$ that converges to $\mathcal{H}_0$ at the parametric rate $n^{-1/2}$. We construct our final test statistic based on $\hat{\theta}_n(K,L)$ and demonstrate that it is asymptotically $\chi^2$ distributed under the null hypothesis. Throughout this section, $K$ and $L$ are assumed to be fixed positive integers.

\subsection{Asymptotic null distribution}
\label{Asymptotic null distribution}

Let $\mathcal{A}\subset \mathbb{R}^{d_\alpha}$ denote the parameter space of the SCOMDY model. Define quantities
\begin{align*}
a_{tj} = V_{tj}^{-1/2}(\alpha_0)\frac{\partial \mu_{tj}}{\partial \alpha}(\alpha_0), \,\,\, b_{tj} = \big(2V_{tj}(\alpha_0)\big)^{-1}\frac{\partial V_{tj}}{\partial \alpha}(\alpha_0),\,\,\,\bar{a}_{nj} = \frac{1}{n}\sum_{t=1}^na_{tj},\,\,\,\bar{b}_{nj} = \frac{1}{n}\sum_{t=1}^nb_{tj}.
\end{align*} 
We make the following assumptions.

\begin{assumption}
\label{A1}
The process $\{(Y_t',X_t')'\}_{t=1}^n$ satisfies the SCOMDY specification in \eqref{SCOMDY}. For every $\alpha\in\mathcal{A}$, $\mu_t(\alpha)$ and $V_t(\alpha)$ are $\mathcal{F}_{t-1}$-measurable, and $V_{tj}(\alpha)>0$ almost surely for all $t=1,2,\ldots,n$ and $j=1,2,\ldots,d$.
\end{assumption}

\begin{assumption}
\label{A2}
The parameter space $\mathcal{A}$ is compact, and the true parameter $\alpha_0$ lies in the interior of $\mathcal{A}$. Moreover, $\sqrt{n}(\hat{\alpha}_n-\alpha_0)=O_p(1)$.
\end{assumption}

\begin{assumption}
\label{A3}
The innovation vectors $\{\varepsilon_t\}_{t=1}^n$ are independent across $t$, and each $\varepsilon_t$ is independent of $\mathcal{F}_{t-1}$. For each $j=1,2,\ldots,d$, \textup{(i)} $\{\varepsilon_{tj}\}_{t=1}^n$ are i.i.d. with distribution function $F_j$ and satisfy $\mathbb{E}(\varepsilon_{tj})=0$ and $\mathbb{E}(\varepsilon_{tj}^2)=1$; \textup{(ii)} There exists a fixed neighborhood $\mathcal{A}_0\subset\operatorname{int}(\mathcal{A})$ of $\alpha_0$ such that $m_{tj}(\alpha):=V_{tj}^{-1/2}(\alpha_0)\big(\mu_{tj}(\alpha)-\mu_{tj}(\alpha_0)\big)$ and $s_{tj}(\alpha):=\big(V_{tj}(\alpha)/V_{tj}(\alpha_0)\big)^{1/2}$ are twice continuously differentiable with respect to $\alpha$ in $\mathcal{A}_0$ almost surely; \textup{(iii)} For some $\delta>0$, $G_{1,tj}:=\sup_{\alpha\in\mathcal{A}_0}\left(\|\partial m_{tj}(\alpha)/\partial \alpha\|+\|\partial s_{tj}(\alpha)/\partial \alpha\|\right)$ and $G_{2,tj}:=\sup_{\alpha\in\mathcal{A}_0}\left(\|\partial^2m_{tj}(\alpha)/\partial \alpha\partial \alpha'\|+\|\partial^2s_{tj}(\alpha)/\partial \alpha\partial \alpha'\|\right)$ satisfy $\sup_{n\geqslant1}\max_{1\leqslant t\leqslant n}\mathbb{E}\left(G_{1,tj}^{2+\delta}+G_{2,tj}^{2+\delta}\right)<\infty$; \textup{(iv)}  $F_j$ admits a continuously differentiable density $f_j$, and $M_{0j}:=\sup_{x_j\in\mathbb{R}}(1+|x_j|)f_j(x_j)<\infty$, $M_{1j}:=\sup_{x_j\in\mathbb{R}}(1+|x_j|)^2|\dot f_j(x_j)|<\infty$. \textup{(v)} There exist sequences $\{a_j\}_{j=1}^d$ and $\{b_j\}_{j=1}^d$ such that $\bar{a}_{nj}\stackrel{p}\longrightarrow a_j$ and $\bar{b}_{nj}\stackrel{p}\longrightarrow b_j$.
\end{assumption}

\begin{assumption}
\label{A4}
For every $\tau\in[0,1]$, the copula $C(\cdot;\tau)$ is absolutely continuous on $[0,1]^d$ with copula density $c(u;\tau):=\partial^dC(u;\tau)/\partial u_1\partial u_2\cdots\partial u_d$. Moreover, 
the map $\tau\mapsto c(\cdot;\tau)$ is continuous except at finitely many points, at each of which finite left- and right-hand limits exist.
\end{assumption}

\begin{assumption}
\label{A5}
$\{\pi_k\}_{k=0}^\infty$ and $\{\psi_l\}_{l=0}^\infty$ are complete orthonormal bases of $\mathcal{L}_2([0,1]^d)$ and $\mathcal{L}_2([0,1])$, respectively, with $\pi_0(u)=\psi_0(\tau)\equiv1$ and $\int_{[0,1]^d}\pi_i(u)\pi_{j}(u)du=\int_0^1\psi_i(\tau)\psi_{j}(\tau)d\tau=\delta_{ij}$, the Kronecker delta. For $k=1,2,\ldots,K$ and $l=1,2,\ldots,L$, $\pi_k(u)$ and $\psi_l(\tau)$ are twice differentiable, with $\sup_{u\in[0,1]^d}\left(|\pi_k(u)|+\|\partial \pi_k(u)/\partial u\|+\|\partial^2 \pi_k(u)/\partial u\partial u'\|\right)<\infty$ and $\sup_{\tau\in[0,1]}\left(|\psi_l(\tau)|+|\dot\psi_l(\tau)|+|\ddot\psi_l(\tau)|\right)<\infty$.
\end{assumption}

Assumptions \ref{A1}--\ref{A3} impose regularity conditions on the dynamic specification, parameter estimation, and innovations of the SCOMDY model; see also \citet{chen2006estimation}, \citet{chan2009multivariate}, and \citet{lu2025adaptive}. A key difference from \citet{chen2006estimation} is that we do not impose strict stationarity on the observed process $\{(Y_t',X_t')'\}_{t=1}^n$. Instead, \autoref{A3} places conditions directly on the innovation sequence. It maintains temporal independence and time-invariant marginal distributions while allowing the contemporaneous dependence to vary over time. Such a formulation is therefore sufficiently general to accommodate structural changes in the innovation copula. In \autoref{A4}, the piecewise continuity in $\tau$ accommodates both smooth evolution and finitely many abrupt changes, facilitating the Riemann-sum approximations used in the subsequent analysis. \autoref{A4} also accommodates copulas with unbounded densities, including the Clayton and Gumbel families, whose densities may diverge near the boundary of \([0,1]^d\). \autoref{A5} is a basic assumption about the orthonormal system, which can be found in a large number of studies concerning Neyman's smooth tests; see e.g., \citet{barton1953neyman}, \citet{ledwina1994data}, \citet{bera2002neyman}, \citet{kraus2009adaptive}, and \citet{song2022smooth}. A popular choice for $\pi_k(u)$ is  the product basis functions
\begin{align}
\pi_{\mathbf{k}}(u)
=
\prod_{j=1}^d\varphi_{k_j}(u_j),
\,\,\,
\mathbf{k}
=
(k_1,k_2,\ldots,k_d)'
\in
\mathbb{N}_0^d,
\label{product bases}
\end{align}
where $\{\varphi_k\}_{k=0}^{\infty}$ is a complete orthonormal system of $\mathcal{L}_2([0,1])$ satisfying $\varphi_0(z)\equiv1$ and $\int_0^1\varphi_i(z)\varphi_j(z)dz=\delta_{ij}$. Examples for $\varphi_k(u)$ and $\psi_l(\tau)$ include (i) the Legendre polynomials
\begin{equation}
g_j(z)=(-1)^{j}\sqrt{2j+1}  \sum_{r=0}^{j}C_j^r
C_{j+r}^{r}(-z)^{r},\,\,\, j\in\mathbb{N}_0,
\label{Le}
\end{equation}
where $C_j^r=\frac{j!}{r!(j-r)!}$ and $m!=1\times2\times\cdots\times m$ for any positive integer $m$ with the convention $0!=1$, and (ii) the trigonometric bases
\begin{align}
g_0(z)=1, \,\,\, g_{2j-1}(z)=\sqrt{2}\sin{(2\pi j z)},\,\,\, g_{2j}(z)=\sqrt{2}\cos{(2\pi j z)}, \,\,\, j\geqslant 1.
\label{triangonal}
\end{align}
Additionally, other bases may be considered, such as the Chebyshev and Jacobi polynomials; see \citet{rivlin2020chebyshev} and \citet{szeg1939orthogonal} for detailed discussions.

Let $I_L$ denote the $L$-dimensional identity matrix. Define
\begin{align*}
\mu_K = \int_{[0,1]^d}\pi(u)\bar{c}(u)du\,\,\,\text{and}\,\,\, \Sigma_K = \int_{[0,1]^d}(\pi(u) - \mu_K)(\pi(u) - \mu_K)'\bar{c}(u)du.
\end{align*}
We first establish an important asymptotic expansion of $\hat{\theta}_n(K,L)$ under the null hypothesis. Our result reveals that, under $\mathcal{H}_0$, estimating the dynamic parameters and the unknown marginal distributions does not introduce any first-order estimation effect. Consequently, $\hat{\theta}_n(K,L)$ has the same first-order asymptotic behavior as its infeasible counterpart constructed from the unobservable innovations.

\begin{theorem}
\label{Th1}
Suppose Assumptions \ref{A1}--\ref{A5} hold and $\Sigma_K$ is positive definite. Then, under $\mathcal{H}_0$ in \eqref{oriH_0}, as $n\rightarrow\infty$,
\begin{align}
\hat{\theta}_{kl} = \frac{1}{n}\sum_{t=1}^n\pi_k(U_t)\psi_l(t/n) + o_p(n^{-1/2}).
\label{hat theta kl expansion}
\end{align}
Furthermore, $\sqrt{n}\hat{\theta}_n(K,L)\stackrel{d}\longrightarrow \mathcal{N}_{KL}(0_{KL},I_L\otimes \Sigma_K)$, where $\otimes$ is the Kronecker product.
\end{theorem}

The expansion in \eqref{hat theta kl expansion} establishes an oracle equivalence: under $\mathcal{H}_0$, replacing $U_t$ with $\hat{U}_t$ affects each coefficient only by a term of order $o_p(n^{-1/2})$. Recall that in \autoref{A2}, we allow a wide range of estimators for $\alpha_0$ by requiring only $\sqrt{n}(\hat{\alpha}_n-\alpha_0)=O_p(1)$. Such a rate condition neither requires an asymptotic linear representation nor implies asymptotic normality. For example, some Lasso-type estimators can be $\sqrt{n}$-consistent while having non-Gaussian limiting distributions \citep{knight2000asymptotics,potscher2009distribution}. Hence, our established theory also accommodates such penalized estimators, provided that they are $\sqrt{n}$-consistent. Furthermore, the Kronecker structure $I_L\otimes\Sigma_K$ implies that the coefficient vectors associated with different temporal bases are asymptotically independent, with $\Sigma_K$ capturing the contemporaneous dependence. We require $\Sigma_K$ to be positive definite to ensure that the limiting distribution is nondegenerate, permitting the subsequent construction of a standardized quadratic-form test statistic. 

Inspired by \autoref{Th1}, we derive an asymptotic $\chi^2$ test by first estimating the unknown covariance matrix $\Sigma_K$. Let $\hat{\mu}_{n,K} = n^{-1}\sum_{t=1}^n\pi(\hat{U}_t)$. We propose to estimate $\Sigma_K$ via
\begin{align*}
\hat{\Sigma}_{n,K} = \frac{1}{n}\sum_{t=1}^n(\pi(\hat{U}_t) - \hat{\mu}_{n,K})(\pi(\hat{U}_t) - \hat{\mu}_{n,K})'.
\end{align*}
Then, we standardize $\hat\theta_n(K,L)$ by $\hat{\Pi}_n(K,L):=((\hat\Sigma_{n,K}^{-1/2}\hat{\theta}_{n,1})',(\hat\Sigma_{n,K}^{-1/2}\hat{\theta}_{n,2})',\ldots,(\hat\Sigma_{n,K}^{-1/2}\hat{\theta}_{n,L})')'$. The final test statistic is constructed as the squared  Frobenius norm of $\hat{\Pi}_n(K,L)$:
\begin{align*}
\hat{\Psi}_n^2(K,L):= \|\hat{\Pi}_n(K,L)\|^2 = \sum_{l=1}^L \hat{\theta}_{n,l}'\hat{\Sigma}_{n,K}^{-1}\hat{\theta}_{n,l}.
\end{align*}
The next theorem formally states the asymptotic properties of $\hat{\Psi}_n^2(K,L)$ under the null hypothesis.

\begin{theorem}
\label{Th2}
Suppose Assumptions \ref{A1}--\ref{A5} hold and $\Sigma_K$ is positive definite. Then, under $\mathcal{H}_0$ in \eqref{oriH_0}, as $n\rightarrow\infty$, $n\hat{\Psi}_n^2(K,L)\stackrel{d}\longrightarrow \chi^2_{KL}$.
\end{theorem}

\autoref{Th2} indicates that the $n$-scaled smooth test statistic $n\hat{\Psi}_n^2(K,L)$ is asymptotically $ \chi^2$-distributed under the null hypothesis of copula constancy, with critical values readily available. For a given significance level $\alpha\in(0,1/2)$, we reject $\mathcal{H}_0$ whenever $n\hat \Psi_{n}^2(K,L)>\chi_{KL, {1-\alpha}}^{2}$, where $\chi_{s, {1-\alpha}}^{2}$ is the $(1-\alpha)$-percentile of the $\chi_s^{2}$ distribution.

\subsection{Asymptotic power}
\label{asymptotic power}

Next, we investigate the asymptotic power properties of $\hat{\Psi}_n^2(K,L)$. We first consider fixed alternatives, followed by a sequence of local alternatives that converges to the null at the parametric rate $n^{-1/2}$.

\subsubsection{Power against fixed alternatives}
\label{Power against fixed alternatives}

We first analyze the asymptotic properties of our test under the following fixed alternatives:
\begin{align}
\begin{split}
\mathcal{H}_{1a}:&\ C(\cdot;\tau)\neq\bar{C}(\cdot)\,\,\, \text{for some}\,\,\,\tau\in[0,1],\\
&\ \text{and there exist $1\leqslant k\leqslant K$ and $1\leqslant l\leqslant L$, such that $\theta_{kl}\neq 0$.}
\end{split}
\label{adjH_1}
\end{align}
Note that in \eqref{adjH_1}, the additional condition on $\theta_{kl}$ requires that the departure have a nonzero projection onto at least one of the considered directions. Thus, $\mathcal{H}_{1a}$ consists of fixed alternatives that are detectable under the given testing orders $K$ and $L$.

Define $\theta_l := (\theta_{1l},\theta_{2l},\ldots,\theta_{Kl})'$, $\Pi(K,L):=((\Sigma_K^{-1/2}\theta_1)',(\Sigma_K^{-1/2}\theta_2)',\ldots,(\Sigma_K^{-1/2}\theta_L)')'$, and 
\begin{align*}
\Psi^2(K,L):=\left\|\Pi(K,L)\right\|^2=\sum_{l=1}^L\theta_l'\Sigma_K^{-1}\theta_l.
\end{align*}
The following theorem shows that our proposed test $\hat{\Psi}_n^2(K,L)$ is consistent against $\mathcal{H}_{1a}$.

\begin{theorem}
\label{Th3}
Suppose Assumptions \ref{A1}--\ref{A5} hold and $\Sigma_K$ is positive definite. Then, under $\mathcal{H}_{1a}$ in \eqref{adjH_1}, as $n\rightarrow\infty$, $\hat{\Psi}_n^2(K,L)\stackrel{p}\longrightarrow \Psi^2(K,L)>0$.
\end{theorem}

\autoref{Th3} implies that, under $\mathcal{H}_{1a}$, $n\hat{\Psi}_n^2(K,L)$ diverges to infinity in probability. In particular, we have
\begin{align*}
\mathbb{P}\left\{n\hat{\Psi}_n^2(K,L)>\chi^2_{KL,1-\alpha}\right\}\rightarrow1,\,\,\,n\rightarrow\infty.
\end{align*}
Hence, the proposed test is consistent against all fixed alternatives in $\mathcal{H}_{1a}$. 

It is worth noting that \eqref{adjH_1} does not encompass all possible deviations in \eqref{oriH_1}. Specifically, \eqref{oriH_1} may hold even when $\theta_{kl}=0$ for all $1\leqslant k\leqslant K$ and $1\leqslant l\leqslant L$. Thus, with fixed $K$ and $L$, our proposed test is directional rather than omnibus. Nevertheless, we do not view this limitation as a severe drawback. The reason is that smooth tests are generally designed to improve power by concentrating on directions with the most pronounced departures, while excluding directions that carry little signal. This trade-off often enables smooth tests to achieve higher power than conventional omnibus tests against a broad range of alternatives, although it also means that they are no longer consistent against every fixed alternative in $\mathcal{H}_1$. Following the existing literature, a theoretically appealing way to achieve the omnibus property is to select the testing orders $K$ and $L$ using data-driven procedures \citep{inglot1997data,kallenberg1997data,kallenberg1999data,janic2000data}. We leave this issue to \autoref{Data-driven selections} and focus primarily on the case of fixed $K$ and $L$ in this section.

\subsubsection{Power against local alternatives}
\label{Power against local alternatives}

Now, we consider the power performance of our proposed test against the following sequence of local alternatives:
\begin{align}
\mathcal{H}_{1L}: c_n(u;\tau)=\bar c(u)+n^{-1/2}r(u;\tau),
\label{H_1L}
\end{align}
where $r:[0,1]^d\times [0,1]\rightarrow \mathbb{R}$ satisfies the following assumption.

\begin{assumption}
\label{A6}
(i) There exists a function $\zeta:[0,1]^d\times[0,1]\rightarrow\mathbb{R}$ such that $r(u;\tau)=\bar c(u)\zeta(u;\tau)$, where $\zeta(u;\tau)$ is bounded on $[0,1]^d\times [0,1]$ and $\int_0^1\int_{[0,1]^d}\zeta^2(u;\tau)\bar c(u)dud\tau>0$. (ii) The map $\tau\mapsto\zeta(\cdot;\tau)$ is continuous except at finitely many points,  at each of which finite left-
and right-hand limits exist. (iii) $\int_0^1r(u;\tau)d\tau=0$ 
for all $u\in[0,1]^d$, and for every $j=1,2,\ldots,d$, $\int_{[0,1]^{d-1}}r(u;\tau)du_{-j}=0$ for all $(u_j,\tau)\in[0,1]\times[0,1]$, where $du_{-j}$ denotes integration with respect to all components of $u$ except $u_j$.
\end{assumption}

\autoref{A6} imposes regularity and normalization conditions on the perturbation $r(u;\tau)$. In \autoref{A6}(i), the boundedness of $\zeta(u;\tau)$ guarantees the nonnegativity of $c_n(u;\tau)$ as long as $n$ is large enough and facilitates establishing contiguity between probability measures under $\mathcal{H}_0$ and $\mathcal{H}_{1L}$ \citep[p.88, Lemma~6.4]{Vaart_1998}. \autoref{A6}(ii) permits finitely many discontinuities in the temporal evolution of the local perturbation, while providing sufficient regularity for the relevant Riemann-sum approximations. \autoref{A6}(iii) imposes some necessary centering and marginal restrictions. Specifically, the first restriction ensures that $\int_0^1c_n(u;\tau)d\tau=\bar c(u)$, whereas the second preserves the uniform marginal density of each component $U_{tj}$. Taken together, these restrictions ensure that $c_n(u;\tau)$ is a valid copula density for every $\tau\in[0,1]$ and all sufficiently large $n$.

For $l=1,2,\ldots,L$, define $\Delta_l = (\Delta_{1l},\Delta_{2l},\ldots,\Delta_{Kl})'$ with $\Delta_{kl}$ given by
\begin{align*}
\Delta_{kl} = \int_0^1\int_{[0,1]^d}r(u;\tau)\pi_k(u)\psi_l(\tau)dud\tau.
\end{align*}
The next theorem establishes the asymptotic properties of $\hat{\Psi}_n^2(K,L)$ under the sequence of local alternatives in \eqref{H_1L}.

\begin{theorem}
\label{Th4}
Suppose Assumptions \ref{A1}--\ref{A6} hold and $\Sigma_K$ is positive definite. Then, under $\mathcal{H}_{1L}$ in \eqref{H_1L}, as $n\rightarrow\infty$,
\begin{align*}
n\hat{\Psi}_n^2(K,L)\stackrel{d}\longrightarrow \chi_{KL}^2\left(\sum_{l=1}^L\Delta_l'\Sigma_K^{-1}\Delta_l\right),
\end{align*}
where $\chi_s^2(c)$ denotes the noncentral $\chi^2$ distribution with $s$ degrees of freedom and nonnegative noncentrality parameter $c$.
\end{theorem}

According to \autoref{Th4}, if $\Delta_{kl}\neq0$ for at least one pair $(k,l)$, the noncentrality parameter in the limiting distribution is then strictly positive. Consequently, the test based on $n\hat{\Psi}_n^2(K,L)$ has nontrivial asymptotic power against $\mathcal{H}_{1L}$ in \eqref{H_1L}. It is worth noting that our test can detect local alternatives that approach the null at rate $n^{-1/2}$. In comparison, the detection rate established in \citet{lu2025adaptive} is $n^{-1/2}h^{-1/4}$, where $h=cn^{-b}$ with $c\in(0,\infty)$ and $b\in(0,1/2)$. Since $n^{-1/2}h^{-1/4}=c^{-1/4}n^{-1/2+b/4}$, the rate converges to zero more slowly than $n^{-1/2}$. Therefore, compared with \citet{lu2025adaptive}, our test achieves a faster detection rate and is more sensitive to weak local departures from copula constancy.

\section{Data-driven selections of $K$ and $L$}
\label{Data-driven selections}
In this section, we propose a data-driven selection procedure to automatically determine the truncation orders $K$ and $L$. It is well known that the truncation orders significantly affect the performance of smooth test statistics \citep{ledwina1994data,kallenberg1995consistency,kallenberg1995data,kallenberg1997data,kallenberg1999data,inglot1997data,janic2000data}. On the one hand, large $K$ and $L$ may introduce directions that contribute little but only inflate the degrees of freedom of the test statistic, leading to potential size distortion. On the other hand, if $K$ and $L$ are too small, we risk $\theta_{kl}=0$ for $k=1,2,\ldots,K$ and $l=1,2,\ldots,L$, rendering the test statistic powerless. Therefore, choosing appropriate truncation orders is crucial to the applicability of our proposed smooth test.

In the existing literature, the data-driven selection procedure for Neyman's smooth tests was first proposed by \citet{ledwina1994data} to test uniformity based on Schwarz's selection criterion. The basic idea of this selection rule is to enable Neyman's smooth tests to automatically focus on the ``right'' directions. Desirable theoretical properties and numerical results based on this selection rule are further investigated by \citet{kallenberg1995consistency,kallenberg1995data} and \citet{inglot1996asymptotic}. Later, modified Schwarz's selection rules were proposed to circumvent the computational burden associated with maximizing the log-likelihood in \citet{ledwina1994data}; see also \citet{kallenberg1997data}, \citet{inglot1997data}, \citet{kallenberg1999data}, and \citet{janic2000data}.

For the convenience of implementation, in this section, we propose to use the product basis functions considered in \eqref{product bases}:
\begin{align*}
\pi_{\mathbf{k}}(u)
=
\prod_{j=1}^d\varphi_{k_j}(u_j),
\,\,\,
\mathbf{k}
=
(k_1,k_2,\ldots,k_d)'
\in
\mathbb{N}_0^d,
\end{align*}
where $\{\varphi_k\}_{k=0}^{\infty}$ is a complete orthonormal system of $\mathcal{L}_2([0,1])$ satisfying $\varphi_0(z)\equiv1$ and $\int_0^1\varphi_i(z)\varphi_j(z)dz=\delta_{ij}$. Obviously, $\{\pi_{\mathbf{k}}\}_{\mathbf{k}\in\mathbb{N}_0^d}$ forms a complete orthonormal system of $\mathcal{L}_2([0,1]^d)$. Define
\begin{align*}
\|\mathbf{k}\|_0
=
\sum_{j=1}^d1(k_j>0),
\,\,\,
\|\mathbf{k}\|_1
=
\sum_{j=1}^d k_j.
\end{align*}
We now reconstruct $\pi(u)$ as a vector consisting of all $\pi_{\mathbf{k}}(u)$ satisfying $\|\mathbf{k}\|_0\geqslant2$ and $\|\mathbf{k}\|_1\leqslant P$, where $P\geqslant2$ is a positive integer to be selected. Recall that if $\|\mathbf{k}\|_0\leqslant1$, then
\begin{align*}
\int_0^1\int_{[0,1]^d}
c(u;\tau)\pi_{\mathbf{k}}(u)\psi_l(\tau)
dud\tau
\equiv0,
\end{align*}
so the corresponding component cannot distinguish $\mathcal{H}_1$ from $\mathcal{H}_0$. Hence, the restriction $\|\mathbf{k}\|_0\geqslant2$ is imposed to exclude unnecessary testing components. Note that when the product basis is used, the dimension of $\pi(u)$ depends on $P$ and is given by $K=K(P)=C_{P+d}^d-1-dP$. We propose to select $P$ and $L$ over $2\leqslant P\leqslant D_P(n)$ and $1\leqslant L\leqslant D_L(n)$, respectively, where $D_P(n)$ and $D_L(n)$ diverge slowly as $n\rightarrow\infty$ to ensure that the resulting data-driven test is asymptotically consistent against any fixed alternative.

Following the existing literature, we propose to adopt a modified Schwarz selection rule to determine $(P,L)$. Let
\begin{align*}
\mathcal{M}_n
=
\operatorname*{arg\,max}_{\substack{
2\leqslant P\leqslant D_P(n)\\
1\leqslant L\leqslant D_L(n)
}}
\left\{
n\left\|\hat{\theta}_n(K(P),L)\right\|^2
-
K(P)L\log n
\right\}.
\end{align*}
We define $(\hat{P},\hat{L})$ as the element of $\mathcal{M}_n$ having the smallest value of $K(P)L$. If there is still more than one such element, we select the one with the smallest value of $P$. The data-driven test statistic is then given by $\hat{\Psi}_n^2(K(\hat{P}),\hat{L})$. Define
\begin{align*}
&B_{\varphi,0k}
=
\sup_{z\in[0,1]}|\varphi_k(z)|,
\,\,\,
B_{\varphi,1k}
=
\sup_{z\in[0,1]}|\dot{\varphi}_k(z)|,
\,\,\,
B_{\varphi,2k}
=
\sup_{z\in[0,1]}|\ddot{\varphi}_k(z)|,\\
&B_{\psi,0l}
=
\sup_{\tau\in[0,1]}|\psi_l(\tau)|,\,\,\,B_{\psi,1l}
=
\sup_{\tau\in[0,1]}|\dot\psi_l(\tau)|,\,\,\,B_{\psi,2l}
=
\sup_{\tau\in[0,1]}|\ddot\psi_l(\tau)|.
\end{align*}
To ensure the validity of the selection procedure, we need to impose several additional regularity conditions on the basis functions, $D_P(n)$, and $D_L(n)$, which are summarized below.

\begin{assumption}
\label{AD1}
(i) $\{\varphi_k\}_{k=0}^\infty$ are complete orthonormal bases of $\mathcal{L}_2([0,1])$ with $\varphi_0(z)\equiv1$ and $\int_0^1\varphi_i(z)\varphi_{j}(z)dz=\delta_{ij}$. In addition, each $\varphi_k(z)$ is twice differentiable on $[0,1]$. (ii) For some nonnegative constants $r_{\varphi,i}$ and $r_{\psi,i}$ with $i=0,1,2$, $B_{\varphi,ik}=O(k^{r_{\varphi,i}})$, $B_{\psi,il}=O(l^{r_{\psi,i}})$. (iii) As $n\rightarrow\infty$, $D_P(n)\rightarrow\infty$, $D_L(n)\rightarrow\infty$, and 
\begin{align}
D_P^{\rho_1}(n)D_L^{r_{\psi,0}}(n)
=o(\sqrt{\log n}),\,\,\,
D_P^{2\rho_0}(n)\log\!\left(D_P^d(n)D_L(n)\right)=o(\log n),
\label{rate DP DL}
\end{align}
where $\rho_0 = dr_{\varphi,0}$ and $\rho_1 = r_{\varphi,1} + (d-1)r_{\varphi,0}$.
\end{assumption}

To provide some intuition for the convergence rate in \eqref{rate DP DL}, we set $\varphi_j(z)=\psi_j(z)=g_j(z)$ for illustration, where $\{g_j\}_{j=0}^{\infty}$ may be chosen as either the Legendre polynomials in \eqref{Le} or the trigonometric bases in \eqref{triangonal}. Specifically, for the Legendre polynomials in \eqref{Le}, we have $r_{\varphi,0}=r_{\psi,0}=1/2$, $r_{\varphi,1}=r_{\psi,1}=5/2$, and $r_{\varphi,2}=r_{\psi,2}=9/2$. Then, \eqref{rate DP DL} reduces to $D_P^{d+4}(n)D_L(n)=o(\log n)$. On the other hand, when $\{g_j\}_{j=0}^{\infty}$ is chosen as the trigonometric bases in \eqref{triangonal}, we have $r_{\varphi,0}=r_{\psi,0}=0$, $r_{\varphi,1}=r_{\psi,1}=1$, and $r_{\varphi,2}=r_{\psi,2}=2$. \eqref{rate DP DL} then simplifies to $D_P(n)=o(\sqrt{\log n})$ and $\log D_L(n)=o(\log n)$. Hence, the trigonometric bases are particularly attractive from a theoretical perspective because the rate restrictions on $D_P(n)$ and $D_L(n)$ do not depend on the dimension $d$. This dimension-free property facilitates their application to multivariate settings and provides an important advantage over the Legendre polynomials.

Next, we establish the following theorem, which shows that our data-driven procedure selects the least complex testing model under the null hypothesis while being consistent against any fixed alternative. 

\begin{theorem}
\label{ThD}
Suppose Assumptions \ref{A1}--\ref{A5} and \ref{AD1} hold.

\textup{(i)} Under $\mathcal{H}_0$ in \eqref{oriH_0}, as $n\rightarrow\infty$, $\mathbb{P}\{(\hat{P},\hat{L})=(2,1)\}\rightarrow1$. If in addition $\Sigma_{K(2)}$ is positive definite, then $n\hat{\Psi}_n^2(K(\hat{P}),\hat{L})
\stackrel{d}{\longrightarrow}
\chi_{d(d-1)/2}^2.$

\textup{(ii)} Under $\mathcal{H}_1$ in \eqref{oriH_1}, suppose $$\mathbb{P}\left\{
\min_{2\leqslant P\leqslant D_P(n)}
\lambda_{\min}(\hat{\Sigma}_{n,K(P)})>0\right\}\rightarrow1,\,\,\,n\rightarrow\infty.$$ Then, as $n\rightarrow\infty$, for every $x\in\mathbb{R}$, $\mathbb{P}\{
n\hat{\Psi}_n^2(K(\hat{P}),\hat{L})\leqslant x\}\rightarrow0$.
\end{theorem}

Following \autoref{ThD}(i), for a given significance level $\alpha\in(0,1/2)$, we reject $\mathcal{H}_0$ whenever $n\hat{\Psi}_{n}^2(K(\hat{P}),\hat{L})>\chi_{d(d-1)/2,1-\alpha}^{2}$. In finite samples, however, the event $\{(\hat{P},\hat{L})\neq(2,1)\}$ may still occur with small frequency, causing the null distribution of the data-driven test statistic to deviate from the limiting $\chi^2$ distribution. Such a selection uncertainty often necessitates a finite-sample approximation to the null distribution in conventional data-driven smooth tests \citep{kallenberg1995data,kallenberg1997data,kallenberg1999data,inglot1997data,janic2000data}. Nevertheless, as shown in the next section, for our proposed data-driven test, the empirical frequency of $\{(\hat{P},\hat{L})\neq(2,1)\}$ under $\mathcal{H}_0$ is sufficiently low and does not distort the empirical size. Consequently, $\chi_{d(d-1)/2,1-\alpha}^{2}$ can be used directly as the critical value without any additional finite-sample correction.

\section{Monte Carlo simulations}
\label{Simulations}

In this section, we conduct numerical simulations to verify the theoretical results. We set the sample size $n\in\{200,500,800\}$, and fix the significance level at $\alpha=0.05$. All empirical rejection rates are computed based on 1000 replications. For $j=1,2,\ldots,d$, we take $F_j(\cdot) = \Phi(\cdot)$, the distribution function of the standard normal. For both the directional test $\hat{\Psi}_n^2(K,L)$ and the data-driven test $\hat{\Psi}_n^2(K(\hat{P}),\hat{L})$, we consider the product construction \eqref{product bases} for $\{\pi_k(u)\}_{k=0}^\infty$ and choose the trigonometric bases  \eqref{triangonal} for $\varphi_j(z)$ and $\psi_j(z)$ due to the dimension-free property. 

To examine how the numerical performance of our proposed method differs from other existing tests, we compare our smooth tests with (i) the kernel-based test of \citet{lu2025adaptive}:
\begin{align*}
\hat{M}_n: = \frac{nh^{1/2}\hat{Q} - \hat{B}}{\sqrt{\hat{V}}},
\end{align*}
where 
\begin{align*}
\hat{Q} = \frac{1}{n}\sum_{t=1}^n\int_{[0,1]^d}\left(\hat{C_t}(u) - \hat{C}_0(u)\right)^2W(u)du
\end{align*}
is the main building block, $W(u)$ is a nonnegative weight function,  $\hat{C}_t(u) = \sum_{s=1}^nw_h(s,t)1(\hat{U}_s\leqslant u)$, $\hat{C}_0(u) = n^{-1}\sum_{t=1}^n1(\hat{U}_t\leqslant u)$, $w_h(s,t) = (\sum_{r=1}^nK_h^*(r,t))^{-1}K_h^*(s,t)$, 
\begin{align*}
K_h^*(s,t)= \begin{cases}
K_h(\frac{s-t}{n}) + K_h(\frac{1-s-t}{n}),\,& 1\leqslant t\leqslant \lfloor nh\rfloor,\\
K_h(\frac{s-t}{n}),\, &\lfloor nh \rfloor + 1\leqslant t \leqslant n - \lfloor nh \rfloor,\\
K_h(\frac{s-t}{n}) + K_h(\frac{2n+1-s-t}{n}),\,&n - \lfloor nh \rfloor + 1\leqslant t \leqslant n,
\end{cases}
\end{align*}
$K_h(x) = K(x/h)/h$, with $K(\cdot)$ and $h$ denoting the kernel function and the bandwidth, respectively, and $\hat{B}$ and $\hat{V}$ are the estimates of bias and variance given in (2.13) and (2.14) of \citet{lu2025adaptive}; and (ii) the KS-type and CvM-type tests considered by \citet{nasri2022change}:
\begin{align*}
\mathcal{T}_n = \max_{1\leqslant j\leqslant n}\max_{1\leqslant t\leqslant n}\left|\mathbb{A}_n(j/n,\hat{U}_t)\right|,\,\,\, \mathcal{S}_n = \max_{1\leqslant j\leqslant n}\frac{1}{n}\sum_{t=1}^n\mathbb{A}_n^2(j/n,\hat{U}_t)
\end{align*}
where $\mathbb{A}_n(s,u) = n^{-1/2}\sum_{t=1}^{\lfloor ns \rfloor}(1(\hat{U}_t\leqslant u) - \hat{C}_0(u))$. Following \citet{lu2025adaptive}, we use the Epanechnikov kernel for $K(\cdot)$ and set $W(u)=\prod_{j=1}^d\phi(u_j;0.4,0.1)$, where $\phi(x;\mu,\sigma)$ denotes the density of $\mathcal{N}(\mu,\sigma^2)$. For the bandwidth, \citet{lu2025adaptive} adopt the rule-of-thumb choice $h=n^{-1/5}/\sqrt{12}$. Here, we set $h=cn^{-1/5}/\sqrt{12}$ with $c\in\{0.5,1,2\}$ to examine how the performance of $\hat{M}_n$ varies across different bandwidth choices. The corresponding test statistics are denoted by $\hat{M}_n(c)$. Under some regularity conditions, \citet{lu2025adaptive} show that, as $n\rightarrow\infty$, $\hat{M}_n\stackrel{d}\longrightarrow\mathcal{N}(0,1)$. Thus, we use the one-sided critical value $1.6449$ for the test. The limiting distributions of $\mathcal{T}_n$ and $\mathcal{S}_n$ are non-pivotal, and a bootstrap procedure is needed to determine their critical values. We use $B=200$ bootstrap replications in our simulations.

In the following, we first consider the two-dimensional case in \autoref{Simulation 1}. Then, in \autoref{Simulation 2}, we consider the case with $d=4$, where the temporal changes in the copula may be either dense or sparse. 

\subsection{Simulation 1}
\label{Simulation 1}
We consider a location-scale model and a dynamic model to generate $Y_t$:
\begin{align*}
&\text{LS1:} \,\,\,Y_{t1} = 2 + \sqrt{5}\varepsilon_{t1},\,\,\, Y_{t2} = 1 + \sqrt{3}\varepsilon_{t2};\\
& \text{DY1:}\,\,\, Y_{t1} = 0.3 + 0.4Y_{t-1,1} + 0.2\varepsilon_{t1} + 0.1\varepsilon_{t-1,1},\,\,\, Y_{t2} = 2X_t + \sigma_{t2}\varepsilon_{t2},\\
&\qquad\,\,\,\,\,\,\,\,\, \sigma_{t2}^2 = 0.2 + 0.4\sigma_{t-1,2}^2 + 0.4\sigma_{t-1,2}^2\varepsilon_{t-1,2}^2,\,\,\, X_t\stackrel{i.i.d}\sim U[0,1],\,\,\, X_t\perp \varepsilon_t.
\end{align*}
Let $C_{\text{Ct}}^{(d)}(u;\tau)$ and $C_{\text{Gs}}^{(d)}(u;\tau)$ denote the $d$-dimensional Clayton copula and Gaussian copula with Kendall rank correlation coefficient $\tau$. We consider the following data-generating processes (DGPs) for the innovation copula:
\begin{align*}
\text{DGP 1.1: } C_t(u) =& C_{\text{Ct}}^{(2)}(u;0.7);\\
\text{DGP 1.2: } C_t(u) =& C_{\text{Gs}}^{(2)}(u;0.7);\\
\text{DGP 1.3: } C_t(u) =& C_{\text{Ct}}^{(2)}(u;0.3)1(t/n\leqslant 0.5) + C_{\text{Ct}}^{(2)}(u;0.7)1(t/n> 0.5);\\
\text{DGP 1.4: } C_t(u) =& C_{\text{Ct}}^{(2)}(u;0.3)1(t/n\leqslant 0.4) + C_{\text{Gs}}^{(2)}(u;0.7)1(t/n> 0.4);\\
\text{DGP 1.5: } C_t(u) =& C_{\text{Gs}}^{(2)}(u;0)1(t/n\leqslant 0.3) \\
&+ C_{\text{Gs}}^{(2)}(u;0.3)1(0.3<t/n\leqslant 0.7) + C_{\text{Gs}}^{(2)}(u;0.7)1(t/n> 0.7);\\
\text{DGP 1.6: } C_t(u) =& C_{\text{Ct}}^{(2)}(u;\tau_t)\,\,\,\tau_t = 0.4\exp(-10(t/n - 0.5)^2) + 0.3;\\
\text{DGP 1.7: } C_t(u) =& C_{\text{Gs}}^{(2)}(u;\tau_t)\,\,\,\tau_t = 0.8\sin(4\pi t/n).
\end{align*}
Clearly, DGPs 1.1--1.2 fall under $\mathcal{H}_0$, whereas DGPs 1.3--1.7 fall under $\mathcal{H}_1$. In particular, DGPs 1.3--1.4 involve a single abrupt change; DGP 1.5 involves multiple abrupt changes; DGP 1.6 exhibits a smooth change in the Kendall rank correlation coefficient; and DGP 1.7 displays periodic variation.

For notational simplicity, we use the shorthand
\begin{align*}
\varphi_{k_1}\cdots\varphi_{k_d}:=\prod_{j=1}^d\varphi_{k_j}(u_j).
\end{align*}
Then, for the directional test $\hat{\Psi}_n^2(K,L)$, we take $L\in\{3,4\}$, and consider the following product bases:
\begin{align*}
&\pi_{(1)}(u) = (\varphi_1\varphi_1,\varphi_1\varphi_2,\varphi_2\varphi_1,\varphi_2\varphi_2)', \,\,\,\pi_{(2)}(u) = (\varphi_1\varphi_1,\varphi_2\varphi_2,\varphi_3\varphi_3)'.
\end{align*}
For notational brevity, we use $\hat{\Psi}^2_{n,\pi_{(i)}}(L)$ to denote the directional test associated with $\pi_{(i)}(u)$ and $L$. For the data-driven test, we take $D_P(n) = 6$ and $D_L(n) = 8$. The numerical results are summarized in Tables \ref{DGPs1-1--1-7 LS1}--\ref{DGPs1-1--1-7 DY1}.

We first examine the empirical size. Under DGPs 1.1--1.2, both the directional and data-driven tests perform comparably to the existing tests, with empirical rejection rates close to the nominal level. Turning to empirical power, our proposed tests yield substantial finite-sample power gains over the existing methods under DGPs 1.3--1.7. For example, under DGPs 1.3--1.5 with $n=200$, our proposed tests attain empirical power exceeding 0.85, whereas the empirical power of $\hat{M}_n$, $\mathcal{T}_n$, and $\mathcal{S}_n$ remains below 0.6. In addition, under DGP 1.6, our proposed tests deliver nontrivial empirical power even when $n=200$. In contrast, the empirical power of $\hat{M}_n$, $\mathcal{T}_n$, and $\mathcal{S}_n$ remains close to the nominal level, indicating that these tests have little power against the smooth alternative. Furthermore, the empirical power of $\hat{M}_n$ varies noticeably across different bandwidth choices, highlighting its sensitivity to bandwidth selection in finite samples.

\begin{table}[htbp]
\renewcommand\arraystretch{0.69}
\centering
\caption{Empirical rejection rates under DGPs 1.1–1.7 with model LS1.} 
\resizebox{\linewidth}{!}{
\begin{tabular}{cccccccccccc}
\hline
DGP                     & $n$ & $\hat{\Psi}_{n,\pi_{(1)}}^2(3)$ & $\hat{\Psi}_{n,\pi_{(1)}}^2(4)$ & $\hat{\Psi}_{n,\pi_{(2)}}^2(3)$ & $\hat{\Psi}_{n,\pi_{(2)}}^2(4)$ & $\hat{\Psi}_{n}^2(K(\hat{P}),\hat{L})$ & $\hat{M}_n(0.5)$ & $\hat{M}_n(1)$ & $\hat{M}_n(2)$ & $\mathcal{T}_n$ & $\mathcal{S}_n$ \\ \hline
\multirow[t]{3}{*}{DGP 1.1} & 200 & 0.0440                    & 0.0460                    & 0.0500                    & 0.0580                    & 0.0450                                 & 0.0570            & 0.0560          & 0.0650          & 0.0520          & 0.0510          \\
                        & 500 & 0.0530                    & 0.0460                    & 0.0430                    & 0.0430                    & 0.0700                                 & 0.0760            & 0.0760          & 0.0670          & 0.0460          & 0.0470          \\
                        & 800 & 0.0430                    & 0.0430                    & 0.0490                    & 0.0520                    & 0.0520                                 & 0.0510            & 0.0560          & 0.0570          & 0.0530          & 0.0490          \\ \hline
\multirow[t]{3}{*}{DGP 1.2} & 200 & 0.0550                    & 0.0390                    & 0.0540                    & 0.0460                    & 0.0490                                 & 0.0660            & 0.0550          & 0.0690          & 0.0510          & 0.0530          \\
                        & 500 & 0.0420                    & 0.0400                    & 0.0480                    & 0.0500                    & 0.0480                                 & 0.0650            & 0.0720          & 0.0670          & 0.0510          & 0.0590          \\
                        & 800 & 0.0550                    & 0.0580                    & 0.0640                    & 0.0620                    & 0.0640                                 & 0.0710            & 0.0640          & 0.0770          & 0.0460          & 0.0420          \\ \hline
\multirow[t]{3}{*}{DGP 1.3} & 200 & 0.9370                    & 0.9100                    & 0.9560                    & 0.9330                    & 0.9450                                 & 0.1740            & 0.2100          & 0.2810          & 0.2090          & 0.1970          \\
                        & 500 & 1.0000                    & 1.0000                    & 1.0000                    & 1.0000                    & 1.0000                                 & 0.3300            & 0.4320          & 0.5490          & 0.4970          & 0.4290          \\
                        & 800 & 1.0000                    & 1.0000                    & 1.0000                    & 1.0000                    & 1.0000                                 & 0.5340            & 0.6560          & 0.7690          & 0.7410          & 0.6560          \\ \hline
\multirow[t]{3}{*}{DGP 1.4} & 200 & 0.9520                    & 0.9390                    & 0.9710                    & 0.9690                    & 0.8750                                 & 0.1320            & 0.1800          & 0.2220          & 0.2140          & 0.2060          \\
                        & 500 & 1.0000                    & 1.0000                    & 1.0000                    & 1.0000                    & 0.9980                                 & 0.2830            & 0.3580          & 0.4620          & 0.5260          & 0.4800          \\
                        & 800 & 1.0000                    & 1.0000                    & 1.0000                    & 1.0000                    & 1.0000                                 & 0.4360            & 0.5580          & 0.6480          & 0.7120          & 0.6580          \\ \hline
\multirow[t]{3}{*}{DGP 1.5} & 200 & 0.9840                    & 0.9760                    & 0.9900                    & 0.9830                    & 0.9060                                 & 0.3590            & 0.4570          & 0.5470          & 0.3350          & 0.3630          \\
                        & 500 & 1.0000                    & 1.0000                    & 1.0000                    & 1.0000                    & 1.0000                                 & 0.6980            & 0.7980          & 0.8760          & 0.7010          & 0.7320          \\
                        & 800 & 1.0000                    & 1.0000                    & 1.0000                    & 1.0000                    & 1.0000                                 & 0.9060            & 0.9680          & 0.9810          & 0.9100          & 0.9160          \\ \hline
\multirow[t]{3}{*}{DGP 1.6} & 200 & 0.6260                    & 0.5870                    & 0.6920                    & 0.6400                    & 0.5000                                 & 0.0990            & 0.1230          & 0.1320          & 0.0680          & 0.0590          \\
                        & 500 & 0.9920                    & 0.9850                    & 0.9970                    & 0.9920                    & 0.8990                                 & 0.1430            & 0.1990          & 0.2530          & 0.0760          & 0.0900          \\
                        & 800 & 1.0000                    & 1.0000                    & 1.0000                    & 1.0000                    & 0.9850                                 & 0.2240            & 0.2950          & 0.3690          & 0.0820          & 0.0980          \\ \hline
\multirow[t]{3}{*}{DGP 1.7} & 200 & 1.0000                    & 1.0000                    & 1.0000                    & 1.0000                    & 1.0000                                 & 0.9910            & 0.9980          & 0.9690          & 0.2600          & 0.4650          \\
                        & 500 & 1.0000                    & 1.0000                    & 1.0000                    & 1.0000                    & 1.0000                                 & 1.0000            & 1.0000          & 1.0000          & 0.7770          & 0.9600          \\
                        & 800 & 1.0000                    & 1.0000                    & 1.0000                    & 1.0000                    & 1.0000                                 & 1.0000            & 1.0000          & 1.0000          & 0.9800          & 1.0000          \\ \hline
\end{tabular}}
\label{DGPs1-1--1-7 LS1}
\end{table}

\begin{table}[htbp]
\renewcommand\arraystretch{0.69}
\centering
\caption{Empirical rejection rates under DGPs 1.1–1.7 with model DY1.} 
\resizebox{\linewidth}{!}{
\begin{tabular}{cccccccccccc}
\hline
DGP                     & $n$ & $\hat{\Psi}_{n,\pi_{(1)}}^2(3)$ & $\hat{\Psi}_{n,\pi_{(1)}}^2(4)$ & $\hat{\Psi}_{n,\pi_{(2)}}^2(3)$ & $\hat{\Psi}_{n,\pi_{(2)}}^2(4)$ & $\hat{\Psi}_{n}^2(K(\hat{P}),\hat{L})$ & $\hat{M}_n(0.5)$ & $\hat{M}_n(1)$ & $\hat{M}_n(2)$ & $\mathcal{T}_n$ & $\mathcal{S}_n$ \\ \hline
\multirow[t]{3}{*}{DGP 1.1} & 200 & 0.0470                    & 0.0500                    & 0.0460                    & 0.0390                    & 0.0640                                 & 0.0470            & 0.0540          & 0.0640          & 0.0520          & 0.0450          \\
                        & 500 & 0.0650                    & 0.0670                    & 0.0550                    & 0.0600                    & 0.0550                                 & 0.0740            & 0.0750          & 0.0700          & 0.0640          & 0.0590          \\
                        & 800 & 0.0530                    & 0.0540                    & 0.0510                    & 0.0480                    & 0.0490                                 & 0.0580            & 0.0690          & 0.0760          & 0.0510          & 0.0480          \\ \hline
\multirow[t]{3}{*}{DGP 1.2} & 200 & 0.0490                    & 0.0520                    & 0.0450                    & 0.0510                    & 0.0630                                 & 0.0670            & 0.0820          & 0.1000          & 0.0500          & 0.0520          \\
                        & 500 & 0.0380                    & 0.0530                    & 0.0480                    & 0.0560                    & 0.0540                                 & 0.0650            & 0.0610          & 0.0760          & 0.0460          & 0.0460          \\
                        & 800 & 0.0460                    & 0.0470                    & 0.0430                    & 0.0540                    & 0.0500                                 & 0.0530            & 0.0600          & 0.0590          & 0.0430          & 0.0410          \\ \hline
\multirow[t]{3}{*}{DGP 1.3} & 200 & 0.9990                    & 0.9980                    & 1.0000                    & 1.0000                    & 0.9970                                 & 0.2010            & 0.2770          & 0.3530          & 0.3230          & 0.2670          \\
                        & 500 & 1.0000                    & 1.0000                    & 1.0000                    & 1.0000                    & 1.0000                                 & 0.4530            & 0.6030          & 0.6880          & 0.7190          & 0.6350          \\
                        & 800 & 1.0000                    & 1.0000                    & 1.0000                    & 1.0000                    & 1.0000                                 & 0.6720            & 0.8010          & 0.8810          & 0.9060          & 0.8520          \\ \hline
\multirow[t]{3}{*}{DGP 1.4} & 200 & 0.9320                    & 0.9100                    & 0.9450                    & 0.9320                    & 0.8530                                 & 0.1390            & 0.1790          & 0.2320          & 0.2050          & 0.1870          \\
                        & 500 & 1.0000                    & 1.0000                    & 1.0000                    & 1.0000                    & 0.9990                                 & 0.2650            & 0.3570          & 0.4520          & 0.5120          & 0.4630          \\
                        & 800 & 1.0000                    & 1.0000                    & 1.0000                    & 1.0000                    & 1.0000                                 & 0.4130            & 0.5260          & 0.6390          & 0.7440          & 0.6890          \\ \hline
\multirow[t]{3}{*}{DGP 1.5} & 200 & 0.9830                    & 0.9620                    & 0.9930                    & 0.9820                    & 0.8800                                 & 0.3140            & 0.4050          & 0.5110          & 0.3100          & 0.3250          \\
                        & 500 & 1.0000                    & 1.0000                    & 1.0000                    & 1.0000                    & 1.0000                                 & 0.7110            & 0.8150          & 0.8940          & 0.7210          & 0.7320          \\
                        & 800 & 1.0000                    & 1.0000                    & 1.0000                    & 1.0000                    & 1.0000                                 & 0.9020            & 0.9620          & 0.9840          & 0.9050          & 0.9190          \\ \hline
\multirow[t]{3}{*}{DGP 1.6} & 200 & 0.6250                    & 0.5550                    & 0.6780                    & 0.6190                    & 0.4900                                 & 0.0890            & 0.1120          & 0.1290          & 0.0620          & 0.0610          \\
                        & 500 & 0.9870                    & 0.9840                    & 0.9910                    & 0.9880                    & 0.8890                                 & 0.1850            & 0.2370          & 0.2550          & 0.0920          & 0.0920          \\
                        & 800 & 0.9990                    & 0.9990                    & 1.0000                    & 1.0000                    & 0.9830                                 & 0.2080            & 0.2870          & 0.3570          & 0.0750          & 0.0950          \\ \hline
\multirow[t]{3}{*}{DGP 1.7} & 200 & 1.0000                    & 1.0000                    & 1.0000                    & 1.0000                    & 1.0000                                 & 0.9870            & 0.9920          & 0.9660          & 0.2660          & 0.4730          \\
                        & 500 & 1.0000                    & 1.0000                    & 1.0000                    & 1.0000                    & 1.0000                                 & 1.0000            & 1.0000          & 1.0000          & 0.7680          & 0.9530          \\
                        & 800 & 1.0000                    & 1.0000                    & 1.0000                    & 1.0000                    & 1.0000                                 & 1.0000            & 1.0000          & 1.0000          & 0.9690          & 0.9990          \\ \hline
\end{tabular}}
\label{DGPs1-1--1-7 DY1}
\end{table}

\subsection{Simulation 2}
\label{Simulation 2}

Now, we consider the four-dimensional case. Similar to \autoref{Simulation 1}, we consider a location-scale model and a dynamic model to generate $Y_t$:
\begin{align*}
&\text{LS2: } Y_{t1} = 0.5 + X_{t1} + 0.5X_{t2} + 0.8\varepsilon_{t1},\,\,\,Y_{t2} = -0.5 + 2X_{t1} - 0.5X_{t2} + \varepsilon_{t2},\\
&\qquad \,\,Y_{t3} = 1 - X_{t1} + 1.5X_{t2} + 1.2\varepsilon_{t3},\,\,\, Y_{t4} = 0.5X_{t1} - 1X_{t2}  + 0.7\varepsilon_{t4};\\
&\text{DY2: }
Y_{t1}=0.5+0.25Y_{t-1,1}+X_{t1}+0.5X_{t2}+\sigma_{t1}\varepsilon_{t1},\,\,\,
Y_{t2}=-0.5+0.4Y_{t-1,2}+2X_{t1}-0.5X_{t2}+\sigma_{t2}\varepsilon_{t2},\\
&\qquad\,\,\,
Y_{t3}=1+0.55Y_{t-1,3}-X_{t1}+1.5X_{t2}+\sigma_{t3}\varepsilon_{t3},\,\,\,
Y_{t4}=0.3Y_{t-1,4}+0.5X_{t1}-X_{t2}+\sigma_{t4}\varepsilon_{t4},\\
&\,\,\,\,\,\,\,\,\,\,\,\,\,\,\,\sigma_{t1}^2=0.0384+0.86\sigma_{t-1,1}^2+0.08\sigma_{t-1,1}^2\varepsilon_{t-1,1}^2,\,\,\,
\sigma_{t2}^2=0.08+0.80\sigma_{t-1,2}^2+0.12\sigma_{t-1,2}^2\varepsilon_{t-1,2}^2,\\
&\,\,\,\,\,\,\,\,\,\,\,\,\,\,\,\sigma_{t3}^2=0.0864+0.84\sigma_{t-1,3}^2+0.10\sigma_{t-1,3}^2\varepsilon_{t-1,3}^2,\,\,\,\sigma_{t4}^2=0.049+0.75\sigma_{t-1,4}^2+0.15\sigma_{t-1,4}^2\varepsilon_{t-1,4}^2,
\end{align*}
where $X_{t1},X_{t2}\stackrel{i.i.d.}\sim U[0,1]$ and are independent of $\varepsilon_t$. For the innovation copula, we consider 
the following DGPs:
\begin{align*}
\text{DGP 2.1: } C_t(u) 
=&\, C_{\mathrm{Ct}}^{(4)}(u;0.7);\\
\text{DGP 2.2: } C_t(u)
=&\, C_{\mathrm{Gs}}^{(2)}(u_{12};0.7)
    C_{\mathrm{Gs}}^{(2)}(u_{34};0.3);\\
\text{DGP 2.3: } C_t(u)
=&\, C_{\mathrm{Ct}}^{(2)}(u_{12};0.3)
    C_{\mathrm{Ct}}^{(2)}(u_{34};0.3)
    1(t/n\leqslant 0.5)\\
    &+ C_{\mathrm{Ct}}^{(2)}(u_{12};0.8)
    C_{\mathrm{Ct}}^{(2)}(u_{34};0.3)
    1(t/n>0.5);\\
\text{DGP 2.4: } C_t(u)
=&\, C_{\mathrm{Ct}}^{(4)}(u;0.3)
    1(t/n\leqslant 0.4)
  + C_{\mathrm{Gs}}^{(4)}(u;0.7)
    1(t/n>0.4);\\
\text{DGP 2.5: } C_t(u)
=&\, C_{\mathrm{Gs}}^{(2)}(u_{12};0)
    C_{\mathrm{Gs}}^{(2)}(u_{34};0.3)
    1(t/n\leqslant 0.3)\\
&+ C_{\mathrm{Gs}}^{(2)}(u_{12};0.3)
    C_{\mathrm{Gs}}^{(2)}(u_{34};0.3)
    1(0.3<t/n\leqslant 0.7)\\
&+ C_{\mathrm{Gs}}^{(2)}(u_{12};0.7)
    C_{\mathrm{Gs}}^{(2)}(u_{34};0.3)
    1(t/n>0.7);\\
\text{DGP 2.6: } C_t(u)
=&\, C_{\mathrm{Ct}}^{(4)}(u;\tau_t),\,\,\,
\tau_t=0.4\exp\left(-10(t/n-0.5)^2\right)+0.3;\\
\text{DGP 2.7: } C_t(u)
=&\, C_{\mathrm{Gs}}^{(2)}(u_{12};\tau_t)
    C_{\mathrm{Gs}}^{(2)}(u_{34};0.3),\,\,\,
\tau_t=0.8\sin(4\pi t/n),
\end{align*}
where $u = (u_{12}',u_{34}')'$ with $u_{12} = (u_1,u_2)'$ and $u_{34} = (u_3,u_4)'$. Note that DGPs 2.1--2.2 fall under $\mathcal{H}_0$, whereas DGPs 2.3--2.7 correspond to $\mathcal{H}_1$. In particular, DGPs 2.3, 2.5, and 2.7 represent sparse alternatives, under which Kendall's rank correlation changes only in a subset of component pairs. By contrast, DGPs 2.4 and 2.6 represent dense alternatives in which the dependence structure among all components changes jointly over time. We retain the same setups for $L$, $D_P(n)$, and $D_L(n)$ as in \autoref{Simulation 1}. The considered product bases are
\begin{align*}
&\pi_{(3)}(u) = (\varphi_1\varphi_1\varphi_0\varphi_0,\varphi_1\varphi_0\varphi_1\varphi_0,\varphi_1\varphi_0\varphi_0\varphi_1,\varphi_0\varphi_1\varphi_1\varphi_0,\varphi_0\varphi_1\varphi_0\varphi_1,\varphi_0\varphi_0\varphi_1\varphi_1)',\\
&\pi_{(4)}(u) = (\varphi_2\varphi_2\varphi_0\varphi_0,\varphi_2\varphi_0\varphi_2\varphi_0,\varphi_2\varphi_0\varphi_0\varphi_2,\varphi_0\varphi_2\varphi_2\varphi_0,\varphi_0\varphi_2\varphi_0\varphi_2,\varphi_0\varphi_0\varphi_2\varphi_2)'.
\end{align*}
We summarize the numerical results in Tables \ref{DGPs2-1--2-7 LS2}--\ref{DGPs2-1--2-7 DY2}.

Again, under the null hypothesis, both of our proposed tests deliver accurate empirical sizes comparable to those of the existing tests. Under the alternatives, however, the data-driven test $\hat{\Psi}_n^2(K(\hat{P}),\hat{L})$ generally performs better. For example, under DGP 2.7, where Kendall's rank correlation changes periodically only between the first two components, the directional tests based on $\pi_{(4)}$ suffer from severe power loss. In sharp contrast, the data-driven test maintains relatively stable power and uniformly outperforms the competing tests. Such a power enhancement can be attributed to the adaptive nature of the data-driven procedure: it automatically selects the truncation orders that capture the most pronounced deviations within the candidate sets, whereas $\hat{\Psi}_n^2(K,L)$ is restricted to prespecified directions. Moreover, under the dense alternatives DGPs 2.4 and 2.6, the data-driven test still performs comparably to the directional test while continuing to dominate the existing tests. Hence, for practical applications, the data-driven test $\hat{\Psi}_n^2(K(\hat{P}),\hat{L})$ is more recommended, as it automatically determines the truncation orders and adapts effectively to unknown patterns of copula changes.

To summarize, the simulation results highlight the practical usefulness of our proposed tests. We credit the good finite-sample performance of our asymptotic $\chi^2$ test to several distinctive features, which are (i) the transformation to jointly test the generalized Fourier coefficients, (ii) the elimination of the first-order estimation effect under the null hypothesis, and (iii) the avoidance of local smoothing parameters such as the bandwidth. Taken together, we believe our proposed tests offer a promising alternative to existing approaches for detecting copula structural changes in the SCOMDY model.

\begin{table}[htbp]
\renewcommand\arraystretch{0.69}
\centering
\caption{Empirical rejection rates under DGPs 2.1–2.7 with model LS2.} 
\resizebox{\linewidth}{!}{
\begin{tabular}{cccccccccccc}
\hline
DGP                      & $n$ & $\hat{\Psi}_{n,\pi_{(3)}}^2(3)$ & $\hat{\Psi}_{n,\pi_{(3)}}^2(4)$ & $\hat{\Psi}_{n,\pi_{(4)}}^2(3)$ & $\hat{\Psi}_{n,\pi_{(4)}}^2(4)$ & $\hat{\Psi}_{n}^2(K(\hat{P}),\hat{L})$ & $\hat{M}_n(0.5)$ & $\hat{M}_n(1)$ & $\hat{M}_n(2)$ & $\mathcal{T}_n$ & $\mathcal{S}_n$ \\ \hline
\multirow[t]{3}{*}{DGP 2.1} & 200 & 0.0480                          & 0.0490                          & 0.0410                          & 0.0430                          & 0.0390                                 & 0.0520           & 0.0500         & 0.0560         & 0.0560          & 0.0530          \\
                         & 500 & 0.0440                          & 0.0440                          & 0.0490                          & 0.0570                          & 0.0450                                 & 0.0620           & 0.0640         & 0.0800         & 0.0560          & 0.0630          \\
                         & 800 & 0.0360                          & 0.0500                          & 0.0510                          & 0.0470                          & 0.0480                                 & 0.0470           & 0.0540         & 0.0540         & 0.0400          & 0.0440          \\ \hline
\multirow[t]{3}{*}{DGP 2.2} & 200 & 0.0510                          & 0.0360                          & 0.0420                          & 0.0320                          & 0.0530                                 & 0.0640           & 0.0600         & 0.0630         & 0.0480          & 0.0490          \\
                         & 500 & 0.0390                          & 0.0480                          & 0.0360                          & 0.0410                          & 0.0450                                 & 0.0810           & 0.0770         & 0.0770         & 0.0410          & 0.0610          \\
                         & 800 & 0.0320                          & 0.0310                          & 0.0500                          & 0.0440                          & 0.0450                                 & 0.0640           & 0.0530         & 0.0650         & 0.0360          & 0.0440          \\ \hline
\multirow[t]{3}{*}{DGP 2.3} & 200 & 0.7010                          & 0.6350                          & 0.7370                          & 0.6600                          & 0.9090                                 & 0.0780           & 0.0880         & 0.1190         & 0.1620          & 0.1450          \\
                         & 500 & 0.9990                          & 0.9970                          & 0.9990                          & 0.9980                          & 1.0000                                 & 0.1320           & 0.1590         & 0.1950         & 0.4330          & 0.2780          \\
                         & 800 & 1.0000                          & 1.0000                          & 1.0000                          & 1.0000                          & 1.0000                                 & 0.2000           & 0.2510         & 0.3080         & 0.7210          & 0.4440          \\ \hline
\multirow[t]{3}{*}{DGP 2.4} & 200 & 0.7960                          & 0.7540                          & 0.8520                          & 0.8190                          & 0.9210                                 & 0.3430           & 0.4580         & 0.5670         & 0.7470          & 0.7040          \\
                         & 500 & 0.9990                          & 0.9990                          & 1.0000                          & 1.0000                          & 1.0000                                 & 0.7590           & 0.8730         & 0.9260         & 0.9860          & 0.9810          \\
                         & 800 & 1.0000                          & 1.0000                          & 1.0000                          & 1.0000                          & 1.0000                                 & 0.9420           & 0.9730         & 0.9910         & 1.0000          & 1.0000          \\ \hline
\multirow[t]{3}{*}{DGP 2.5} & 200 & 0.5670                          & 0.4850                          & 0.5070                          & 0.4210                          & 0.5500                                 & 0.1100           & 0.1410         & 0.1710         & 0.1970          & 0.1960          \\
                         & 500 & 0.9960                          & 0.9790                          & 0.9700                          & 0.9610                          & 0.9800                                 & 0.1750           & 0.2320         & 0.3130         & 0.4310          & 0.3710          \\
                         & 800 & 0.9990                          & 0.9990                          & 1.0000                          & 0.9990                          & 0.9990                                 & 0.2710           & 0.3500         & 0.4610         & 0.6840          & 0.5830          \\ \hline
\multirow[t]{3}{*}{DGP 2.6} & 200 & 0.5660                          & 0.4800                          & 0.4680                          & 0.4110                          & 0.5460                                 & 0.1810           & 0.2400         & 0.2820         & 0.0720          & 0.0960          \\
                         & 500 & 0.9810                          & 0.9630                          & 0.9570                          & 0.9270                          & 0.9770                                 & 0.4500           & 0.5660         & 0.6720         & 0.1440          & 0.2120          \\
                         & 800 & 1.0000                          & 1.0000                          & 0.9990                          & 0.9950                          & 1.0000                                 & 0.6440           & 0.7730         & 0.8560         & 0.2410          & 0.3410          \\ \hline
\multirow[t]{3}{*}{DGP 2.7} & 200 & 1.0000                          & 0.9990                          & 0.0400                          & 0.0510                          & 0.6760                                 & 0.3480           & 0.4180         & 0.3070         & 0.1370          & 0.1960          \\
                         & 500 & 1.0000                          & 1.0000                          & 0.0370                          & 0.0490                          & 1.0000                                 & 0.8410           & 0.9280         & 0.8680         & 0.4790          & 0.5750          \\
                         & 800 & 1.0000                          & 1.0000                          & 0.0430                          & 0.0480                          & 1.0000                                 & 0.9880           & 0.9990         & 0.9930         & 0.8030          & 0.8570          \\ \hline
\end{tabular}}
\label{DGPs2-1--2-7 LS2}
\end{table}

\begin{table}[htbp]
\renewcommand\arraystretch{0.69}
\centering
\caption{Empirical rejection rates under DGPs 2.1–2.7 with model DY2.} 
\resizebox{\linewidth}{!}{
\begin{tabular}{cccccccccccc}
\hline
DGP                      & $n$ & $\hat{\Psi}_{n,\pi_{(3)}}^2(3)$ & $\hat{\Psi}_{n,\pi_{(3)}}^2(4)$ & $\hat{\Psi}_{n,\pi_{(4)}}^2(3)$ & $\hat{\Psi}_{n,\pi_{(4)}}^2(4)$ & $\hat{\Psi}_{n}^2(K(\hat{P}),\hat{L})$ & $\hat{M}_n(0.5)$ & $\hat{M}_n(1)$ & $\hat{M}_n(2)$ & $\mathcal{T}_n$ & $\mathcal{S}_n$ \\ \hline
\multirow[t]{3}{*}{DGP 2.1} & 200 & 0.0440                          & 0.0390                          & 0.0540                          & 0.0530                          & 0.0510                                 & 0.0530           & 0.0640         & 0.0750         & 0.0450          & 0.0500          \\
                         & 500 & 0.0490                          & 0.0460                          & 0.0500                          & 0.0490                          & 0.0470                                 & 0.0460           & 0.0460         & 0.0530         & 0.0460          & 0.0450          \\
                         & 800 & 0.0430                          & 0.0560                          & 0.0550                          & 0.0470                          & 0.0420                                 & 0.0620           & 0.0660         & 0.0790         & 0.0550          & 0.0530          \\ \hline
\multirow[t]{3}{*}{DGP 2.2} & 200 & 0.0400                          & 0.0480                          & 0.0560                          & 0.0490                          & 0.0460                                 & 0.0730           & 0.0650         & 0.0660         & 0.0540          & 0.0500          \\
                         & 500 & 0.0520                          & 0.0610                          & 0.0430                          & 0.0370                          & 0.0450                                 & 0.0640           & 0.0750         & 0.0720         & 0.0470          & 0.0480          \\
                         & 800 & 0.0440                          & 0.0520                          & 0.0360                          & 0.0400                          & 0.0430                                 & 0.0640           & 0.0630         & 0.0610         & 0.0540          & 0.0510          \\ \hline
\multirow[t]{3}{*}{DGP 2.3} & 200 & 0.6870                          & 0.6100                          & 0.6840                          & 0.6230                          & 0.9140                                 & 0.0830           & 0.1090         & 0.1140         & 0.1480          & 0.1270          \\
                         & 500 & 1.0000                          & 0.9970                          & 0.9960                          & 0.9920                          & 1.0000                                 & 0.1350           & 0.1640         & 0.2050         & 0.4460          & 0.2980          \\
                         & 800 & 1.0000                          & 1.0000                          & 1.0000                          & 1.0000                          & 1.0000                                 & 0.1970           & 0.2490         & 0.3210         & 0.7160          & 0.4780          \\ \hline
\multirow[t]{3}{*}{DGP 2.4} & 200 & 0.7760                          & 0.7470                          & 0.8230                          & 0.7980                          & 0.9120                                 & 0.3450           & 0.4720         & 0.5520         & 0.7070          & 0.6530          \\
                         & 500 & 1.0000                          & 1.0000                          & 1.0000                          & 1.0000                          & 1.0000                                 & 0.7410           & 0.8570         & 0.9170         & 0.9890          & 0.9820          \\
                         & 800 & 1.0000                          & 1.0000                          & 1.0000                          & 1.0000                          & 1.0000                                 & 0.9290           & 0.9680         & 0.9810         & 0.9990          & 0.9970          \\ \hline
\multirow[t]{3}{*}{DGP 2.5} & 200 & 0.5320                          & 0.4680                          & 0.4690                          & 0.3960                          & 0.5460                                 & 0.1040           & 0.1120         & 0.1550         & 0.1700          & 0.1530          \\
                         & 500 & 0.9810                          & 0.9660                          & 0.9790                          & 0.9550                          & 0.9650                                 & 0.2020           & 0.2600         & 0.3300         & 0.4550          & 0.4010          \\
                         & 800 & 1.0000                          & 1.0000                          & 1.0000                          & 0.9980                          & 1.0000                                 & 0.2750           & 0.3480         & 0.4490         & 0.6850          & 0.5950          \\ \hline
\multirow[t]{3}{*}{DGP 2.6} & 200 & 0.5080                          & 0.4350                          & 0.4240                          & 0.3890                          & 0.4810                                 & 0.1740           & 0.2310         & 0.2630         & 0.0730          & 0.0860          \\
                         & 500 & 0.9850                          & 0.9650                          & 0.9420                          & 0.9120                          & 0.9800                                 & 0.4380           & 0.5740         & 0.6570         & 0.1480          & 0.2190          \\
                         & 800 & 0.9990                          & 0.9990                          & 0.9960                          & 0.9940                          & 1.0000                                 & 0.6610           & 0.7830         & 0.8500         & 0.2620          & 0.3670          \\ \hline
\multirow[t]{3}{*}{DGP 2.7} & 200 & 1.0000                          & 0.9990                          & 0.0370                          & 0.0460                          & 0.6250                                 & 0.3080           & 0.4020         & 0.2870         & 0.1360          & 0.1640          \\
                         & 500 & 1.0000                          & 1.0000                          & 0.0350                          & 0.0410                          & 1.0000                                 & 0.8310           & 0.9200         & 0.8770         & 0.4620          & 0.5850          \\
                         & 800 & 1.0000                          & 1.0000                          & 0.0370                          & 0.0430                          & 1.0000                                 & 0.9860           & 0.9960         & 0.9910         & 0.8300          & 0.8840          \\ \hline
\end{tabular}}
\label{DGPs2-1--2-7 DY2}
\end{table}

\section{Empirical applications}
\label{Empirical applications}

In this section, we evaluate the empirical performance of our proposed smooth directional test and data-driven test using two real-world datasets. In \autoref{Exchange rates}, we investigate the constancy of the innovation copula among three exchange rates considered by \citet{chen2006estimation}. Then, in \autoref{Stock returns}, we examine whether the innovation copula among the S\&P500, Cisco, and Intel returns, as analyzed by \citet{chan2009multivariate}, remains constant over time. Throughout this section, we retain the product construction \eqref{product bases} for $\{\pi_k(u)\}_{k=0}^\infty$. Both $\{\varphi_j(z)\}_{j=0}^\infty$ and $\{\psi_j(z)\}_{j=0}^\infty$ are chosen as the trigonometric bases in \eqref{triangonal} due to the dimension-free property.

\subsection{Exchange rates}
\label{Exchange rates}

We first analyze the exchange rate data considered by \citet{chen2006estimation}. In Section 6 of their paper, the authors examine the contemporaneous dependence among daily exchange rate returns for the Japanese yen, German mark, and French franc against the US dollar, denoted by JAP, DEM, and FRA, respectively. Let $S_{tj}$ denote the $j$th exchange rate on day $t$. For $j=1,2,\ldots,d$, the corresponding daily exchange rate return is constructed as $Y_{tj}=\log S_{tj}-\log S_{t-1,j}$. To account for serial dependence and conditional heteroskedasticity in each return series, the authors fit the following AR(1)--GARCH(1,1) model for the data:
\begin{align*}
Y_{tj}=c_j+\theta_{1j}Y_{t-1,j}+\sigma_{tj}\varepsilon_{tj},\,\,\,\sigma_{tj}^2=\omega_j+\theta_{2j}\sigma_{t-1,j}^2+\theta_{3j}\sigma_{t-1,j}^2\varepsilon_{t-1,j}^2,\,\,\,j=1,2,\ldots,d.
\end{align*}
The standardized innovations $\varepsilon_t=(\varepsilon_{t1},\ldots,\varepsilon_{td})'$ remove the marginal serial dynamics and volatility clustering from the exchange rate returns, while their copula captures the contemporaneous dependence across currencies. \citet{chen2006estimation} estimate and compare a wide range of parametric copula specifications for $\varepsilon_t$. Their analysis, however, requires the key assumption that the innovation copula remains constant throughout the sample period. We therefore revisit this assumption by applying our proposed directional and data-driven tests, along with competing tests, to examine whether the dependence structure among exchange rate innovations changes over time.

Following \citet{chen2006estimation}, we use daily exchange rate data for JAP, DEM, and FRA from April 28, 1986, to October 26, 1998, obtained from \url{https://www.federalreserve.gov/releases/H10/hist/}. The estimated coefficients of the AR(1)--GARCH(1,1) models are reported in \autoref{parameters exchange rates} and are numerically close to Table 1 of \citet{chen2006estimation}. Based on the fitted models, we obtain the standardized residuals and then apply the proposed methods to examine the constancy of the innovation copula.

\begin{table}[!h]
\renewcommand\arraystretch{0.67}
\centering
\caption{Parameter estimation for the AR(1)--GARCH(1,1) model.} 
\begin{tabular}{cccccc}
\hline
    & $c_j\times 10^{-5}$ & $\omega_j\times 10^{-7}$ & $\theta_{1j}$ & $\theta_{2j}$ & $\theta_{3j}$ \\ \hline
JAP & -6.8311             & 17.1456                 & 0.0476       & 0.8984       & 0.0675       \\
DEM & -4.7343             & 10.5631                 & 0.0408       & 0.9176       & 0.0590       \\
FRA & -6.8216             & 9.6405                  & 0.0406       & 0.9156       & 0.0610       \\ \hline
\end{tabular}
\label{parameters exchange rates}
\end{table}

We first test the constancy of the trivariate innovation copula. Although the original sample extends from April 28, 1986, to October 26, 1998, we restrict the analysis to the period from 1986 to 1991. This shorter subsample allows us to examine whether the dependence structure had already changed during the early part of the sample. For the directional test, we consider the following basis vectors:
\begin{align*}
\pi_{(5)}(u) = (\varphi_1\varphi_1\varphi_0, \varphi_0\varphi_1\varphi_1,\varphi_1\varphi_0\varphi_1)', \,\,\, \pi_{(6)}(u) = (\varphi_2\varphi_2\varphi_0, \varphi_0\varphi_2\varphi_2,\varphi_2\varphi_0\varphi_2)'.
\end{align*}
The choices of $L$, $D_P(n)$, $D_L(n)$, and $\hat{M}_n$ are the same as those used in \autoref{Simulations}, while the number of bootstrap replications for $\mathcal{T}_n$ and $\mathcal{S}_n$ is set to $B=2000$. The testing results are reported in \autoref{three exchange rates}. Most tests reject the null of copula constancy at the 1\% significance level, providing strong evidence that structural changes were already in place during the early part of the sample period.

\begin{table}[!h]
\renewcommand\arraystretch{0.75}
\centering
\caption{Testing $p$-values for the trivariate copula of the exchange rate data.} 
\resizebox{\linewidth}{!}{
\begin{tabular}{cccccccccc}
\hline
$\hat{\Psi}_{n,\pi_{(5)}}^2(3)$ & $\hat{\Psi}_{n,\pi_{(5)}}^2(4)$ & $\hat{\Psi}_{n,\pi_{(6)}}^2(3)$ & $\hat{\Psi}_{n,\pi_{(6)}}^2(4)$ & $\hat{\Psi}_{n}^2(K(\hat{P}),\hat{L})$ & $\hat{M}_n(0.5)$ & $\hat{M}_n(1)$ & $\hat{M}_n(2)$ & $\mathcal{T}_n$ & $\mathcal{S}_n$ \\\hline
$0.0001^{***}$                  & $0.0000^{***}$                  & $0.0000^{***}$                  & $0.0000^{***}$                  & $0.0000^{***}$                         & $0.0044^{***}$   & $0.2378$       & $0.5970$       & $0.0025^{***}$  & $0.0180^{**}$  \\\hline
\end{tabular}}
\vspace{-1.5em}
\begin{tablenotes}[flushleft]
\small
\setlength{\topsep}{0pt}
\setlength{\itemsep}{0pt}
\setlength{\parsep}{0pt}
\item \textit{Note.} $^{***}$, $^{**}$, and $^*$ indicate statistical significance at the 1\%, 5\%, and 10\% levels, respectively.
\end{tablenotes}
\label{three exchange rates}
\end{table}

To gain further insight into the structural changes detected in the trivariate copula, we examine copula constancy for each pair of exchange-rate innovations. The setups for the directional test, data-driven test, and $\hat{M}_n$ remain the same as those used in \autoref{Simulations}, and the number of bootstrap replications for $\mathcal{T}_n$ and $\mathcal{S}_n$ is set to $B=2000$. The testing $p$-values are summarized in \autoref{two exchange rates}. The results provide stronger evidence of time variation in the innovation copulas for the JAP--DEM and JAP--FRA pairs, whereas the evidence for the FRA--DEM pair is comparatively weaker. Taken together, these findings complement the copula analysis in Section 6 of \citet{chen2006estimation} by highlighting potential temporal variation in the dependence structure. Hence, when applying a constant copula specification to these exchange rates, it may be useful to consider a shorter sample period or to explicitly account for possible changes in the innovation copula over time.

\begin{table}[htbp]
\renewcommand\arraystretch{0.7}
\centering
\caption{Testing $p$-values for the bivariate copulas of the exchange rate data.} 
\resizebox{\linewidth}{!}{
\begin{tabular}{ccccccccccc}
\hline
                          & $\hat{\Psi}_{n,\pi_{(1)}}^2(3)$ & $\hat{\Psi}_{n,\pi_{(1)}}^2(4)$ & $\hat{\Psi}_{n,\pi_{(2)}}^2(3)$ & $\hat{\Psi}_{n,\pi_{(2)}}^2(4)$ & $\hat{\Psi}_{n}^2(K(\hat{P}),\hat{L})$ & $\hat{M}_n(0.5)$ & $\hat{M}_n(1)$ & $\hat{M}_n(2)$ & $\mathcal{T}_n$ & $\mathcal{S}_n$ \\ \hline
JAP $\leftrightarrow$ DEM & $0.0000^{***}$                  & $0.0000^{***}$                  & $0.0000^{***}$                  & $0.0000^{***}$                  & $0.0000^{***}$                         & $0.0011^{***}$   & $0.1081$       & $0.3079$       & $0.0015^{***}$  & $0.0045^{***}$  \\
JAP $\leftrightarrow$ FRA & $0.0000^{***}$                  & $0.0000^{***}$                  & $0.0000^{***}$                  & $0.0000^{***}$                  & $0.0000^{***}$                         & $0.0052^{***}$   & $0.1553$       & $0.4038$       & $0.0015^{***}$  & $0.0195^{**}$   \\
FRA $\leftrightarrow$ DEM & $0.0239^{**}$                   & $0.0069^{***}$                  & $0.0706^{*}$                    & $0.0223^{**}$                   & $0.1279$                               & $0.0565^{*}$     & $0.2486$       & $0.4329$       & $0.4470$        & $0.4900$        \\ \hline
\end{tabular}}
\vspace{-0.9em}
\begin{tablenotes}[flushleft]
\small
\setlength{\topsep}{0pt}
\setlength{\itemsep}{0pt}
\setlength{\parsep}{0pt}
\item \textit{Note.} $^{***}$, $^{**}$, and $^*$ indicate statistical significance at the 1\%, 5\%, and 10\% levels, respectively.
\end{tablenotes}
\label{two exchange rates}
\end{table}

\subsection{Stock returns}
\label{Stock returns}

We next analyze the daily log returns of the S\&P500 index (S\&P500), Cisco Systems (CSCO), and Intel Corporation (INTC) over the period from January 2, 1991, to December 31, 1999. The data are obtained from Yahoo Finance and are available at \url{https://finance.yahoo.com}. Let $S_{tj}$ denote the closing price of asset $j$ on day $t$. We construct the corresponding percentage log return as $Y_{tj}=100\times(\log S_{tj}-\log S_{t-1,j})$. Following \citet{chan2009multivariate}, we fit a GARCH(1,1) model to each return series:
\begin{align*}
Y_{tj}=\sigma_{tj}\varepsilon_{tj},\,\,\,\sigma_{tj}^2=\omega_j+\theta_{2j}\sigma_{t-1,j}^2+\theta_{3j}\sigma_{t-1,j}^2\varepsilon_{t-1,j}^2,\,\,\,j=1,2,\ldots,d.
\end{align*}
The resulting parameter estimates are reported in \autoref{parameters stock markets} and are numerically similar to Table 2 of \citet{chan2009multivariate}. The fitted models yield the standardized innovations $\varepsilon_t=(\varepsilon_{t1},\ldots,\varepsilon_{td})'$, whose copula characterizes the contemporaneous dependence among the three return series after accounting for their time-varying conditional volatilities. The copula analysis in \citet{chan2009multivariate}, as in \citet{chen2006estimation}, assumes that the innovation copula remains unchanged over time. We therefore apply our proposed methods to assess the empirical validity of this assumption. In what follows, we use the same specifications for all test statistics as in \autoref{Exchange rates}.

\begin{table}[!h]
\renewcommand\arraystretch{0.67}
\centering
\caption{Parameter estimation for the GARCH(1,1) model.} 
\begin{tabular}{cccc}
\hline
      & $\omega_j$ & $\theta_{2j}$ & $\theta_{3j}$ \\ \hline
S\&P500 & 0.0048     & 0.9466        & 0.0473        \\
CSCO  & 0.2679     & 0.9033        & 0.0647        \\
INTC  & 0.0457     & 0.9778        & 0.0150        \\ \hline
\end{tabular}
\label{parameters stock markets}
\end{table}

We first consider the trivariate innovation copula. The corresponding test results are reported in \autoref{three stock markets}. Most of the reported tests reject the null hypothesis of a constant innovation copula at the conventional 5\% significance level, suggesting that the dependence structure among the three stock returns may change over the sample period.

\begin{table}[!h]
\renewcommand\arraystretch{0.7}
\centering
\caption{Testing $p$-values for the trivariate copula of the stock return data.} 
\resizebox{\linewidth}{!}{
\begin{tabular}{cccccccccc}
\hline
$\hat{\Psi}_{n,\pi_{(5)}}^2(3)$ & $\hat{\Psi}_{n,\pi_{(5)}}^2(4)$ & $\hat{\Psi}_{n,\pi_{(6)}}^2(3)$ & $\hat{\Psi}_{n,\pi_{(6)}}^2(4)$ & $\hat{\Psi}_{n}^2(K(\hat{P}),\hat{L})$ & $\hat{M}_n(0.5)$ & $\hat{M}_n(1)$ & $\hat{M}_n(2)$ & $\mathcal{T}_n$ & $\mathcal{S}_n$ \\ \hline
$0.0136^{**}$                   & $0.0050^{***}$                  & $0.0003^{***}$                  & $0.0010^{***}$                  & $0.0428^{**}$                          & $0.0451^{**}$    & $0.0633^{*}$   & $0.0110^{**}$  & $0.3175$        & $0.1475$        \\ \hline
\end{tabular}}
\vspace{-1.5em}
\begin{tablenotes}[flushleft]
\small
\setlength{\topsep}{0pt}
\setlength{\itemsep}{0pt}
\setlength{\parsep}{0pt}
\item \textit{Note.} $^{***}$, $^{**}$, and $^*$ indicate statistical significance at the 1\%, 5\%, and 10\% levels, respectively.
\end{tablenotes}
\label{three stock markets}
\end{table}

We next examine the three bivariate innovation copulas. The corresponding results are reported in \autoref{two stock markets}. The evidence against copula constancy is strongest for the S\&P500--CSCO pair, while the corresponding evidence for the S\&P500--INTC and CSCO--INTC pairs is less pronounced. This finding is economically plausible given the substantial changes in financial market conditions over the 1991--1999 sample period. In particular, the sample spans the rapid expansion of the technology sector during the late-1990s dot-com boom, as well as episodes of severe market stress associated with the Asian financial crisis and the Russian default. Although this historical correspondence does not identify the source of the detected changes, it provides an economically plausible interpretation of our findings. Taken together, the results suggest that the empirical framework of \citet{chan2009multivariate} may be more appropriately applied to shorter and relatively homogeneous subsamples or extended to accommodate a time-varying innovation copula.

\begin{table}[htbp]
\renewcommand\arraystretch{0.7}
\centering
\caption{Testing $p$-values for the bivariate copulas of the stock return data.} 
\resizebox{\linewidth}{!}{
\begin{tabular}{ccccccccccc}
\hline
                             & $\hat{\Psi}_{n,\pi_{(1)}}^2(3)$ & $\hat{\Psi}_{n,\pi_{(1)}}^2(4)$ & $\hat{\Psi}_{n,\pi_{(2)}}^2(3)$ & $\hat{\Psi}_{n,\pi_{(2)}}^2(4)$ & $\hat{\Psi}_{n}^2(K(\hat{P}),\hat{L})$ & $\hat{M}_n(0.5)$ & $\hat{M}_n(1)$ & $\hat{M}_n(2)$ & $\mathcal{T}_n$ & $\mathcal{S}_n$ \\ \hline
S\&P500 $\leftrightarrow$ CSCO & $0.0001^{***}$                  & $0.0002^{***}$                  & $0.0000^{***}$                  & $0.0003^{***}$                  & $0.0874^{*}$                           & $0.5009$         & $0.5522$       & $0.3077$       & $0.2480$        & $0.5045$        \\
S\&P500 $\leftrightarrow$ INTC & $0.1419$                        & $0.1040$                        & $0.0838^{*}$                    & $0.0504^{*}$                    & $0.0528^{*}$                           & $0.3638$         & $0.2195$       & $0.0436^{**}$  & $0.2195$        & $0.2180$        \\
CSCO $\leftrightarrow$ INTC  & $0.0486^{**}$                   & $0.0479^{**}$                   & $0.0501^{*}$                    & $0.0573^{*}$                    & $0.0227^{**}$                          & $0.1660$         & $0.4113$       & $0.4015$       & $0.2310$        & $0.3445$        \\ \hline
\end{tabular}}
\vspace{-0.9em}
\begin{tablenotes}[flushleft]
\small
\setlength{\topsep}{0pt}
\setlength{\itemsep}{0pt}
\setlength{\parsep}{0pt}
\item \textit{Note.} $^{***}$, $^{**}$, and $^*$ indicate statistical significance at the 1\%, 5\%, and 10\% levels, respectively.
\end{tablenotes}
\label{two stock markets}
\end{table}

\section{Concluding comments}
\label{Concluding comments}

This paper develops new smooth tests for structural changes in the innovation copula within the SCOMDY framework of \citet{chen2006estimation}. We transform the original problem of testing copula constancy into a joint significance test for generalized Fourier coefficients along both the copula and time directions. This reformulation yields test statistics with asymptotically $\chi^2$ distributions, making statistical inference straightforward and computationally efficient. Under the null hypothesis, estimation of the dynamic parameters and the unknown marginals has no first-order effect, so the feasible statistics constructed from estimated residuals are asymptotically equivalent to their infeasible counterparts based on the latent innovations. Moreover, the proposed tests can detect local alternatives that approach the null at the parametric rate of $n^{-1/2}$, while avoiding local smoothing and the associated bandwidth selection. Compared with existing approaches, our tests avoid the resampling-based critical values required by \citet{nasri2022change} and the bandwidth choice involved in \citet{lu2025adaptive}. To enhance the practical applicability of our proposal, we also develop a data-driven test that automatically searches across both the copula and time directions for the most informative deviations. Our Monte Carlo experiments show that both tests deliver accurate empirical size, while the data-driven test achieves substantial power gains under sparse structural changes without sacrificing performance under dense alternatives. Two empirical applications further highlight the practical applicability of the proposed smooth tests.

This paper also motivates further topics for exploration. For example, once the null hypothesis of copula constancy is rejected, a natural next step is to move beyond detection to determine how the innovation copula changes. This calls for extending the time-invariant copula specification of \citet{chen2006estimation} to models that accommodate smooth evolution, multiple abrupt breaks, or dependence driven by observed state variables, together with valid estimation and inference for the timing and magnitude of such changes. A complete discussion of these procedures is beyond the scope of this article and is left for future research.

\setlength{\bibsep}{0.25em} 
\spacingset{0.5}
\bibliography{reference}
\bibliographystyle{chicago}
\spacingset{1.8}

\newpage
\clearpage
\newgeometry{left=1.5cm,right=1.5cm,top=2cm,bottom=2cm}
\def\spacingset#1{\renewcommand{\baselinestretch}%
{#1}\small\normalsize} \spacingset{1}
\spacingset{1.8}
\appendix
\setcounter{page}{1}
\begin{center}
{\Large\bf Detecting Copula Structural Changes:  \\ [ -0.8em ]
A Smooth Testing Approach \\[ -0.1em ]
Online Supplementary Material \par}
{\large Shiyao Huang and Xiaojun Song\\
Peking University}
\end{center}

\setcounter{section}{1}  
\renewcommand{\theequation}{\thesection.\arabic{equation}}
\setcounter{equation}{0}
\phantomsection

In the Online Supplementary Material, we provide the proofs of our main theoretical results. We first state and prove several auxiliary lemmas that support the proof of our main theorems. For notational convenience, we employ the notation $C$ to represent a universal constant, whose value may differ from line to line. 

For $j=1,2,\ldots,d$, define the infeasible empirical distribution function as
\begin{align*}
F_{nj}(x_j) = \frac{1}{n}\sum_{t=1}^n 1(\varepsilon_{tj}\leqslant x_j).
\end{align*}
We first establish an important lemma that characterizes the parameter estimation effect arising from $\hat{\alpha}_n$.

\begin{lemma}
\label{Lemma1}
Suppose Assumptions \ref{A1}--\ref{A3} hold. Then, as $n\rightarrow\infty$, 
\begin{align*}
\sqrt{n}\left(\hat{F}_{nj}(x_j) - F_{nj}(x_j)\right) = f_j(x_j)(a_j + x_jb_j)'\sqrt{n}(\hat{\alpha}_n - \alpha_0) + o_p(1),
\end{align*}
where the $o_p(1)$ is uniform in $x_j\in\mathbb{R}$.
\end{lemma}

\begin{proof}[\textit{\textbf{Proof of \autoref{Lemma1}.}}]
We prove \autoref{Lemma1} using a series of propositions.

\begin{proposition}
\label{pro1}
Suppose \autoref{A3} holds. Let $r,v\in\mathbb{R}$ and $w\geqslant 0$ be such that $|r|+|v|+w\leqslant 1/2$. Then, 
\begin{align*}
&(i)\,\,\, \sup_{x_j\in\mathbb{R}}\big|F_j(x_j + r + vx_j) - F_j(x_j)\big|\leqslant 2M_{0j}(|r|+|v|),\\
&(ii)\,\,\, \sup_{x_j\in\mathbb{R}}\big|F_j(x_j + r + vx_j) - F_j(x_j) - f_j(x_j)(r+vx_j)\big|\leqslant 4M_{1j}(r^2+v^2),\\
&(iii)\,\,\, \sup_{x_j\in\mathbb{R}}\big|F_j(x_j + r + vx_j + w(1+|x_j|)) - F_j(x_j + r + vx_j - w(1+|x_j|))\big|\leqslant 4M_{0j}w.
\end{align*}
\end{proposition}

\begin{proof}[\textit{\textbf{Proof of \autoref{pro1}.}}]
For $\theta\in[0,1]$, define $x_{j,\theta} = x_j + \theta(r+vx_j)$. Since $|r|+|v|\leqslant 1/2$, we have
\begin{align}
\frac{1}{2}(1+|x_j|)\leqslant (1-|r|-|v|)(1+|x_j|)\leqslant 1-|r| + (1-|v|)|x_j|\leqslant 1+|x_{j,\theta}|.
\label{pro1 1}
\end{align}
Then, under \autoref{A3}, we get
\begin{align}
f_j(x_{j,\theta})\leqslant \frac{M_{0j}}{1+|x_{j,\theta}|}\leqslant \frac{2M_{0j}}{1+|x_j|}.
\label{pro1 2}
\end{align}
Observe that
\begin{align}
F_j(x_j+r+vx_j) - F_j(x_j)=\int_{x_j}^{x_j+r+vx_j}f_j(x)dx = (r+vx_j)\int_0^1f_j(x_{j,\theta})d\theta.
\label{pro1 3}
\end{align}
Then, gathering \eqref{pro1 2}--\eqref{pro1 3}, we have
\begin{align*}
\big|F_j(x_j + r + vx_j) - F_j(x_j)\big|\leqslant 2M_{0j}\frac{|r|+|v||x_j|}{1+|x_j|}\leqslant 2M_{0j}(|r|+|v|),\,\,\, x_j\in\mathbb{R},
\end{align*}
which proves (i).

Next, for $\theta,\eta\in[0,1]$, define $x_{j,\eta\theta} = x_j+\eta\theta(r+vx_j)$. Then, by \eqref{pro1 1}, it follows that $(1+|x_j|)/2\leqslant 1+ |x_{j,\eta\theta}|$. Under \autoref{A3}, we further have
\begin{align}
|\dot f_j(x_{j,\eta\theta})|\leqslant \frac{M_{1j}}{(1+|x_{j,\eta\theta}|)^2}\leqslant \frac{4M_{1j}}{(1+|x_j|)^2}.
\label{pro1 4}
\end{align}
Note that
\begin{align*}
F_j(x_j+r+vx_j) - F_j(x_j)= & \int_{x_j}^{x_j+r+vx_j}f_j(x)dx\\
=&\int_{x_j}^{x_j+r+vx_j}\left(f_j(x_j) + \int_{x_j}^x\dot f_j(u)du\right)dx\\
=&f_j(x_j)(r+vx_j) + \int_{x_j}^{x_j+r+vx_j}\int_{x_j}^x\dot f_j(u)dudx.
\end{align*}
Thus, by a change of variables, we get
\begin{align}
F_j(x_j+r+vx_j) - F_j(x_j) - f_j(x_j)(r+vx_j) = (r+vx_j)^2\int_0^1\int_0^1\dot f_j(x_{j,\eta\theta})\theta d\eta d\theta.
\label{pro1 5}
\end{align}
Taking \eqref{pro1 4}--\eqref{pro1 5} together, we have
\begin{align*}
\big|F_j(x_j+r+vx_j) - F_j(x_j) - f_j(x_j)(r+vx_j)\big|\leqslant \frac{4M_{1j}(r+vx_j)^2}{(1+|x_j|)^2}\int_0^1\theta d\theta\leqslant 4M_{1j}(r^2 + v^2).
\end{align*}

Finally, for all $y\in [x_j+r+vx_j - w(1+|x_j|), x_j+r+vx_j + w(1+|x_j|)]$, we have
\begin{align*}
\frac{1}{2}(1+|x_j|)\leqslant (1 - |r| - |v| - w)(1+|x_j|)\leqslant 1-|r| + (1-|v|)|x_j| - w(1+|x_j|)\leqslant 1+|y|.
\end{align*}
Thus, we obtain
\begin{align*}
\big|F_j(x_j + r + vx_j + w(1+|x_j|)) - F_j(x_j + r + vx_j - w(1+|x_j|))\big|\leqslant \frac{2M_{0j}w(1+|x_j|)}{1+|y_\theta|}\leqslant 4M_{0j}w,
\end{align*}
where $y_\theta = x_j + r + vx_j + \theta w(1+|x_j|)$ with $\theta\in[-1,1]$ is a middle value between $x_j+r+vx_j - w(1+|x_j|)$ and $x_j+r+vx_j + w(1+|x_j|)$. Then, the proof is completed.
\end{proof}

\begin{proposition} 
\label{pro2} 
Suppose \autoref{A3} holds. For any fixed $M>0$, there exists $n_M\in\mathbb{N}^+$ such that, with probability 1, for all $n\geqslant n_M$, $t=1,2,\ldots,n$, $x_j\in\mathbb{R}$, and $h\in\mathbb{R}^{d_\alpha}$ satisfying $\|h\|\leqslant M$, \begin{align*} m_{tj}(\alpha_0+n^{-1/2}h)+x_js_{tj}(\alpha_0+n^{-1/2}h)=x_j+\frac{1}{\sqrt{n}}(a_{tj}+x_jb_{tj})'h+R_{ntj}(x_j,h), \end{align*} where the remainder satisfies \begin{align*} |R_{ntj}(x_j,h)|\leqslant\frac{M^2}{2n}(1+|x_j|)G_{2,tj}. \end{align*} \end{proposition}

\begin{proof}[\textit{\textbf{Proof of \autoref{pro2}.}}]
Note that $m_{tj}(\alpha_0) = 0$, $s_{tj}(\alpha_0) = 1$,  $\partial m_{tj}(\alpha_0)/\partial \alpha = a_{tj}$, and $\partial s_{tj}(\alpha_0)/\partial \alpha = b_{tj}$. Then, with a Taylor expansion with respect to $\alpha_0$, we obtain
\begin{align*}
m_{tj}(\alpha_0 + n^{-1/2}h) + x_js_{tj}(\alpha_0 + n^{-1/2}h) = x_j + \frac{1}{\sqrt{n}}(a_{tj} + x_{j}b_{tj})'h + R_{ntj}(x_j,h),
\end{align*}
where 
\begin{align*}
R_{ntj}(x_j,h) = \frac{1}{2n}h'\left(\frac{\partial^2 m_{tj}}{\partial \alpha\partial \alpha'}(\tilde{\alpha}_n) +x_j \frac{\partial^2 s_{tj}}{\partial \alpha\partial \alpha'}(\tilde{\alpha}_n)\right)h,
\end{align*}
and $\tilde{\alpha}_n = \alpha_0 + \theta h/\sqrt{n}$ for some $\theta\in[0,1]$. Note that there exists some integer $n_M$, such that for all $\theta\in[0,1]$ and $\|h\|\leqslant M$, $\alpha_0 + \theta h/\sqrt{n}\in\mathcal{A}_0$ for $n\geqslant n_M$. As such, as long as $n\geqslant n_M$, we have
\begin{align*}
|R_{ntj}(x_j,h)|\leqslant \frac{\|h\|^2}{2n}\left(\sup_{\alpha\in\mathcal{A}_0}\left\|\frac{\partial^2 m_{tj}}{\partial \alpha\partial \alpha'}(\alpha)\right\|+|x_j|\sup_{\alpha\in\mathcal{A}_0}\left\|\frac{\partial^2 s_{tj}}{\partial \alpha\partial \alpha'}(\alpha)\right\|\right)\leqslant \frac{M^2}{2n}(1+|x_j|)G_{2,tj},
\end{align*}
which completes the proof.
\end{proof}

\begin{proposition} 
\label{pro3} 
Suppose \autoref{A3} holds. Then, as $n\rightarrow\infty$,
\begin{align*}
\frac{1}{n}\sum_{t=1}^nG_{1,tj}^2= O_p(1),\,\,\, \frac{1}{n}\sum_{t=1}^nG_{2,tj} = O_p(1),\,\,\,\frac{1}{\sqrt{n}}\max_{1\leqslant t\leqslant n}G_{1,tj} = o_p(1),\,\,\,\frac{1}{n}\max_{1\leqslant t\leqslant n}G_{2,tj} = o_p(1).
\end{align*}
\end{proposition}

\begin{proof}[\textit{\textbf{Proof of \autoref{pro3}.}}]
Under \autoref{A3}, we have
\begin{align*}
\limsup_{n\rightarrow\infty}\frac{1}{n}\sum_{t=1}^n \mathbb{E}(G_{1,tj}^2)<\infty.
\end{align*}
Then, by Markov's inequality, we have $n^{-1}\sum_{t=1}^nG_{1,tj}^2= O_p(1)$. Next, for any $\varepsilon>0$, 
\begin{align*}
\mathbb{P}\left\{\max_{1\leqslant t\leqslant n}G_{1,tj}\geqslant \sqrt{n}\varepsilon\right\}\leqslant \sum_{t=1}^n\mathbb{P}\left\{G_{1,tj}\geqslant \sqrt{n}\varepsilon\right\}\leqslant \sum_{t=1}^n\frac{\mathbb{E}(G_{1,tj}^{2+\delta})}{n^{1+\delta/2}\varepsilon^{2+\delta}}\rightarrow 0
\end{align*}
by \autoref{A3}, implying that $n^{-1/2}\max_{1\leqslant t\leqslant n}G_{1,tj} = o_p(1)$. The proof for $G_{2,tj}$ is similar, and we omit it for brevity.
\end{proof}

For $j=1,2,\ldots,d$, define $e_{tj}(x) = 1(\varepsilon_{tj}\leqslant x) - F_j(x)$. 
\begin{proposition} 
\label{pro4} 
Suppose Assumptions \ref{A1} and \ref{A3} hold and $M>0$ is fixed. Then, as $n\rightarrow\infty$,
\begin{align*}
\sup_{\|h\|\leqslant M}\sup_{x\in\mathbb{R}}\left|\frac{1}{\sqrt{n}}\sum_{t=1}^n\left(e_{tj}(x + n^{-1/2}(a_{tj} + xb_{tj})'h)- e_{tj}(x)\right)\right|=o_p(1).
\end{align*}
\end{proposition}

\begin{proof}[\textit{\textbf{Proof of \autoref{pro4}.}}]
For $\beta,w\in\mathbb{R}^{1+2d_\alpha}$, $c\in \mathbb{R}$, let $H_{\beta,c}\subset \mathbb{R}^{1+2d_\alpha}$ denote the affine halfspace $H_{\beta,c} = \{w\in\mathbb{R}^{1+2d_\alpha}: \beta'w\leqslant c\}$. Define $\mathscr{H} = \{H_{\beta,c}:\beta\in\mathbb{R}^{1+2d_\alpha},c\in\mathbb{R}\}$. Let $\mathcal{G}$ denote the function space $\mathcal{G}=\{w\mapsto 1(w\in H_1) - 1(w\in H_2): H_1,H_2\in\mathscr{H}\}$. Obviously, $\mathcal{G}$ has envelope function $G(w)\equiv1$: $\sup_{g\in\mathcal{G}}|g(w)|\leqslant G(w)$ for all $w\in\mathbb{R}^{1+2d_\alpha}$.

Let $\mathscr M(\mathbb{R}^{1+2d_\alpha})$ denote the set of all measures $\gamma$ on $\mathbb{R}^{1+2d_\alpha}$ such that $\gamma(G^2) = \int_{\mathbb{R}^{1+2d_\alpha}}G^2(w)d\gamma(w)<\infty$. Since $\gamma(G^2)=\int_{\mathbb{R}^{1+2d_\alpha}}d\gamma(w)$, we further have $\mathscr M(\mathbb{R}^{1+2d_\alpha})=\{\gamma:\int_{\mathbb{R}^{1+2d_\alpha}}d\gamma(w)<\infty\}$. Let $\mathcal{F}$ denote a function class such that for all $f\in \mathcal{F}$, $f:\mathbb{R}^{1+2d_\alpha}\rightarrow \mathbb{R}$. Let $N(u,\mathcal{F},\|\cdot\|)$ denote the minimum number of balls $B(h,u)=\{g\in\mathcal{F}:\|g - h\|< u\}$ needed to cover $\mathcal{F}$. For any probability measure $Q$ on $\mathbb{R}^{1+2d_\alpha}$ and $f_1,f_2\in \mathcal{F}$, define $\|f_1-f_2\|_{Q,2} = (\int_{\mathbb{R}^{1+2d_\alpha}}(f_1(w) - f_2(w))^2dQ(w))^{1/2}$. Consider the function class $\mathcal{I}=\{w\mapsto 1(w\in H): H\in\mathscr{H}\}$. By Theorem 2.6.7 of van der Vaart and Wellner (1996), there exist constants $A_0$ and $v_0$ only depending on $d_\alpha$, such that for every probability measure $Q$, $N(u,\mathcal{I},\|\cdot\|_{Q,2})\leqslant (A_0/u)^{v_0}$, $0<u<1$. Since $N(u,\mathcal{G},\|\cdot\|_{Q,2})\leqslant N^2(u/2,\mathcal{I},\|\cdot\|_{Q,2})$, we have constants $A$ and $v$, only depending on $d_\alpha$, such that $N(u,\mathcal{G},\|\cdot\|_{Q,2})\leqslant (A/u)^{v}$ holds for any probability measure $Q$ and $0<u<1$.

Now, for any $\gamma\in \mathscr M(\mathbb{R}^{1+2d_\alpha})$, define probability measure $Q_\gamma$ by $dQ_\gamma= \gamma^{-1}(G^2)d\gamma$. As such, for any $g_1,g_2\in\mathcal{G}$, 
\begin{align*}
d_\gamma(g_1,g_2):=\left(\int_{\mathbb{R}^{1+2d_\alpha}}(g_1(w) - g_2(w))^2d\gamma(w)\right)^{1/2} = \gamma^{1/2}(G^2)\|g_1-g_2\|_{Q_\gamma,2}.
\end{align*}
It then follows that 
\begin{align*}
N(u\gamma^{1/2}(G^2),\mathcal{G},d_\gamma) = N(u,\mathcal{G},\|\cdot\|_{Q_\gamma,2})\leqslant (A/u)^v,\,\,\, 0<u<1,
\end{align*}
where the constants $A$ and $v$ do not depend on $\gamma$. Therefore, if $0<u<1$, 
\begin{align*}
\sup_{\gamma\in \mathscr M(\mathbb{R}^{1+2d_\alpha})}\left[\ln N(u\gamma^{1/2}(G^2),\mathcal{G},d_\gamma)\right]^{1/2}\leqslant \sqrt{v\ln A + v\ln(1/u)}.
\end{align*}
If $u> 1$, since $0\in\mathcal{G}$ and $|g|\leqslant 1$ for all $g\in\mathcal{G}$, we have
\begin{align*}
d^2_\gamma(g,0) = \int_{\mathbb{R}^{1+2d_\alpha}}g^2(w)d\gamma(w)\leqslant \gamma(G^2). 
\end{align*}
Then, we directly obtain $N(u\gamma^{1/2}(G^2),\mathcal{G},d_\gamma)=1$ for $u> 1$. As such, it follows that
\begin{align}
\int_0^\infty \sup_{\gamma\in \mathscr M(\mathbb{R}^{1+2d_\alpha})}\left[\ln N(u\gamma^{1/2}(G^2),\mathcal{G},d_\gamma)\right]^{1/2}du \leqslant \int_0^1 \sqrt{v\ln A + v\ln(1/u)}du<\infty.
\label{uniform entropy}
\end{align}
That is, $\mathcal{G}$ has uniformly integrable entropy with respect to the $\mathcal{L}_2$-norm.

Define the random variable $W_{tj} = (\varepsilon_{tj},a_{tj}',b_{tj}')'$. 
For any $g\in\mathcal{G}$, define $V_{ntj}(g) = n^{-1/2}(g(W_{tj}) - \mathbb{E}[g(W_{tj})|\mathcal{F}_{t-1}])$. Obviously, $V_{ntj}(g)$ is $\mathcal{F}_t$-measurable and $\mathbb{E}[V_{ntj}(g)|\mathcal{F}_{t-1}]=0$. Define the random measure 
\begin{align*}
\mu_{n,j}(\mathcal{A}) = \frac{1}{n}\sum_{t=1}^n\mathbb{P}\{W_{tj}\in\mathcal{A}|\mathcal{F}_{t-1}\}.
\end{align*}
For any $g_1,g_2\in\mathcal{G}$, define
\begin{align*}
d_{\mu_{n,j}}(g_1,g_2) = \left(\frac{1}{n}\sum_{t=1}^n\mathbb{E}[(g_1(W_{tj}) - g_2(W_{tj}))^2|\mathcal{F}_{t-1}]\right)^{1/2}.
\end{align*}
Define 
\begin{align*}
\sigma^2_{n,j}(g_1,g_2) := \sum_{t=1}^n \mathbb{E}[(V_{ntj}(g_1) - V_{ntj}(g_2))^2|\mathcal{F}_{t-1}]. 
\end{align*}
Obviously, $\sigma^2_{n,j}(g_1,g_2)\leqslant d_{\mu_{n,j}}^2(g_1,g_2)$ for all $g_1,g_2\in\mathcal{G}$, indicating that
\begin{align}
\mathbb{P}\left\{\sup_{g_1,g_2\in\mathcal{G}}\frac{\sigma^2_{n,j}(g_1,g_2)}{d_{\mu_{n,j}}^2(g_1,g_2)}\geqslant 2\right\}=0.
\label{conditional variance}
\end{align}

Define $\overline{V}_{ntj} = \sup_{g\in\mathcal{G}}|V_{ntj}(g)|$. Note that $\overline{V}_{ntj}\leqslant 2/\sqrt{n}$. As such, for any $\eta>0$, as long as $n\geqslant 4/\eta^2$, we get $\overline V_{ntj}\leqslant \eta$. It then follows that
\begin{align}
\frac{6}{\eta}\sum_{t=1}^n\mathbb{E}[\overline V_{ntj}^21(\overline V_{ntj}>\eta)]\rightarrow 0,\,\,\,n\rightarrow\infty.
\label{lindeberg}
\end{align}
For $g\in\mathcal{G}$, let $S_{nj}(g)=\sum_{t=1}^nV_{ntj}(g)$. Gathering \eqref{uniform entropy}, \eqref{conditional variance}, and \eqref{lindeberg}, by Theorem 1 of Bae et al. (2010), for any $\varepsilon,\eta>0$, there exists $z>0$ such that
\begin{align}
\limsup_{n\rightarrow\infty}\mathbb{P}\left\{\sup_{\substack{g_1,g_2\in\mathcal{G}\\d_{\mu_{n,j}}(g_1,g_2)<z}}|S_{nj}(g_1) - S_{nj}(g_2)|>\eta\right\}<\varepsilon.
\label{SE}
\end{align}

Let $w = (\varepsilon,a',b')'\in\mathbb{R}^{1+2d_\alpha}$ with $a,b\in\mathbb{R}^{d_\alpha}$ and $\varepsilon\in\mathbb{R}$. Define the function 
\begin{align*}
g_{n,x,h}(w) = 1(\varepsilon\leqslant x + n^{-1/2}(a+xb)'h) - 1(\varepsilon\leqslant x).
\end{align*}
Obviously, $g_{n,x,h}\in\mathcal{G}$ for all $n\in\mathbb{N}^+$, $x\in\mathbb{R}$, and $h\in\mathbb{R}^{d_\alpha}$. Then, it is straightforward to see that
\begin{align*}
\frac{1}{\sqrt{n}}\sum_{t=1}^n\left(e_{tj}(x + n^{-1/2}(a_{tj} + xb_{tj})'h)- e_{tj}(x)\right) = S_{nj}(g_{n,x,h}) - S_{nj}(0) = S_{nj}(g_{n,x,h}).
\end{align*}
Define the random variable $\xi_{n,M} = \sup_{\|h\|\leqslant M}\sup_{x\in\mathbb{R}}d_{\mu_{n,j}}(g_{n,x,h},0)$. Since $\varepsilon_{tj}$ is independent of $\mathcal{F}_{t-1}$, we have
\begin{align*}
\xi_{n,M}^2 = \sup_{\|h\|\leqslant M}\sup_{x\in\mathbb{R}}\frac{1}{n}\sum_{t=1}^n\left|F_j\left(x+\frac{1}{\sqrt{n}}(a_{tj} + xb_{tj})'h\right)- F_j(x)\right|.
\end{align*}
Note that
\begin{align*}
\left|\frac{a_{tj}'h}{\sqrt{n}}\right| + \left|\frac{b_{tj}'h}{\sqrt{n}}\right|\leqslant \frac{MG_{1,tj}}{\sqrt{n}}.
\end{align*}
In addition, by \autoref{pro3}, for any $\varepsilon
>0$, there exists some constant $M_\varepsilon>0$ such that 
\begin{align*}
\limsup_{n\rightarrow\infty}\mathbb{P}\left\{\frac{1}{n}\sum_{t=1}^nG_{1,tj}>M_\varepsilon\right\}<\varepsilon.
\end{align*}
As such, by \autoref{pro1}, it follows that 
\begin{align*}
&\mathbb{P}\{\sqrt{n}\xi_{n,M}^2>2M_{0j}MM_\varepsilon\}\\
= & \mathbb{P}\left\{\sqrt{n}\xi_{n,M}^2>2M_{0j}MM_\varepsilon,\frac{M}{\sqrt{n}}\max_{1\leqslant t\leqslant n}G_{1,tj}\geqslant 1/2\right\} + \mathbb{P}\left\{\sqrt{n}\xi_{n,M}^2>2M_{0j}MM_\varepsilon,\frac{M}{\sqrt{n}}\max_{1\leqslant t\leqslant n}G_{1,tj} <1/2\right\}\\
\leqslant &\mathbb{P}\left\{\frac{M}{\sqrt{n}}\max_{1\leqslant t\leqslant n}G_{1,tj}\geqslant 1/2\right\} + \mathbb{P}\left\{\frac{1}{n}\sum_{t=1}^nG_{1,tj}>M_\varepsilon\right\}.
\end{align*}
Thus, by \autoref{pro3}, we get 
\begin{align*}
\limsup_{n\rightarrow\infty }\mathbb{P}\{\sqrt{n}\xi_{n,M}^2>2M_{0j}MM_\varepsilon\}\leqslant \limsup_{n\rightarrow\infty}\mathbb{P}\left\{\frac{1}{n}\sum_{t=1}^nG_{1,tj}>M_\varepsilon\right\}<\varepsilon.
\end{align*}
That is, $\xi_{n,M}^2 = O_p(n^{-1/2})$, implying that $\xi_{n,M} = O_p(n^{-1/4}) = o_p(1)$.

Finally, for $\varepsilon$ and $\eta$ given in \eqref{SE}, we have
\begin{align*}
&\mathbb{P}\left\{\sup_{\|h\|\leqslant M}\sup_{x\in\mathbb{R}}|S_{nj}(g_{n,x,h})|>\eta\right\}\\
= &\mathbb{P}\left\{\sup_{\|h\|\leqslant M}\sup_{x\in\mathbb{R}}|S_{nj}(g_{n,x,h})|>\eta,\xi_{n,M}<z\right\} + \mathbb{P}\left\{\sup_{\|h\|\leqslant M}\sup_{x\in\mathbb{R}}|S_{nj}(g_{n,x,h})|>\eta,\xi_{n,M}\geqslant z\right\}\\
\leqslant &\mathbb{P}\left\{\sup_{\substack{g_1,g_2\in\mathcal{G}\\d_{\mu_{n,j}}(g_1,g_2)<z}}|S_{nj}(g_1) - S_{nj}(g_2)|>\eta\right\} + \mathbb{P}\{\xi_{n,M}\geqslant z\}.
\end{align*}
By \eqref{SE} and the fact that $\xi_{n,M} = o_p(1)$, taking limit with respect to $n$ on both sides of the above inequality yields that
\begin{align*}
\limsup_{n\rightarrow\infty}\mathbb{P}\left\{\sup_{\|h\|\leqslant M}\sup_{x\in\mathbb{R}}|S_{nj}(g_{n,x,h})|>\eta\right\}\leqslant \limsup_{n\rightarrow\infty}\mathbb{P}\left\{\sup_{\substack{g_1,g_2\in\mathcal{G}\\d_{\mu_{n,j}}(g_1,g_2)<z}}|S_{nj}(g_1) - S_{nj}(g_2)|>\eta\right\}<\varepsilon.
\end{align*}
Since $\varepsilon,\eta>0$ are arbitrary, we have
\begin{align*}
\limsup_{n\rightarrow\infty}\mathbb{P}\left\{\sup_{\|h\|\leqslant M}\sup_{x\in\mathbb{R}}|S_{nj}(g_{n,x,h})|>\eta\right\}=0,
\end{align*}
which completes the proof.
\end{proof}

\begin{proposition}
\label{pro5}
Suppose Assumptions \ref{A1} and \ref{A3} hold and $M>0$ is fixed. Then, as $n\rightarrow\infty$,
\begin{align*}
\sup_{\|h\|\leqslant M}\sup_{x\in\mathbb{R}}\left|\frac{1}{\sqrt{n}}\sum_{t=1}^n\left(F_j(x+n^{-1/2}(a_{tj} + xb_{tj})'h) - F_j(x)\right) - \frac{f_j(x)}{n}\sum_{t=1}^n(a_{tj} + xb_{tj})'h\right| = o_p(1).
\end{align*}
\end{proposition}

\begin{proof}[\textit{\textbf{Proof of \autoref{pro5}.}}]
Observe that
\begin{align*}
\left|\frac{a_{tj}'h}{\sqrt{n}}\right| + \left|\frac{b_{tj}'h}{\sqrt{n}}\right|\leqslant \frac{M}{\sqrt{n}}\max_{1\leqslant t\leqslant n}G_{1,tj},\,\,\,\|h\|\leqslant M.
\end{align*}
Thus, if $n^{-1/2}M\max_{1\leqslant t\leqslant n}G_{1,tj}\leqslant 1/2$, \autoref{pro1} implies that
\begin{align*}
&\sup_{\|h\|\leqslant M}\sup_{x\in\mathbb{R}}\left|\frac{1}{\sqrt{n}}\sum_{t=1}^n\left(F_j(x+n^{-1/2}(a_{tj} + xb_{tj})'h) - F_j(x)\right) - \frac{f_j(x)}{n}\sum_{t=1}^n(a_{tj} + xb_{tj})'h\right|\\
\leqslant & \frac{1}{\sqrt{n}}\sum_{t=1}^n\sup_{\|h\|\leqslant M}\sup_{x\in\mathbb{R}}\left|\left(F_j(x+n^{-1/2}(a_{tj} + xb_{tj})'h) - F_j(x)\right) - \frac{f_j(x)}{\sqrt{n}}(a_{tj} + xb_{tj})'h\right|\\
\leqslant & \frac{4M_{1j}M^2}{n^{3/2}}\sum_{t=1}^n(\|a_{tj}\|^2+\|b_{tj}\|^2)\\
\leqslant &\frac{4M_{1j}M^2}{n^{3/2}}\sum_{t=1}^n G_{1,tj}^2.
\end{align*}
By \autoref{pro3}, we have $n^{-3/2}\sum_{t=1}^n G_{1,tj}^2 = o_p(1)$. Thus, for any $\varepsilon>0$, 
\begin{align*}
&\mathbb{P}\left\{\sup_{\|h\|\leqslant M}\sup_{x\in\mathbb{R}}\left|\frac{1}{\sqrt{n}}\sum_{t=1}^n\left(F_j(x+n^{-1/2}(a_{tj} + xb_{tj})'h) - F_j(x)\right) - \frac{f_j(x)}{n}\sum_{t=1}^n(a_{tj} + xb_{tj})'h\right|>\varepsilon\right\}\\
\leqslant &\mathbb{P}\left\{\frac{M}{\sqrt{n}}\max_{1\leqslant t\leqslant n}G_{1,tj}>\frac{1}{2}\right\} + \mathbb{P}\left\{\frac{4M_{1j}M^2}{n^{3/2}}\sum_{t=1}^nG_{1,tj}^2>\varepsilon\right\}\\
\rightarrow& 0,\,\,\, n\rightarrow\infty,
\end{align*}
which completes the proof.
\end{proof}

For notational brevity, let $q_{ntj}(x,h) = m_{tj}(\alpha_0+n^{-1/2}h)+xs_{tj}(\alpha_0 + n^{-1/2}h)$ and $\ell_{ntj}(x,h) = x+ n^{-1/2}(a_{tj} + xb_{tj})'h$.

\begin{proposition}
\label{pro6}
Suppose Assumptions \ref{A1} and \ref{A3} hold and $M>0$ is fixed. Then, as $n\rightarrow\infty$,
\begin{align*}
\sup_{\|h\|\leqslant M}\sup_{x\in\mathbb{R}}\left|\frac{1}{\sqrt{n}}\sum_{t=1}^n\big(1(\varepsilon_{tj}\leqslant q_{ntj}(x,h) ) - 1(\varepsilon_{tj}\leqslant \ell_{ntj}(x,h) )\big)\right| = o_p(1).
\end{align*}
\end{proposition}

\begin{proof}[\textit{\textbf{Proof of \autoref{pro6}.}}]
By \autoref{pro2}, with probability 1, for $n$ large enough, we have $q_{ntj}(x,h) = \ell_{ntj}(x,h) + R_{ntj}(x,h)$, where 
\begin{align*}
|R_{ntj}(x,h)|\leqslant \frac{M^2}{2n}G_{2,tj}(1+|x|):=w_{ntj}(1+|x|).
\end{align*}
It then follows that $\ell_{ntj}(x,h) - w_{ntj}(1+|x|)\leqslant q_{ntj}(x,h)\leqslant \ell_{ntj}(x,h) + w_{ntj}(1+|x|)$. Define $J_{ntj}(x,h) = 1(\varepsilon_{tj}\leqslant \ell_{ntj}(x,h) + w_{ntj}(1+|x|)) - 1(\varepsilon_{tj}\leqslant \ell_{ntj}(x,h) - w_{ntj}(1+|x|))$. Since $|1(\varepsilon_{tj}\leqslant q_{ntj}(x,h) ) - 1(\varepsilon_{tj}\leqslant \ell_{ntj}(x,h) )|\leqslant J_{ntj}(x,h)$, it then suffices to prove
\begin{align}
\sup_{\|h\|\leqslant M}\sup_{x\in\mathbb{R}}\frac{1}{\sqrt{n}}\sum_{t=1}^nJ_{ntj}(x,h) = o_p(1).
\label{J ntj}
\end{align}

Define $p_{ntj}(x,h) = \mathbb{E}[J_{ntj}(x,h)|\mathcal{F}_{t-1}]$. We first show that 
\begin{align}
\sup_{\|h\|\leqslant M}\sup_{x\in\mathbb{R}}\frac{1}{\sqrt{n}}\sum_{t=1}^n p_{ntj}(x,h) = o_p(1).
\label{p ntj}
\end{align}
Since $\varepsilon_{tj}$ is independent of $\mathcal{F}_{t-1}$, we have
\begin{align*}
p_{ntj}(x,h) = F_j\big(\ell_{ntj}(x,h) +  w_{ntj}(1+|x|)\big) - F_j\big(\ell_{ntj}(x,h) -  w_{ntj}(1+|x|)\big).
\end{align*}
Observe that
\begin{align*}
\left|\frac{a_{tj}'h}{\sqrt{n}}\right| + \left|\frac{b_{tj}'h}{\sqrt{n}}\right| + w_{ntj}\leqslant \frac{M}{\sqrt{n}}\max_{1\leqslant t\leqslant n}G_{1,tj} + \frac{M^2}{2n}\max_{1\leqslant t\leqslant n}G_{2,tj}.
\end{align*}
Then, if $n^{-1/2}M\max_{1\leqslant t\leqslant n}G_{1,tj} + (2n)^{-1}M^2\max_{1\leqslant t\leqslant n}G_{2,tj}\leqslant 1/2$, \autoref{pro1} implies that
\begin{align*}
\sup_{\|h\|\leqslant M}\sup_{x\in\mathbb{R}}p_{ntj}(x,h)\leqslant \frac{2M_{0j}M^2}{n} G_{2,tj}.
\end{align*}
By \autoref{pro3}, we have $n^{-3/2}\sum_{t=1}^nG_{2,tj}=o_p(1)$. Thus, for any $\varepsilon>0$, we get
\begin{align*}
&\mathbb{P}\left\{\sup_{\|h\|\leqslant M}\sup_{x\in\mathbb{R}}\frac{1}{\sqrt{n}}\sum_{t=1}^n p_{ntj}(x,h) >\varepsilon\right\}\\
= & \mathbb{P}\left\{\sup_{\|h\|\leqslant M}\sup_{x\in\mathbb{R}}\frac{1}{\sqrt{n}}\sum_{t=1}^n p_{ntj}(x,h) >\varepsilon,\frac{M}{\sqrt{n}}\max_{1\leqslant t\leqslant n}G_{1,tj} + \frac{M^2}{2n}\max_{1\leqslant t\leqslant n}G_{2,tj}>\frac{1}{2}\right\}\\
&+ \mathbb{P}\left\{\sup_{\|h\|\leqslant M}\sup_{x\in\mathbb{R}}\frac{1}{\sqrt{n}}\sum_{t=1}^n p_{ntj}(x,h) >\varepsilon,\frac{M}{\sqrt{n}}\max_{1\leqslant t\leqslant n}G_{1,tj} + \frac{M^2}{2n}\max_{1\leqslant t\leqslant n}G_{2,tj}\leqslant \frac{1}{2}\right\}\\
\leqslant & \mathbb{P}\left\{\frac{M}{\sqrt{n}}\max_{1\leqslant t\leqslant n}G_{1,tj} + \frac{M^2}{2n}\max_{1\leqslant t\leqslant n}G_{2,tj}>\frac{1}{2}\right\} + \mathbb{P}\left\{\frac{2M_{0j}M^2}{n^{3/2}}\sum_{t=1}^n G_{2,tj}>\varepsilon\right\}\\
\rightarrow & 0,\,\,\, n\rightarrow\infty,
\end{align*}
which proves \eqref{p ntj}.

Next, we show that
\begin{align}
\sup_{\|h\|\leqslant M}\sup_{x\in\mathbb{R}}\left|\frac{1}{\sqrt{n}}\sum_{t=1}^n\big(J_{ntj}(x,h) - p_{ntj}(x,h)\big)\right| = o_p(1).
\label{centered J}
\end{align}
In fact, if we define $\tilde{H}_{\beta,c} = \{w\in\mathbb{R}^{2+2d_\alpha}:\beta'w\leqslant c\}$, $\tilde{\mathscr{H}} = \{\tilde{H}_{\beta,c}:\beta\in\mathbb{R}^{2+2d_\alpha}, c\in\mathbb{R}\}$, and $\tilde{\mathcal{G}} = \{w\mapsto1(w\in\tilde{H}_1) - 1(w\in\tilde{H}_2):\tilde{H}_1,\tilde{H}_2\in\tilde{\mathscr{H}}\}$, it is easy to see that an analogue of \eqref{SE} holds for $\tilde{\mathcal G}$. That is, for any $\varepsilon,\eta>0$, there exists $z>0$, such that
\begin{align}
\limsup_{n\rightarrow\infty}\mathbb{P}\left\{\sup_{\substack{g_1,g_2\in\tilde{\mathcal G}\\d_{\tilde\mu_{n,j}}(g_1,g_2)<z}}|S_{nj}(g_1) - S_{nj}(g_2)|>\eta\right\}<\varepsilon,
\label{SE tilde G}
\end{align}
where 
\begin{align*}
\tilde{\mu}_{n,j}(\mathcal{A}) = \frac{1}{n}\sum_{t=1}^n\mathbb{P}\{\tilde{W}_{tj}\in\mathcal{A}|\mathcal{F}_{t-1}\}
\end{align*}
with $\tilde{W}_{tj} = (\varepsilon_{tj},a_{tj}',b_{tj}',G_{2,tj})'$, and 
\begin{align*}
d_{\tilde\mu_{n,j}}(g_1,g_2) = \left(\frac{1}{n}\sum_{t=1}^n\mathbb{E}[(g_1(\tilde W_{tj}) - g_2(\tilde W_{tj}))^2|\mathcal{F}_{t-1}]\right)^{1/2}.
\end{align*}
Let $\tilde{w} = (\varepsilon,a',b',g)'\in\mathbb{R}^{2+2d_\alpha}$. Define the function 
\begin{align*}
\mathcal{J}_{n,x,h}(\tilde w) = 1\left(\varepsilon\leqslant x+ \frac{1}{\sqrt{n}}(a + xb)'h + \frac{M^2}{2n}g(1+|x|)\right) - 1\left(\varepsilon\leqslant x+ \frac{1}{\sqrt{n}}(a + xb)'h - \frac{M^2}{2n}g(1+|x|)\right).
\end{align*}
Obviously, $\mathcal{J}_{n,x,h},0\in \tilde{\mathcal G}$ and $\mathcal{J}_{n,x,h}(\tilde{W}_{tj}) = J_{ntj}(x,h)$. Since $\mathcal{J}_{n,x,h}^2=\mathcal{J}_{n,x,h}$ for $g\geqslant 0$, we have
\begin{align*}
d_{\tilde\mu_{n,j}}^2(\mathcal{J}_{n,x,h},0) = \frac{1}{n}\sum_{t=1}^n\Big(F_j\big(\ell_{ntj}(x,h) +  w_{ntj}(1+|x|)\big) - F_j\big(\ell_{ntj}(x,h) -  w_{ntj}(1+|x|)\big)\Big) = \frac{1}{n}\sum_{t=1}^np_{ntj}(x,h).
\end{align*}
Thus, by \eqref{p ntj}, we have $\tilde{\xi}_{n,M} = \sup_{\|h\|\leqslant M}\sup_{x\in\mathbb{R}}d_{\tilde{\mu}_{n,j}}(\mathcal{J}_{n,x,h},0) = o_p(1)$.
Finally, for $\varepsilon$ and $\eta$ given in \eqref{SE tilde G}, we have
\begin{align*}
&\mathbb{P}\left\{\sup_{\|h\|\leqslant M}\sup_{x\in\mathbb{R}}|S_{nj}(\mathcal{J}_{n,x,h})|>\eta\right\}\\
= &\mathbb{P}\left\{\sup_{\|h\|\leqslant M}\sup_{x\in\mathbb{R}}|S_{nj}(\mathcal{J}_{n,x,h})|>\eta,\tilde\xi_{n,M}<z\right\} + \mathbb{P}\left\{\sup_{\|h\|\leqslant M}\sup_{x\in\mathbb{R}}|S_{nj}(\mathcal{J}_{n,x,h})|>\eta,\tilde\xi_{n,M}\geqslant z\right\}\\
\leqslant &\mathbb{P}\left\{\sup_{\substack{g_1,g_2\in\tilde{\mathcal{G}}\\d_{\tilde\mu_{n,j}}(g_1,g_2)<z}}|S_{nj}(g_1) - S_{nj}(g_2)|>\eta\right\} + \mathbb{P}\{\tilde\xi_{n,M}\geqslant z\}.
\end{align*}
By \eqref{SE tilde G} and the fact that $\tilde\xi_{n,M} = o_p(1)$, taking limit with respect to $n$ on both sides of the above inequality yields that
\begin{align*}
\limsup_{n\rightarrow\infty}\mathbb{P}\left\{\sup_{\|h\|\leqslant M}\sup_{x\in\mathbb{R}}|S_{nj}(\mathcal{J}_{n,x,h})|>\eta\right\}\leqslant \limsup_{n\rightarrow\infty}\mathbb{P}\left\{\sup_{\substack{g_1,g_2\in\tilde{\mathcal{G}}\\d_{\tilde\mu_{n,j}}(g_1,g_2)<z}}|S_{nj}(g_1) - S_{nj}(g_2)|>\eta\right\}<\varepsilon.
\end{align*}
Since $\varepsilon,\eta>0$ are arbitrary, we have
\begin{align*}
\limsup_{n\rightarrow\infty}\mathbb{P}\left\{\sup_{\|h\|\leqslant M}\sup_{x\in\mathbb{R}}|S_{nj}(\mathcal{J}_{n,x,h})|>\eta\right\}=0,
\end{align*}
which yields \eqref{centered J}. Then, \eqref{J ntj} is proved by gathering \eqref{p ntj} and \eqref{centered J}.
\end{proof}

\begin{proposition}
\label{pro7}
Suppose Assumptions \ref{A1} and \ref{A3} hold and $M>0$ is fixed. Then, as $n\rightarrow\infty$,
\begin{align*}
\sup_{\|h\|\leqslant M}\sup_{x\in\mathbb{R}}\left|\frac{1}{\sqrt{n}}\sum_{t=1}^n\big(1(\varepsilon_{tj}\leqslant q_{ntj}(x,h)) - 1(\varepsilon_{tj}\leqslant x)\big) - \frac{f_j(x)}{n}\sum_{t=1}^n(a_{tj} + xb_{tj})'h\right| = o_p(1).
\end{align*}
\end{proposition}

\begin{proof}[\textit{\textbf{Proof of \autoref{pro7}.}}]
Note that
\begin{align*}
&\frac{1}{\sqrt{n}}\sum_{t=1}^n\big(1(\varepsilon_{tj}\leqslant q_{ntj}(x,h)) - 1(\varepsilon_{tj}\leqslant x)\big) - \frac{f_j(x)}{n}\sum_{t=1}^n(a_{tj} + xb_{tj})'h\\
=& \frac{1}{\sqrt{n}}\sum_{t=1}^n\big(1(\varepsilon_{tj}\leqslant q_{ntj}(x,h) ) - 1(\varepsilon_{tj}\leqslant \ell_{ntj}(x,h) )\big) + \frac{1}{\sqrt{n}}\sum_{t=1}^n\left(e_{tj}(x + n^{-1/2}(a_{tj} + xb_{tj})'h)- e_{tj}(x)\right)\\
&+ \frac{1}{\sqrt{n}}\sum_{t=1}^n\left(F_j(x+n^{-1/2}(a_{tj} + xb_{tj})'h) - F_j(x)\right) - \frac{f_j(x)}{n}\sum_{t=1}^n(a_{tj} + xb_{tj})'h.
\end{align*}
Then, by Propositions \ref{pro4}--\ref{pro6}, the desired result directly follows.
\end{proof}

With the above preparations, we are ready to proceed with the proof of \autoref{Lemma1}. Let $\hat{h}_n = \sqrt{n}(\hat{\alpha}_n - \alpha_0)$. By \autoref{A2}, we have $\hat h_n = O_p(1)$. That is, for any $\varepsilon>0$, there exists a constant $M>0$, such that 
\begin{align}
\limsup_{n\rightarrow \infty} \mathbb{P}\{\|\hat h_n\|>M\}<\varepsilon.
\label{hat h n}
\end{align}

Next, since $1(\hat{\varepsilon}_{tj}\leqslant x) = 1(\varepsilon_{tj}\leqslant q_{ntj}(x,\hat{h}_n))$, we get
\begin{align*}
&\sqrt{n}\left(\hat{F}_{nj}(x) - F_{nj}(x)\right) - \frac{f_j(x)}{n}\sum_{t=1}^n(a_{tj} + xb_{tj})'\sqrt{n}(\hat{\alpha}_n - \alpha_0) \\
=& \frac{1}{\sqrt{n}}\sum_{t=1}^n\big(1(\varepsilon_{tj}\leqslant q_{ntj}(x,\hat{h}_n)) - 1(\varepsilon_{tj}\leqslant x)\big) - \frac{f_j(x)}{n}\sum_{t=1}^n(a_{tj} + xb_{tj})'\hat{h}_n.
\end{align*}
Then, for any $\eta>0$, we have
\begin{align*}
&\mathbb{P}\left\{\sup_{x\in\mathbb{R}}\left|\frac{1}{\sqrt{n}}\sum_{t=1}^n\big(1(\varepsilon_{tj}\leqslant q_{ntj}(x,\hat{h}_n)) - 1(\varepsilon_{tj}\leqslant x)\big) - \frac{f_j(x)}{n}\sum_{t=1}^n(a_{tj} + xb_{tj})'\hat{h}_n\right|>\eta\right\}\\
 = &\mathbb{P}\left\{\sup_{x\in\mathbb{R}}\left|\frac{1}{\sqrt{n}}\sum_{t=1}^n\big(1(\varepsilon_{tj}\leqslant q_{ntj}(x,\hat{h}_n)) - 1(\varepsilon_{tj}\leqslant x)\big) - \frac{f_j(x)}{n}\sum_{t=1}^n(a_{tj} + xb_{tj})'\hat{h}_n\right|>\eta, \|\hat{h}_n\|\leqslant M\right\}\\
 &+ \mathbb{P}\left\{\sup_{x\in\mathbb{R}}\left|\frac{1}{\sqrt{n}}\sum_{t=1}^n\big(1(\varepsilon_{tj}\leqslant q_{ntj}(x,\hat{h}_n)) - 1(\varepsilon_{tj}\leqslant x)\big) - \frac{f_j(x)}{n}\sum_{t=1}^n(a_{tj} + xb_{tj})'\hat{h}_n\right|>\eta, \|\hat{h}_n\|> M\right\}\\
 \leqslant & \mathbb{P}\left\{\sup_{\|h\|\leqslant M}\sup_{x\in\mathbb{R}}\left|\frac{1}{\sqrt{n}}\sum_{t=1}^n\big(1(\varepsilon_{tj}\leqslant q_{ntj}(x,h)) - 1(\varepsilon_{tj}\leqslant x)\big) - \frac{f_j(x)}{n}\sum_{t=1}^n(a_{tj} + xb_{tj})'h\right|>\eta\right\} \\
 &+ \mathbb{P}\{\|\hat h_n\|>M\}.
\end{align*}
By \autoref{pro7} and \eqref{hat h n}, it follows that
\begin{align*}
&\limsup_{n\rightarrow\infty}\mathbb{P}\left\{\sup_{x\in\mathbb{R}}\left|\frac{1}{\sqrt{n}}\sum_{t=1}^n\big(1(\varepsilon_{tj}\leqslant q_{ntj}(x,\hat{h}_n)) - 1(\varepsilon_{tj}\leqslant x)\big) - \frac{f_j(x)}{n}\sum_{t=1}^n(a_{tj} + xb_{tj})'\hat{h}_n\right|>\eta\right\}\\
\leqslant & \limsup_{n\rightarrow\infty}\mathbb{P}\left\{\sup_{\|h\|\leqslant M}\sup_{x\in\mathbb{R}}\left|\frac{1}{\sqrt{n}}\sum_{t=1}^n\big(1(\varepsilon_{tj}\leqslant q_{ntj}(x,h)) - 1(\varepsilon_{tj}\leqslant x)\big) - \frac{f_j(x)}{n}\sum_{t=1}^n(a_{tj} + xb_{tj})'h\right|>\eta\right\} \\
&+ \limsup_{n\rightarrow\infty} \mathbb{P}\{\|\hat h_n\|>M\}\\
<& \varepsilon.
\end{align*}
Since $\varepsilon>0$ can be arbitrarily small, we have 
\begin{align*}
\sup_{x\in\mathbb{R}}\left|\frac{1}{\sqrt{n}}\sum_{t=1}^n\big(1(\varepsilon_{tj}\leqslant q_{ntj}(x,\hat{h}_n)) - 1(\varepsilon_{tj}\leqslant x)\big) - \frac{f_j(x)}{n}\sum_{t=1}^n(a_{tj} + xb_{tj})'\hat{h}_n\right| = o_p(1), 
\end{align*}
implying that
\begin{align}
\sqrt{n}\left(\hat{F}_{nj}(x) - F_{nj}(x)\right) = \frac{f_j(x)}{n}\sum_{t=1}^n(a_{tj} + xb_{tj})'\sqrt{n}(\hat{\alpha}_n - \alpha_0) + o_p(1),
\label{Lemma1 eq1}
\end{align}
where the $o_p(1)$ is uniform in $x\in\mathbb{R}$.

Finally, note that
\begin{align}
\begin{split}
\frac{f_j(x)}{n}\sum_{t=1}^n(a_{tj} + xb_{tj})'\sqrt{n}(\hat{\alpha}_n - \alpha_0) = &f_j(x)(a_j + xb_j)'\sqrt{n}(\hat{\alpha}_n - \alpha_0) \\
&+ f_j(x)\big[(\bar{a}_{nj} - a_j) + x(\bar{b}_{nj} - b_j)\big]'\sqrt{n}(\hat{\alpha}_n - \alpha_0),
\end{split}
\label{Lemma1 eq2}
\end{align}
where 
\begin{align}
\begin{split}
&\sup_{x\in\mathbb{R}}\left|f_j(x)\big[(\bar{a}_{nj} - a_j) + x(\bar{b}_{nj} - b_j)\big]'\sqrt{n}(\hat{\alpha}_n - \alpha_0)\right|\\
\leqslant &M_{0j}\big(\|\bar{a}_{nj} - a_j\| + \|\bar{b}_{nj} - b_j\|\big)\|\hat{h}_n\|\\
=&o_p(1)
\end{split}
\label{Lemma1 eq3}
\end{align}
by Assumptions \ref{A2}--\ref{A3}. Therefore, gathering \eqref{Lemma1 eq1}, \eqref{Lemma1 eq2}, and \eqref{Lemma1 eq3} completes the proof of \autoref{Lemma1}.
\end{proof}

With \autoref{Lemma1} in hand, we immediately obtain the following corollary.

\begin{corollary}
\label{coro1}
Suppose Assumptions \ref{A1}--\ref{A3} hold. Then, as $n\rightarrow\infty$,
\begin{align*}
\sup_{x_j\in\mathbb{R}}\left|\hat{F}_{nj}(x_j) - F_j(x_j)\right| = O_p(n^{-1/2}).
\end{align*}
\end{corollary}

\begin{proof}[\textit{\textbf{Proof of \autoref{coro1}.}}]
By \autoref{Lemma1}, we obtain
\begin{align}
\sqrt{n}\left(\hat{F}_{nj}(x_j) - F_{j}(x_j)\right) = \sqrt{n}\left({F}_{nj}(x_j) - F_{j}(x_j)\right) + f_j(x_j)(a_j + x_jb_j)'\sqrt{n}(\hat{\alpha}_n - \alpha_0) + o_p(1),
\label{coro1 1}
\end{align}
where the $o_p(1)$ is uniform in $x_j\in\mathbb{R}$. By Corollary 1 of  Massart (1990), we have
\begin{align}
\sup_{x_j\in\mathbb{R}}\left|\sqrt{n}\left({F}_{nj}(x_j) - F_{j}(x_j)\right)\right| = O_p(1). 
\label{coro1 2}
\end{align}
Then, by Assumptions \ref{A2}--\ref{A3}, we get
\begin{align}
\sup_{x_j\in\mathbb{R}}\left|f_j(x_j)(a_j + x_jb_j)'\sqrt{n}(\hat{\alpha}_n - \alpha_0)\right|\leqslant M_{0j}(\|a_j\| + \|b_j\|)\|\sqrt{n}(\hat{\alpha}_n - \alpha_0)\| = O_p(1).
\label{coro1 3}
\end{align}
Thus, gathering \eqref{coro1 1}, \eqref{coro1 2}, and \eqref{coro1 3} completes the proof.
\end{proof}

\begin{lemma}
\label{Lemma2}
Suppose Assumptions \ref{A1}--\ref{A3} hold. Then, as $n\rightarrow\infty$,
\begin{align*}
\frac{1}{n}\sum_{t=1}^n(\hat{\varepsilon}_{tj} - \varepsilon_{tj})^2 = O_p(n^{-1}).
\end{align*}
\end{lemma}

\begin{proof}[\textit{\textbf{Proof of \autoref{Lemma2}.}}]
By direct calculations, we have
\begin{align*}
\varepsilon_{tj} - \hat{\varepsilon}_{tj} = \frac{m_{tj}(\hat{\alpha}_n) + \varepsilon_{tj}(s_{tj}(\hat{\alpha}_n) - 1)}{s_{tj}(\hat{\alpha}_n)}.
\end{align*}

We first establish that $\max_{1\leqslant t\leqslant n}|s_{tj}(\hat{\alpha}_n) - 1| = o_p(1)$. Consider $\mathcal{C}_n = \{\|\hat\alpha_n - \alpha_0\|\leqslant \rho\}$, where $\rho>0$ is such that $\{\alpha:\|\alpha - \alpha_0\|\leqslant \rho\}\subset \mathcal{A}_0$. By \autoref{A2}, we have $\mathbb{P}\left\{\mathcal{C}_n^c\right\}\rightarrow0$, $n\rightarrow\infty$. In the event $\mathcal{C}_n$, we have $\alpha_0 + r(\hat{\alpha}_n - \alpha_0)\in\mathcal{A}_0$ for all $r\in[0,1]$. Then, we obtain
\begin{align}
\left|s_{tj}(\hat{\alpha}_n) - 1\right|\leqslant \left|\int_0^1\frac{\partial s_{tj}}{\partial \alpha'}(\alpha_0 + r(\hat{\alpha}_n - \alpha_0))(\hat{\alpha}_n - \alpha_0)dr\right|\leqslant G_{1,tj}\left\|\hat{\alpha}_n - \alpha_0\right\|.
\label{Lemma2 1}
\end{align}
Thus, it follows that
\begin{align*}
\max_{1\leqslant t\leqslant n}\left|s_{tj}(\hat{\alpha}_n) - 1\right|\leqslant \frac{1}{\sqrt{n}}\max_{1\leqslant t\leqslant n} G_{1,tj}\left\|\sqrt{n}(\hat{\alpha}_n - \alpha_0)\right\|.
\end{align*}
By \autoref{pro3} and \autoref{A2}, we have $n^{-1/2}\max_{1\leqslant t\leqslant n} G_{1,tj}\left\|\sqrt{n}(\hat{\alpha}_n - \alpha_0)\right\|=o_p(1)$. Thus, for any $\varepsilon>0$, we have
\begin{align*}
\mathbb{P}\left\{\max_{1\leqslant t\leqslant n}\left|s_{tj}(\hat{\alpha}_n) - 1\right|>\varepsilon\right\} = &\mathbb{P}\left\{\max_{1\leqslant t\leqslant n}\left|s_{tj}(\hat{\alpha}_n) - 1\right|>\varepsilon,\mathcal{C}_n\right\} \\
&+ \mathbb{P}\left\{\max_{1\leqslant t\leqslant n}\left|s_{tj}(\hat{\alpha}_n) - 1\right|>\varepsilon,\mathcal{C}_n^c\right\}\\
\leqslant & \mathbb{P}\left\{\frac{1}{\sqrt{n}}\max_{1\leqslant t\leqslant n} G_{1,tj}\left\|\sqrt{n}(\hat{\alpha}_n - \alpha_0)\right\|>\varepsilon\right\} + \mathbb{P}\left\{\mathcal{C}_n^c\right\}\\
\rightarrow &0,\,\,\, n\rightarrow\infty.
\end{align*}
Then, we immediately obtain $\mathbb{P}\left\{\min_{1\leqslant t\leqslant n}s_{tj}(\hat{\alpha}_n)<1/2\right\}\rightarrow 0$, $n\rightarrow\infty$. 

Consider the event $\mathcal{E}_n = \{\min_{1\leqslant t\leqslant n}s_{tj}(\hat{\alpha}_n)<1/2\}\cup \mathcal{C}_n^c$. The above analysis implies that $\mathbb{P}\{\mathcal{E}_n\}\rightarrow 0$, $n\rightarrow\infty$. Note that $m_{tj}(\alpha_0) = 0$, $s_{tj}(\alpha_0) = 1$. Then, if $\mathcal{E}^c_n$ occurs,  following \eqref{Lemma2 1}, we have $|m_{tj}(\hat{\alpha}_n) + \varepsilon_{tj}(s_{tj}(\hat{\alpha}_n) - 1)|\leqslant G_{1,tj}(1+|\varepsilon_{tj}|)\left\|\hat{\alpha}_n - \alpha_0\right\|$. In addition, 
\begin{align*}
\frac{1}{n}\sum_{t=1}^n(\hat{\varepsilon}_{tj} - \varepsilon_{tj})^2\leqslant \frac{8\left\|\hat{\alpha}_n - \alpha_0\right\|^2}{n}\sum_{t=1}^nG_{1,tj}^2(1+\varepsilon_{tj}^2).
\end{align*}
Since $G_{1,tj}$ is $\mathcal{F}_{t-1}$ measurable and $\varepsilon_{tj}$ is independent of $\mathcal{F}_{t-1}$, \autoref{A3} implies that
\begin{align*}
\limsup_{n\rightarrow\infty}\mathbb{E}\left(\frac{1}{n}\sum_{t=1}^nG_{1,tj}^2(1+\varepsilon_{tj}^2)\right) = \limsup_{n\rightarrow\infty}\frac{2}{n}\sum_{t=1}^n\mathbb{E}(G_{1,tj}^2)<\infty.
\end{align*}
Thus, 
\begin{align*}
 \frac{8\left\|\hat{\alpha}_n - \alpha_0\right\|^2}{n}\sum_{t=1}^nG_{1,tj}^2(1+\varepsilon_{tj}^2) = O_p(n^{-1}).
\end{align*}
That is, for any $\varepsilon>0$, there exists some constant $M_\varepsilon$ such that
\begin{align*}
\limsup_{n\rightarrow \infty}\mathbb{P}\left\{8\left\|\hat{\alpha}_n - \alpha_0\right\|^2\sum_{t=1}^nG_{1,tj}^2(1+\varepsilon_{tj}^2)>M_\varepsilon\right\}<\varepsilon.
\end{align*}
It then follows that
\begin{align*}
\mathbb{P}\left\{\sum_{t=1}^n(\hat{\varepsilon}_{tj} - \varepsilon_{tj})^2>M_\varepsilon\right\} =& \mathbb{P}\left\{\sum_{t=1}^n(\hat{\varepsilon}_{tj} - \varepsilon_{tj})^2>M_\varepsilon,\mathcal{E}_n\right\} + \mathbb{P}\left\{\sum_{t=1}^n(\hat{\varepsilon}_{tj} - \varepsilon_{tj})^2>M_\varepsilon,\mathcal{E}_n^c\right\}\\
\leqslant & \mathbb{P}\{\mathcal{E}_n\} + \mathbb{P}\left\{ 8\left\|\hat{\alpha}_n - \alpha_0\right\|^2\sum_{t=1}^nG_{1,tj}^2(1+\varepsilon_{tj}^2)>M_\varepsilon\right\}.
\end{align*}
By letting $n$ diverge to infinity on both sides of the inequality, we get
\begin{align*}
\limsup_{n\rightarrow\infty}\mathbb{P}\left\{\sum_{t=1}^n(\hat{\varepsilon}_{tj} - \varepsilon_{tj})^2>M_\varepsilon\right\}\leqslant \limsup_{n\rightarrow\infty}\mathbb{P}\left\{ 8\left\|\hat{\alpha}_n - \alpha_0\right\|^2\sum_{t=1}^nG_{1,tj}^2(1+\varepsilon_{tj}^2)>M_\varepsilon\right\}<\varepsilon,
\end{align*}
which completes the proof.
\end{proof}

\begin{lemma}
\label{Lemma3}
Suppose Assumptions \ref{A1}--\ref{A3} hold. Then, as $n\rightarrow\infty$,
\begin{align*}
\max_{1\leqslant t\leqslant n}\left|\hat{\varepsilon}_{tj} - \varepsilon_{tj}\right| = o_p(1).
\end{align*}
\end{lemma}

\begin{proof}[\textit{\textbf{Proof of \autoref{Lemma3}.}}]
Note that in the event $\mathcal{E}_n^c$, we have
\begin{align*}
\max_{1\leqslant t \leqslant n}\left|\hat{\varepsilon}_{tj} - \varepsilon_{tj}\right| \leqslant \frac{2}{\sqrt{n}}\max_{1\leqslant t\leqslant n}G_{1,tj}(1+|\varepsilon_{tj}|)\left\|\sqrt{n}(\hat{\alpha}_n - \alpha_0)\right\|. 
\end{align*}
As such, it suffices to show that 
\begin{align}
\frac{1}{\sqrt{n}}\max_{1\leqslant t\leqslant n}G_{1,tj}(1+|\varepsilon_{tj}|) = o_p(1). 
\label{Lemma3 1}
\end{align}

For notational brevity, now we redefine $W_{tj}$ as $W_{tj} = 1+|\varepsilon_{tj}|$. Note that $y^21(W_{tj}>y)\leqslant W_{tj}^21(W_{tj}>y)\leqslant W_{tj}^2$ and $\mathbb{E}(W_{tj}^2)\leqslant 4$, the Dominated Convergence Theorem (DCT) then implies that
\begin{align}
\lim_{y\rightarrow\infty}y^2\mathbb{P}\{W_{tj}>y\}\leqslant \lim_{y\rightarrow\infty}\mathbb{E}\left(W_{tj}^21(W_{tj}>y)\right) = \mathbb{E}\left(\lim_{y\rightarrow\infty}W_{tj}^21(W_{tj}>y)\right)=0.
\label{DCT}
\end{align}
Define $r_j(y) = y^2\mathbb{P}\{W_{tj}>y\}$ and $\tilde{r}_j(y) = \sup_{z\geqslant y} r_j(z)$. \eqref{DCT} indicates that $r_j(y)\leqslant \tilde{r}_j(y)$ for each $y$ and $\tilde{r}_j(y)\downarrow 0$ as $y\rightarrow\infty$.

For any $\varepsilon>0$, we have
\begin{align*}
\mathbb{P}\left\{n^{-1/2}\max_{1\leqslant t\leqslant n}G_{1,tj}(1 + |\varepsilon_{tj}|)>\varepsilon\right\}\leqslant \sum_{t=1}^n\mathbb{P}\left\{G_{1,tj}W_{tj}>\sqrt{n}\varepsilon\right\}.
\end{align*}
Given that $G_{1,tj}$ is $\mathcal{F}_{t-1}$-measurable, $W_{tj}\perp\mathcal{F}_{t-1}$, and that $G_{1,tj} = 0$ implies $\mathbb{P}\left\{G_{1,tj}W_{tj}>\sqrt{n}\varepsilon\right\}=0$, we obtain
\begin{align*}
\mathbb{P}\left\{G_{1,tj}W_{tj}>\sqrt{n}\varepsilon\right\} = \mathbb{E}[1(G_{1,tj}W_{tj}>\sqrt{n}\varepsilon)] = \mathbb{E}\Big(\mathbb{E}[1(G_{1,tj}W_{tj}>\sqrt{n}\varepsilon)|\mathcal{F}_{t-1}]\Big) = \mathbb{E}\left[\frac{G_{1,tj}^2}{n\varepsilon^2}r_j\left(\frac{\sqrt{n}\varepsilon}{G_{1,tj}}\right)\right],
\end{align*}
which further implies 
\begin{align}
\mathbb{P}\left\{n^{-1/2}\max_{1\leqslant t\leqslant n}G_{1,tj}(1 + |\varepsilon_{tj}|)>\varepsilon\right\}\leqslant \frac{1}{n\varepsilon^2}\sum_{t=1}^n\mathbb{E}\left[G_{1,tj}^2r_j\left(\frac{\sqrt{n}\varepsilon}{G_{1,tj}}\right)\right].
\label{Lemma3 2}
\end{align}

For notational brevity, let $C_1 = \sup_{n\geqslant 1}\max_{1\leqslant t\leqslant n}\mathbb{E}(G_{1,tj}^2)$ and $C_2 = \sup_{n\geqslant 1}\max_{1\leqslant t\leqslant n}\mathbb{E}(G_{1,tj}^{2+\delta})$ for the moment. Observe that 
\begin{align}
\mathbb{E}\left[G_{1,tj}^2r_j\left(\frac{\sqrt{n}\varepsilon}{G_{1,tj}}\right)1(G_{1,tj}\leqslant n^{1/4})\right]\leqslant \mathbb{E}\left[G_{1,tj}^2\tilde r_j\left(\frac{\sqrt{n}\varepsilon}{G_{1,tj}}\right)1(G_{1,tj}\leqslant n^{1/4})\right]\leqslant \tilde{r}_j(n^{1/4}\varepsilon)C_1.
\label{part 1}
\end{align}
In addition, since $r_j(y)\leqslant \mathbb{E}(W_{tj}^2)\leqslant 4$, we get
\begin{align}
\mathbb{E}\left[G_{1,tj}^2r_j\left(\frac{\sqrt{n}\varepsilon}{G_{1,tj}}\right)1(G_{1,tj}>n^{1/4})\right]\leqslant 4\mathbb{E}\left[G_{1,tj}^21(G_{1,tj}>n^{1/4})\right]\leqslant 4C_2n^{-\delta/4}.
\label{part 2}
\end{align}
Gathering \eqref{part 1} and \eqref{part 2} yields that
\begin{align}
\begin{split}
\mathbb{E}\left[G_{1,tj}^2r_j\left(\frac{\sqrt{n}\varepsilon}{G_{1,tj}}\right)\right] = & \mathbb{E}\left[G_{1,tj}^2r_j\left(\frac{\sqrt{n}\varepsilon}{G_{1,tj}}\right)1(G_{1,tj}\leqslant n^{1/4})\right] \\
&+ \mathbb{E}\left[G_{1,tj}^2r_j\left(\frac{\sqrt{n}\varepsilon}{G_{1,tj}}\right)1(G_{1,tj}>n^{1/4})\right]\\
\leqslant &\tilde{r}_j(n^{1/4}\varepsilon)C_1 + 4C_2n^{-\delta/4}.
\end{split}
\label{Lemma3 3}
\end{align}
Thus, by \eqref{Lemma3 2} and \eqref{Lemma3 3}, we get
\begin{align*}
\mathbb{P}\left\{n^{-1/2}\max_{1\leqslant t\leqslant n}G_{1,tj}(1 + |\varepsilon_{tj}|)>\varepsilon\right\}\leqslant \frac{\tilde{r}_j(n^{1/4}\varepsilon)C_1 + 4C_2n^{-\delta/4}}{\varepsilon^2}\rightarrow 0,\,\,\, n\rightarrow\infty,
\end{align*}
which validates \eqref{Lemma3 1}.

Finally, note that $\mathbb{P}\{\mathcal{E}_n\}\rightarrow 0$ according to \autoref{Lemma2}. Thus, combining \eqref{Lemma3 1}, for any $\eta>0$, we get
\begin{align*}
\mathbb{P}\left\{\max_{1\leqslant t \leqslant n}\left|\hat{\varepsilon}_{tj} - \varepsilon_{tj}\right|>\eta\right\} =& \mathbb{P}\left\{\max_{1\leqslant t \leqslant n}\left|\hat{\varepsilon}_{tj} - \varepsilon_{tj}\right|>\eta,\mathcal{E}_n\right\} + \mathbb{P}\left\{\max_{1\leqslant t \leqslant n}\left|\hat{\varepsilon}_{tj} - \varepsilon_{tj}\right|>\eta,\mathcal{E}_n^c\right\}\\
\leqslant & \mathbb{P}\{\mathcal{E}_n\} + \mathbb{P}\left\{ \frac{2}{\sqrt{n}}\max_{1\leqslant t\leqslant n}G_{1,tj}(1+|\varepsilon_{tj}|)\left\|\sqrt{n}(\hat{\alpha}_n - \alpha_0)\right\|>\eta\right\}\\
\rightarrow & 0,\,\,\, n\rightarrow\infty,
\end{align*}
which ends the proof.
\end{proof}

\begin{proof}[\textit{\textbf{Proof of \autoref{Th1}.}}]
With a Taylor expansion, for $k,l\geqslant 1$, we have
\begin{align}
\begin{split}
\hat{\theta}_{kl} = & \frac{1}{n}\sum_{t=1}^n\pi_k(U_t)\psi_l(t/n) + \frac{1}{n}\sum_{t=1}^n\frac{\partial \pi_k}{\partial u'}(U_t)(\hat{U}_t - U_t)\psi_l(t/n) \\
&+ \frac{1}{2n}\sum_{t=1}^n(\hat{U}_t - U_t)'\frac{\partial^2 \pi_k}{\partial u \partial u'}(\tilde U_t)(\hat{U}_t - U_t)\psi_l(t/n)\\
:=& \frac{1}{n}\sum_{t=1}^n\pi_k(U_t)\psi_l(t/n) + T_{n,kl} + R_{n,kl},
\end{split}
\label{Taylor hat theta kl}
\end{align}
where $\tilde{U}_t$ lies between $\hat{U}_t$ and $U_t$. 

By \autoref{A5}, $\partial^2 \pi_k(u)/\partial u\partial u'$ is bounded in $[0,1]^d$ and $\psi_l(\tau)$ is bounded in $[0,1]$. Thus, there exists some constant $C$, such that
\begin{align*}
|R_{n,kl}|\leqslant \frac{C}{n}\sum_{t=1}^n \left\|\hat{U}_t - U_t\right\|^2.
\end{align*}
Observe that
\begin{align}
\begin{split}
\frac{1}{n}\sum_{t=1}^n \left\|\hat{U}_t - U_t\right\|^2 = &\frac{1}{n}\sum_{t=1}^n\sum_{j=1}^d\left(\hat{F}_{nj}(\hat{\varepsilon}_{tj}) - F_j(\varepsilon_{tj})\right)^2\\
\leqslant & \frac{2}{n}\sum_{t=1}^n\sum_{j=1}^d\left(\hat{F}_{nj}(\hat{\varepsilon}_{tj}) - F_j(\hat\varepsilon_{tj})\right)^2 + \frac{2}{n}\sum_{t=1}^n\sum_{j=1}^d\left({F}_{j}(\hat{\varepsilon}_{tj}) - F_j(\varepsilon_{tj})\right)^2\\
:= & R_{n,kl,1} + R_{n,kl,2}.
\end{split}
\label{R decomposition}
\end{align}
For $R_{n,kl,1}$, by \autoref{coro1}, we have
\begin{align}
R_{n,kl,1}\leqslant 2\sum_{j=1}^d\left(\sup_{x_j\in\mathbb{R}}\left|\hat{F}_{nj}(x_j) - F_j(x_j)\right|\right)^2 = O_p(n^{-1}) = o_p(n^{-1/2}).
\label{R1}
\end{align}
For $R_{n,kl,2}$, by \autoref{A3} and \autoref{Lemma2}, we obtain
\begin{align}
R_{n,kl,2}\leqslant 2\sum_{j=1}^d\frac{M_{0j}^2}{n}\sum_{t=1}^n\left(\hat{\varepsilon}_{tj} - \varepsilon_{tj}\right)^2 = O_p(n^{-1}) = o_p(n^{-1/2}).
\label{R2}
\end{align}
Gathering \eqref{R decomposition}, \eqref{R1}, and \eqref{R2} implies that $R_{n,kl} = o_p(n^{-1/2})$.

For $T_{n,kl}$, we make the following decomposition
\begin{align*}
T_{n,kl} = & \frac{1}{n}\sum_{t=1}^n\sum_{j=1}^d\frac{\partial \pi_k}{\partial u_j}(U_t)\left(\hat{F}_{nj}(\hat{\varepsilon}_{tj}) - F_j(\varepsilon_{tj})\right)\psi_l(t/n)\\
=& \frac{1}{n}\sum_{t=1}^n\sum_{j=1}^d\frac{\partial \pi_k}{\partial u_j}(U_t)\left(\hat{F}_{nj}(\hat{\varepsilon}_{tj}) - F_{nj}(\hat\varepsilon_{tj})\right)\psi_l(t/n) \\
&+ \frac{1}{n}\sum_{t=1}^n\sum_{j=1}^d\frac{\partial \pi_k}{\partial u_j}(U_t)\left[\left({F}_{nj}(\hat{\varepsilon}_{tj}) - F_{j}(\hat\varepsilon_{tj}) \right) - \left({F}_{nj}({\varepsilon}_{tj}) - F_{j}(\varepsilon_{tj}) \right)\right]\psi_l(t/n)\\
& + \frac{1}{n}\sum_{t=1}^n\sum_{j=1}^d\frac{\partial \pi_k}{\partial u_j}(U_t)\left({F}_{j}(\hat{\varepsilon}_{tj}) - F_{j}(\varepsilon_{tj})\right)\psi_l(t/n)  + \frac{1}{n}\sum_{t=1}^n\sum_{j=1}^d\frac{\partial \pi_k}{\partial u_j}(U_t)\left({F}_{nj}({\varepsilon}_{tj}) - F_{j}(\varepsilon_{tj})\right)\psi_l(t/n) \\
:= &T_{n,kl,1} + T_{n,kl,2} + T_{n,kl,3} + T_{n,kl,4}.
\end{align*}

For $T_{n,kl,1}$, note that $\partial \pi_k(u)/\partial u_j$ is bounded in $[0,1]^d$ and $\psi_l(\tau)$ is bounded in $[0,1]$. Then, by \autoref{Lemma1}, we have
\begin{align*}
T_{n,kl,1} = \frac{1}{n}\sum_{t=1}^n\sum_{j=1}^d \frac{\partial \pi_k}{\partial u_j}(U_t)f_j(\hat{\varepsilon}_{tj})(a_j + \hat{\varepsilon}_{tj}b_j)'(\hat{\alpha}_n - \alpha_0)\psi_l(t/n) + o_p(n^{-1/2}).
\end{align*}

With a Taylor expansion, we get
\begin{align*}
&\frac{1}{n}\sum_{t=1}^n\sum_{j=1}^d \frac{\partial \pi_k}{\partial u_j}(U_t)f_j(\hat{\varepsilon}_{tj})(a_j + \hat{\varepsilon}_{tj}b_j)'(\hat{\alpha}_n - \alpha_0)\psi_l(t/n)  \\
= & \frac{1}{n}\sum_{t=1}^n\sum_{j=1}^d \frac{\partial \pi_k}{\partial u_j}(U_t)f_j({\varepsilon}_{tj})(a_j + {\varepsilon}_{tj}b_j)'(\hat{\alpha}_n - \alpha_0)\psi_l(t/n) \\
&+ \frac{1}{n}\sum_{t=1}^n\sum_{j=1}^d \frac{\partial \pi_k}{\partial u_j}(U_t)\Big[\dot f_j(\tilde{\varepsilon}_{tj})a_j + (\dot f_j(\tilde{\varepsilon}_{tj})\tilde{\varepsilon}_{tj} + f_j(\tilde{\varepsilon}_{tj}))b_j\Big]'(\hat{\alpha}_n - \alpha_0)\psi_l(t/n)(\hat{\varepsilon}_{tj} - \varepsilon_{tj}),
\end{align*}
where $\tilde{\varepsilon}_{tj}$ lies between $\varepsilon_{tj}$ and $\hat{\varepsilon}_{tj}$. For the second term of the expansion, by \autoref{A3}, \autoref{A5}, and \autoref{Lemma3}, we get
\begin{align*}
&\left|\frac{1}{n}\sum_{t=1}^n\sum_{j=1}^d \frac{\partial \pi_k}{\partial u_j}(U_t)\Big[\dot f_j(\tilde{\varepsilon}_{tj})a_j + (\dot f_j(\tilde{\varepsilon}_{tj})\tilde{\varepsilon}_{tj} + f_j(\tilde{\varepsilon}_{tj}))b_j\Big]'(\hat{\alpha}_n - \alpha_0)(\hat{\varepsilon}_{tj} - \varepsilon_{tj})\psi_l(t/n)\right|\\
\leqslant& \sum_{j=1}^d \sup_{u\in[0,1]^d}\left|\frac{\partial \pi_k}{\partial u_j}(u)\right|\sup_{\tau\in[0,1]}|\psi_l(\tau)|(M_{0j} + M_{1j})(\|a_j\| + \|b_j\|)\max_{1\leqslant t\leqslant n}\left|\hat{\varepsilon}_{tj} - \varepsilon_{tj}\right|\left\|\hat{\alpha}_n - \alpha_0\right\|\\
=&o_p(n^{-1/2}).
\end{align*}
Thus, we observe that
\begin{align}
T_{n,kl,1} =  \frac{1}{n}\sum_{t=1}^n\sum_{j=1}^d \frac{\partial \pi_k}{\partial u_j}(U_t)f_j({\varepsilon}_{tj})(a_j + {\varepsilon}_{tj}b_j)'(\hat{\alpha}_n - \alpha_0)\psi_l(t/n) + o_p(n^{-1/2}).
\label{T 1 simplified}
\end{align}

For $T_{n,kl,2}$, let $\mathbb{G}_{nj}(x)=\sqrt{n}\big(F_{nj}(x)-F_j(x)\big)$ denote the empirical process indexed by $x\in\mathbb{R}$. Since $\{\varepsilon_{tj}\}_{t=1}^n$ are i.i.d., the class $\{\mathbf{1}(\cdot\leqslant x):x\in\mathbb{R}\}$ is $F_j$-Donsker. Hence, $\mathbb{G}_{nj}$ is asymptotically stochastically equicontinuous with respect to the semimetric
\begin{align*}
\rho_j(x_1,x_2):=\left\{\mathbb{E}\left[\mathbf{1}(\varepsilon_{tj}\leqslant x_1)-\mathbf{1}(\varepsilon_{tj}\leqslant x_2)\right]^2\right\}^{1/2}=\left|F_j(x_1)-F_j(x_2)\right|^{1/2};
\end{align*}
see, for example,  van der Vaart and Wellner (1996, Theorem 1.5.7 and Example 2.5.4). That is, for every $\varepsilon>0$,
\begin{align*}
\lim_{\delta\downarrow0}\limsup_{n\rightarrow\infty}\mathbb{P}\left(\sup_{\rho_j(x_1,x_2)<\delta}\left|\mathbb{G}_{nj}(x_1)-\mathbb{G}_{nj}(x_2)\right|>\varepsilon\right)=0.
\end{align*}
Moreover, \autoref{A3} implies that
\begin{align*}
\rho_j(x_1,x_2)\leqslant M_{0j}^{1/2}|x_1-x_2|^{1/2}.
\end{align*}
Consequently, for every $\varepsilon>0$,
\begin{align}
\lim_{\delta\downarrow0}\limsup_{n\rightarrow\infty}\mathbb{P}\left(\sup_{|x_1-x_2|<\delta}\left|\mathbb{G}_{nj}(x_1)-\mathbb{G}_{nj}(x_2)\right|>\varepsilon\right)=0.
\label{SE Th 1}
\end{align}

Let $\Delta_{nj}:=\max_{1\leqslant t\leqslant n}|\hat{\varepsilon}_{tj} - \varepsilon_{tj}|$. Note that
\begin{align}
\left|\sqrt{n}T_{n,kl,2}\right|\leqslant \sup_{\tau\in[0,1]}|\psi_l(\tau)|\sum_{j=1}^d\sup_{u\in[0,1]^d}\left\|\frac{\partial \pi_k}{\partial u_j}(u)\right\|\sup_{|x_1-x_2|\leqslant \Delta_{nj}}\left|\mathbb{G}_{nj}(x_1) - \mathbb{G}_{nj}(x_2)\right|.
\label{Lemma 9-1}
\end{align}
$\forall \varepsilon>0, \delta>0$, we have
\begin{align*}
&\mathbb{P}\left\{\sup_{|x_1-x_2|\leqslant \Delta_{nj}}\left|\mathbb{G}_{nj}(x_1) - \mathbb{G}_{nj}(x_2)\right|>\varepsilon\right\}\\
\leqslant &\mathbb{P}\left\{\sup_{|x_1-x_2|\leqslant \Delta_{nj}}\left|\mathbb{G}_{nj}(x_1) - \mathbb{G}_{nj}(x_2)\right|>\varepsilon, \Delta_{nj}<\delta\right\} + \mathbb{P}\left\{\Delta_{nj}\geqslant \delta\right\}\\
\leqslant & \mathbb{P}\left\{\sup_{|x_1-x_2|<\delta}\left|\mathbb{G}_{nj}(x_1) - \mathbb{G}_{nj}(x_2)\right|>\varepsilon\right\} + \mathbb{P}\left\{\Delta_{nj}\geqslant \delta\right\}.
\end{align*}
By \autoref{Lemma3}, $\Delta_{nj} = o_p(1)$. Thus, taking $n\rightarrow \infty$ on both sides of the above inequality, we get
\begin{align*}
\limsup_{n\rightarrow\infty}\mathbb{P}\left\{\sup_{|x_1-x_2|\leqslant \Delta_{nj}}\left|\mathbb{G}_{nj}(x_1) - \mathbb{G}_{nj}(x_2)\right|>\varepsilon\right\}\leqslant\limsup_{n\rightarrow\infty}\mathbb{P}\left\{\sup_{|x_1-x_2|<\delta}\left|\mathbb{G}_{nj}(x_1) - \mathbb{G}_{nj}(x_2)\right|>\varepsilon\right\}.
\end{align*}
Then, by letting $\delta\downarrow 0$, it holds that
\begin{align*}
\limsup_{n\rightarrow\infty}\mathbb{P}\left\{\sup_{|x_1-x_2|\leqslant \Delta_{nj}}\left|\mathbb{G}_{nj}(x_1) - \mathbb{G}_{nj}(x_2)\right|>\varepsilon\right\}
\leqslant\lim_{\delta\downarrow 0}\limsup_{n\rightarrow\infty}\mathbb{P}\left\{\sup_{|x_1-x_2|<\delta}\left|\mathbb{G}_{nj}(x_1) - \mathbb{G}_{nj}(x_2)\right|>\varepsilon\right\}=0
\end{align*}
by \eqref{SE Th 1}. Then, by \eqref{Lemma 9-1}, we obtain $T_{n,kl,2} = o_p(n^{-1/2})$.

Next, we analyze $T_{n,kl,3}$. With a Taylor expansion, we derive that
\begin{align*}
T_{n,kl,3} = \frac{1}{n}\sum_{t=1}^n\sum_{j=1}^d\frac{\partial \pi_k}{\partial u_j}(U_t)f_j(\varepsilon_{tj})(\hat{\varepsilon}_{tj} - \varepsilon_{tj})\psi_l(t/n) + \frac{1}{2n}\sum_{t=1}^n\sum_{j=1}^d\frac{\partial \pi_k}{\partial u_j}(U_t)\dot f_j(\tilde\varepsilon_{tj})(\hat{\varepsilon}_{tj} - \varepsilon_{tj})^2\psi_l(t/n).
\end{align*}
For the second term, by \autoref{A3}, \autoref{A5}, and \autoref{Lemma2}, we have
\begin{align*}
&\left|\frac{1}{2n}\sum_{t=1}^n\sum_{j=1}^d\frac{\partial \pi_k}{\partial u_j}(U_t)\dot f_j(\tilde\varepsilon_{tj})(\hat{\varepsilon}_{tj} - \varepsilon_{tj})^2\psi_l(t/n)\right|\\
\leqslant &\frac{\sup_{\tau\in[0,1]}|\psi_l(\tau)|}{2}\sum_{j=1}^d\sup_{u\in[0,1]^d}\left|\frac{\partial \pi_k}{\partial u_j}(u)\right|M_{1j}\left(\frac{1}{n}\sum_{t=1}^n(\hat{\varepsilon}_{tj} - \varepsilon_{tj})^2\right)\\
=&o_p(n^{-1/2}).
\end{align*}
For the first term, we first note that
\begin{align*}
\hat{\varepsilon}_{tj} - \varepsilon_{tj} = &-\frac{1}{s_{tj}(\hat{\alpha}_n)}(a_{tj} + \varepsilon_{tj}b_{tj})'(\hat{\alpha}_n - \alpha_0) \\
&-\frac{1}{2s_{tj}(\hat{\alpha}_n)}(\hat{\alpha}_n - \alpha_0)'\left(\frac{\partial^2 m_{tj}}{\partial \alpha\partial \alpha'}(\tilde{\alpha}_n) + \varepsilon_{tj}\frac{\partial^2 s_{tj}}{\partial \alpha\partial \alpha'}(\tilde{\alpha}_n)\right)(\hat{\alpha}_n - \alpha_0)\\
=& - (a_{tj} + \varepsilon_{tj}b_{tj})'(\hat{\alpha}_n - \alpha_0) + \frac{1}{s_{tj}(\hat\alpha_n)}\frac{\partial s_{tj}}{\partial \alpha'}(\tilde{\alpha}_n)(\hat{\alpha}_n - \alpha_0)(a_{tj} + \varepsilon_{tj}b_{tj})'(\hat{\alpha}_n - \alpha_0)\\
&-\frac{1}{2s_{tj}(\hat{\alpha}_n)}(\hat{\alpha}_n - \alpha_0)'\left(\frac{\partial^2 m_{tj}}{\partial \alpha\partial \alpha'}(\tilde{\alpha}_n) + \varepsilon_{tj}\frac{\partial^2 s_{tj}}{\partial \alpha\partial \alpha'}(\tilde{\alpha}_n)\right)(\hat{\alpha}_n - \alpha_0)\\
:= &- (a_{tj} + \varepsilon_{tj}b_{tj})'(\hat{\alpha}_n - \alpha_0) + r_{ntj},
\end{align*}
where $\tilde{\alpha}_n = \alpha_0 + \theta(\hat{\alpha}_n - \alpha_0)$ for some $\theta\in[0,1]$. Let $\mathcal{E}_n$ be given in \autoref{Lemma2}. When $\mathcal{E}_n^c$ occurs, it follows that
\begin{align*}
|r_{ntj}|\leqslant (2G_{1,tj}^2 + G_{2,tj})(1+|\varepsilon_{tj}|)\left\|\hat{\alpha}_n - \alpha_0\right\|^2. 
\end{align*}
Therefore, in the event $\mathcal{E}_n^c$, we get
\begin{align*}
&\left|\frac{1}{n}\sum_{t=1}^n\sum_{j=1}^d\frac{\partial \pi_k}{\partial u_j}(U_t)f_j(\varepsilon_{tj})r_{ntj}\psi_l(t/n)\right|\\
\leqslant& \sup_{\tau\in[0,1]}|\psi_l(\tau)|\sum_{j=1}^d\sup_{u\in[0,1]^d}\left|\frac{\partial \pi_k}{\partial u_j}(u)\right|M_{0j}\left\|\hat{\alpha}_n - \alpha_0\right\|^2\left(\frac{1}{n}\sum_{t=1}^n(2G_{1,tj}^2 + G_{2,tj})\right) \\
=& o_p(n^{-1/2}) 
\end{align*}
by \autoref{pro3}. According to the proof of \autoref{Lemma2}, we have $\mathbb{P}\{\mathcal{E}_n\}\rightarrow 0$ as $n\rightarrow\infty$. Thus, for any $\varepsilon>0$,
\begin{align*}
&\mathbb{P}\left\{\left|\frac{1}{\sqrt{n}}\sum_{t=1}^n\sum_{j=1}^d\frac{\partial \pi_k}{\partial u_j}(U_t)f_j(\varepsilon_{tj})r_{ntj}\psi_l(t/n)\right|>\varepsilon\right\}\\
= & \mathbb{P}\left\{\left|\frac{1}{\sqrt{n}}\sum_{t=1}^n\sum_{j=1}^d\frac{\partial \pi_k}{\partial u_j}(U_t)f_j(\varepsilon_{tj})r_{ntj}\psi_l(t/n)\right|>\varepsilon,\mathcal{E}_n\right\}\\
&+\mathbb{P}\left\{\left|\frac{1}{\sqrt{n}}\sum_{t=1}^n\sum_{j=1}^d\frac{\partial \pi_k}{\partial u_j}(U_t)f_j(\varepsilon_{tj})r_{ntj}\psi_l(t/n)\right|>\varepsilon,\mathcal{E}_n^c\right\}\\
\leqslant & \mathbb{P}\{\mathcal{E}_n\} + \mathbb{P}\left\{\sup_{\tau\in[0,1]}|\psi_l(\tau)|\sum_{j=1}^d\sup_{u\in[0,1]^d}\left|\frac{\partial \pi_k}{\partial u_j}(u)\right|M_{0j}\left\|\hat{\alpha}_n - \alpha_0\right\|^2\frac{1}{\sqrt{n}}\sum_{t=1}^n(2G_{1,tj}^2 + G_{2,tj})>\varepsilon\right\}\\
\rightarrow & 0,\,\,\, n\rightarrow\infty,
\end{align*}
which implies that 
\begin{align*}
\frac{1}{n}\sum_{t=1}^n\sum_{j=1}^d\frac{\partial \pi_k}{\partial u_j}(U_t)f_j(\varepsilon_{tj})r_{ntj}\psi_l(t/n) = o_p(n^{-1/2}).
\end{align*}
As such, we conclude that 
\begin{align}
T_{n,kl,3} = -\frac{1}{n}\sum_{t=1}^n\sum_{j=1}^d\frac{\partial \pi_k}{\partial u_j}(U_t)f_j(\varepsilon_{tj})(a_{tj} + \varepsilon_{tj}b_{tj})'(\hat{\alpha}_n - \alpha_0)\psi_l(t/n) + o_p(n^{-1/2}).
\label{T 3 simplified}
\end{align}

Gathering \eqref{T 1 simplified} and \eqref{T 3 simplified}, we obtain
\begin{align*}
T_{n,kl,1} + T_{n,kl,3} = -\frac{1}{n}\sum_{t=1}^n\sum_{j=1}^d\frac{\partial \pi_k}{\partial u_j}(U_t)f_j(\varepsilon_{tj})[(a_{tj} - a_j) + \varepsilon_{tj}(b_{tj} - b_j)]'(\hat{\alpha}_n - \alpha_0)\psi_l(t/n) + o_p(n^{-1/2}).
\end{align*}
By \autoref{A2}, to show that $T_{n,kl,1} + T_{n,kl,3} = o_p(n^{-1/2})$, it suffices to prove that 
\begin{align*}
\frac{1}{n}\sum_{t=1}^n\sum_{j=1}^d\frac{\partial \pi_k}{\partial u_j}(U_t)f_j(\varepsilon_{tj})[(a_{tj} - a_j) + \varepsilon_{tj}(b_{tj} - b_j)]\psi_l(t/n) = o_p(1).
\end{align*}
For notational convenience, we define $W_{tj} = (\partial \pi_k(U_t)/\partial u_j)f_j(\varepsilon_{tj})$ for the moment. Under the null hypothesis, $W_{tj}$'s are i.i.d. with $\mathbb{E}(W_{tj}^2)<\infty$. Consider the quantity
\begin{align*}
\frac{1}{n}\sum_{t=1}^nW_{tj}(a_{tj} - a_j)\psi_l(t/n) = \frac{1}{n}\sum_{t=1}^n(W_{tj} - \mathbb{E}(W_{tj}))(a_{tj} - a_j)\psi_l(t/n) + \mathbb{E}(W_{1j})\frac{1}{n}\sum_{t=1}^n(a_{tj} - a_j)\psi_l(t/n). 
\end{align*}
Note that $a_{tj}$ is $\mathcal{F}_{t-1}$-measurable and $W_{tj}$ is independent of $\mathcal{F}_{t-1}$. Thus, $\mathbb{E}[(W_{tj} - \mathbb{E}(W_{tj}))(a_{tj} - a_j)|\mathcal{F}_{t-1}] = 0$. For $s<t$, we have
\begin{align*}
&\mathbb{E}[(W_{tj} - \mathbb{E}(W_{tj}))(a_{tj} - a_j)(W_{sj} - \mathbb{E}(W_{sj}))(a_{sj} - a_j)] \\
=& \mathbb{E}\{\mathbb{E}[(W_{tj} - \mathbb{E}(W_{tj}))(a_{tj} - a_j)(W_{sj} - \mathbb{E}(W_{sj}))(a_{sj} - a_j)|\mathcal{F}_{t-1}]\}\\
=& \mathbb{E}\{(W_{sj} - \mathbb{E}(W_{sj}))(a_{sj} - a_j)\mathbb{E}[(W_{tj} - \mathbb{E}(W_{tj}))(a_{tj} - a_j)|\mathcal{F}_{t-1}]\}\\
=&0.
\end{align*}
Thus, by Chebyshev's inequality, for any $\varepsilon>0$, as $n\rightarrow\infty$,
\begin{align*}
&\mathbb{P}\left\{\left\|\frac{1}{n}\sum_{t=1}^n(W_{tj} - \mathbb{E}(W_{tj}))(a_{tj} - a_j)\psi_l(t/n) \right\|>\varepsilon\right\}\\
\leqslant &\frac{2}{n\varepsilon^2}\sup_{\tau\in[0,1]}|\psi_l^2(\tau)|\mathbb{E}(W_{1j}^2)\sup_{n\geqslant 1}\max_{1\leqslant t\leqslant n}(\mathbb{E}(G_{1,tj}^2) + \|a_j\|^2)\rightarrow 0
\end{align*}
by \autoref{A3} and \autoref{A5}. Next, by \autoref{A3}, it follows that $\sup_{n\geqslant 1}\mathbb{E}\|\bar a_{nj} - a_j\|^{2+\delta}<\infty$. This uniform moment bound implies that the sequence $\{\|\bar a_{nj}-a_j\|\}_{n\geqslant1}$ is uniformly integrable. Since $\bar a_{nj}\stackrel{p}\longrightarrow a_j$, we obtain $\mathbb{E}\|\bar a_{nj} - a_j\|\rightarrow 0$ as $n\rightarrow\infty$. Observe that 
\begin{align*}
\mathbb{E}\left\|\frac{1}{n}\sum_{t=1}^n(a_{tj} - a_j)\psi_l(t/n)\right\|\leqslant |\psi_l(1)|\mathbb{E}\|\bar a_{nj} - a_j\| + \frac{\|\dot\psi_l\|_\infty}{n^2}\sum_{t=1}^{n-1}t\mathbb{E}\|\bar a_{tj} - a_j\|.
\end{align*}
For any $\eta>0$, let $N$ be the integer such that for $t\geqslant N$, $\mathbb{E}\|\bar a_{tj} - a_j\|<\eta$. Then, it follows that 
\begin{align*}
\frac{1}{n^2}\sum_{t=1}^{n-1}t\mathbb{E}\|\bar a_{tj} - a_j\| \leqslant  \frac{1}{n^2}\sum_{t=1}^{N-1}t\mathbb{E}\|\bar a_{tj} - a_j\| + \frac{\eta}{n^2}\sum_{t=N}^{n-1}t\leqslant \frac{1}{n^2}\sum_{t=1}^{N-1}t\mathbb{E}\|\bar a_{tj} - a_j\| + \frac{\eta}{2}.
\end{align*}
Therefore,
\begin{align*}
\limsup_{n\rightarrow\infty} \frac{1}{n^2}\sum_{t=1}^{n-1}t\mathbb{E}\|\bar a_{tj} - a_j\| \leqslant \frac{\eta}{2}.
\end{align*}
Since $\eta$ is arbitrary, Markov's inequality then implies that 
\begin{align*}
\frac{1}{n}\sum_{t=1}^n(a_{tj} - a_j)\psi_l(t/n) = o_p(1).
\end{align*}
Since $d$ is fixed, we then obtain
\begin{align*}
\frac{1}{n}\sum_{t=1}^n\sum_{j=1}^dW_{tj}(a_{tj} - a_j)\psi_l(t/n) = o_p(1).
\end{align*}
Using a similar argument, we can show that 
\begin{align*}
\frac{1}{n}\sum_{t=1}^n\sum_{j=1}^dW_{tj}\varepsilon_{tj}(b_{tj} - b_j)\psi_l(t/n) = o_p(1).
\end{align*}
Thus, $T_{n,kl,1} + T_{n,kl,3} = o_p(n^{-1/2})$.

For $T_{n,kl,4}$, we first define 
\begin{align*}
h_j(u) = \int_{[0,1]^d}\frac{\partial \pi_k}{\partial u_j}(v)\Big(1(u_j\leqslant v_j) - v_j\Big)\bar c(v)dv,\,\,\, u=(u_1,u_2,\ldots,u_d)'\in\mathbb{R}^d
\end{align*}
and 
\begin{align*}
h_j^o(u,v) = \frac{\partial \pi_k}{\partial u_j}(u)\Big(1(v_j\leqslant u_j) - u_j\Big) - h_j(v). 
\end{align*}
Under the null hypothesis, if $s\neq t$, we have $\mathbb{E}[h_j^o(U_t,U_s)|U_t] = \mathbb{E}[h_j^o(U_t,U_s)|U_s] = \mathbb{E}(h_j(U_t))=0$. Then, direct calculations yield that
\begin{align*}
\frac{1}{n}\sum_{t=1}^n\frac{\partial \pi_k}{\partial u_j}(U_t)(F_{nj}(\varepsilon_{tj}) - F_j(\varepsilon_{tj}))\psi_l(t/n) = &\frac{1}{n^2}\sum_{\substack{i,t=1}}^nh_j(U_i)\psi_l(t/n)\\
&+\frac{1}{n^2}\sum_{t=1}^nh_j^o(U_t,U_t)\psi_l(t/n) + \frac{1}{n^2}\sum_{\substack{i,t=1\\i\neq t}}^nh_j^o(U_t,U_i)\psi_l(t/n).
\end{align*}
Note that $\mathbb{E}(h_j^2(U_t))<\infty$. Thus, by the Central Limit Theorem (CLT), we get $n^{-1}\sum_{\substack{i=1}}^nh_j(U_i)=O_p(n^{-1/2})$. Then, for $l\geqslant 1$, we obtain
\begin{align*}
\frac{1}{n^2}\sum_{\substack{i,t=1}}^nh_j(U_i)\psi_l(t/n) = \frac{1}{n}\sum_{\substack{i=1}}^nh_j(U_i)\frac{1}{n}\sum_{t=1}^n\psi_l(t/n) = o_p(n^{-1/2}).
\end{align*}
Since $h_j^o(u,v)$ is bounded in $[0,1]^d\times [0,1]^d$, we get
\begin{align*}
\left|\frac{1}{n^2}\sum_{t=1}^nh_j^o(U_t,U_t)\psi_l(t/n) \right|\leqslant \frac{C}{n} = o_p(n^{-1/2}).
\end{align*}
Observe that $\mathbb{E}[h_j^o(U_{t_1},U_{s_1})h_j^o(U_{t_2},U_{s_2})]\neq 0$ only if $(t_1,s_1) = (t_2,s_2)$ or $(t_1,s_1) = (s_2,t_2)$. Thus, we get
\begin{align*}
\mathbb{E}\left(\frac{1}{n^2}\sum_{\substack{i,t=1\\i\neq t}}^nh_j^o(U_t,U_i)\psi_l(t/n)\right)^2\leqslant \frac{C}{n^2}
\end{align*}
for some constant $C$. Then, by Chebyshev's inequality, we get
\begin{align*}
\frac{1}{n^2}\sum_{\substack{i,t=1\\i\neq t}}^nh_j^o(U_t,U_i)\psi_l(t/n) = o_p(n^{-1/2}).
\end{align*}
Therefore, $T_{n,kl,4} = o_p(n^{-1/2})$. Since $T_{n,kl,2}=o_p(n^{-1/2})$, $T_{n,kl,4}=o_p(n^{-1/2})$, and $T_{n,kl,1} + T_{n,kl,3}=o_p(n^{-1/2})$, we conclude that $T_{n,kl} = o_p(n^{-1/2})$. Note that under the null hypothesis, $R_{n,kl} = o_p(n^{-1/2})$ and $T_{n,kl} = o_p(n^{-1/2})$. Then, \eqref{hat theta kl expansion} follows directly.

By \eqref{hat theta kl expansion}, it follows that
\begin{align*}
\hat{\theta}_{kl} = \frac{1}{n}\sum_{t=1}^n\Big[\pi_k(U_t) - \mathbb{E}(\pi_k(U_t))\Big]\psi_l(t/n) + \mathbb{E}(\pi_k(U_1))\frac{1}{n}\sum_{t=1}^n\psi_l(t/n) + o_p(n^{-1/2}).
\end{align*}
By \autoref{A5}, for $l\geqslant 1$, we have
\begin{align}
\left|\frac{1}{n}\sum_{t=1}^n\psi_l(t/n)\right| = &\left|\frac{1}{n}\sum_{t=1}^n\psi_l(t/n) - \int_0^1 \psi_l(\tau)d\tau\right| = \left|\sum_{t=1}^n\int_{\frac{t-1}{n}}^{\frac{t}{n}}(\psi_l(t/n) - \psi_l(\tau))d\tau\right|\leqslant \frac{\|\dot \psi_l\|_\infty}{n}.
\label{bound sum psi}
\end{align}
Therefore, under the null hypothesis,
\begin{align*}
\hat{\theta}_{kl} = \frac{1}{n}\sum_{t=1}^n\Big[\pi_k(U_t) - \mathbb{E}(\pi_k(U_t))\Big]\psi_l(t/n) + o_p(n^{-1/2}).
\end{align*}
Consequently, 
\begin{align*}
\sqrt{n}\hat{\theta}_n(K,L) = \frac{1}{\sqrt{n}}\sum_{t=1}^n\psi(t/n)\otimes [\pi(U_t) - \mathbb{E}(\pi(U_t))]  + o_p(1),
\end{align*}
where $\psi(\tau) = (\psi_1(\tau),\psi_2(\tau),\ldots,\psi_L(\tau))'$. For notational brevity, define $Z_{nt} = n^{-1/2}\psi(t/n)\otimes [\pi(U_t) - \mathbb{E}(\pi(U_t))]$ for the moment. Note that $Z_{nt}$'s are mutually independent with $\mathbb{E}(Z_{nt}) = 0_{KL}$. In addition,
\begin{align*}
\sum_{t=1}^n\mathbb{E}(Z_{nt}Z_{nt}') = \frac{1}{n}\sum_{t=1}^n[\psi(t/n)\psi'(t/n)]\otimes \Sigma_K\rightarrow I_L\otimes \Sigma_K.
\end{align*}
Furthermore, for some constant $C$, we have $\|Z_{nt}\|\leqslant Cn^{-1/2} $ for $1\leqslant t\leqslant n$. Therefore, for any $\varepsilon>0$, 
\begin{align*}
\sum_{t=1}^n\mathbb{E}\Big[\|Z_{nt}\|^21(\|Z_{nt}\|>\varepsilon)\Big]\rightarrow 0,\,\,\,n\rightarrow\infty.
\end{align*}
Thus, by Proposition 2.27 of van der Vaart (1998) and Slutsky's Theorem, we obtain $\sqrt{n}\hat{\theta}_n(K,L) = \sum_{t=1}^nZ_{nt}+o_p(1)\stackrel{d}\longrightarrow \mathcal{N}_{KL}(0_{KL},I_L\otimes \Sigma_K)$, which ends the proof.
\end{proof}

\begin{proof}[\textit{\textbf{Proof of \autoref{Th2}.}}]
We first establish an important result that ensures $\hat{\Sigma}_{n,K}\stackrel{p}\longrightarrow \Sigma_K$ regardless of whether the null hypothesis holds.

\begin{proposition}
\label{pro21}
Suppose Assumptions \ref{A1}--\ref{A5} hold. Then, as $n\rightarrow\infty$, $\hat{\Sigma}_{n,K}\stackrel{p}\longrightarrow \Sigma_K$.
\end{proposition}

\begin{proof}[\textit{\textbf{Proof of \autoref{pro21}.}}]
Define 
\begin{align*}
M_K = \int_{[0,1]^d}\pi(u)\pi'(u)\bar{c}(u)du\,\,\,\text{and}\,\,\,\hat{M}_{n,K} = \frac{1}{n}\sum_{t=1}^n\pi(\hat{U}_t)\pi'(\hat{U}_t).
\end{align*}
Observe that $\Sigma_K = M_K - \mu_K\mu_K'$ and $\hat{\Sigma}_{n,K} = \hat{M}_{n,K} - \hat{\mu}_{n,K}\hat{\mu}_{n,K}'$. Then, using the Continuous Mapping Theorem (CMT), it suffices to show that $\hat{\mu}_{n,K}\stackrel{p}\longrightarrow \mu_K$ and $\hat{M}_{n,K}\stackrel{p}\longrightarrow M_K$.

We make the following decomposition
\begin{align*}
\hat{\mu}_{n,K} - \mu_K = \frac{1}{n}\sum_{t=1}^n\left(\pi(\hat{U}_t) - \pi(U_t)\right) + \frac{1}{n}\sum_{t=1}^n\Big(\pi(U_t) - \mathbb{E}(\pi(U_t))\Big) + \frac{1}{n}\sum_{t=1}^n\Big(\mathbb{E}(\pi(U_t)) - \mu_K\Big).
\end{align*}
For the first term, by \autoref{A5} and \eqref{R decomposition}--\eqref{R2}, we have
\begin{align*}
\left\|\frac{1}{n}\sum_{t=1}^n\left(\pi(\hat{U}_t) - \pi(U_t)\right)\right\|\leqslant \frac{C}{n}\sum_{t=1}^n\left\|\hat{U}_t - U_t\right\|\leqslant C\sqrt{\frac{1}{n}\sum_{t=1}^n\left\|\hat{U}_t - U_t\right\|^2} = o_p(1).
\end{align*}
For the second term, since $U_t$'s are independent and $\|\pi(u)\|$ is bounded on $[0,1]^d$, we have
\begin{align*}
\mathbb{E}\left\|\frac{1}{n}\sum_{t=1}^n\Big(\pi(U_t) - \mathbb{E}(\pi(U_t))\Big)\right\|^2\leqslant \frac{1}{n^2}\sum_{t=1}^n\mathbb{E}\left\|\pi(U_t)\right\|^2 = O(n^{-1}).
\end{align*}
Thus, using Chebyshev's inequality, the second term is of order $o_p(1)$. For the third term, we have
\begin{align*}
\frac{1}{n}\sum_{t=1}^n\mathbb{E}(\pi(U_t)) = \frac{1}{n}\sum_{t=1}^n\int_{[0,1]^d}\pi(u)c(u;t/n)du\longrightarrow\int_0^1\int_{[0,1]^d}\pi(u)c(u;\tau)dud\tau = \int_{[0,1]^d}\pi(u)\bar{c}(u)du = \mu_K.
\end{align*}
As such, we obtain $\hat{\mu}_{n,K}\stackrel{p}\longrightarrow \mu_K$. The proof for $\hat{M}_{n,K}\stackrel{p}\longrightarrow M_K$ is similar and is omitted for brevity.
\end{proof}

Gathering \autoref{Th1}, \autoref{pro21}, and Slutsky's Theorem, we have the asymptotic normality:  $\sqrt{n}\hat{\Pi}_n(K,L)\stackrel{d}\longrightarrow \mathcal{N}_{KL}(0_{KL},I_{KL})$. Finally, the CMT indicates that
\begin{align*}
n\hat{\Psi}_n^2(K,L) = \left\|\sqrt{n}\hat{\Pi}_n(K,L)\right\|^2\stackrel{d}\longrightarrow \chi_{KL}^2,
\end{align*}
which yields the desired result.
\end{proof}

\begin{proof}[\textit{\textbf{Proof of \autoref{Th3}.}}]
Observe that 
\begin{align*}
\hat{\theta}_{kl} - \theta_{kl} = &\frac{1}{n}\sum_{t=1}^n\left(\pi_k(\hat{U}_t) - \pi_k(U_t)\right)\psi_l(t/n) + \frac{1}{n}\sum_{t=1}^n\Big(\pi_k(U_t) - \mathbb{E}(\pi_k(U_t))\Big)\psi_l(t/n)\\
&+\left(\frac{1}{n}\sum_{t=1}^n\mathbb{E}(\pi_k(U_t))\psi_l(t/n) - \theta_{kl}\right).
\end{align*}
For the first term, by \autoref{A5} and \eqref{R decomposition}--\eqref{R2}, we have
\begin{align*}
\left|\frac{1}{n}\sum_{t=1}^n\left(\pi_k(\hat{U}_t) - \pi_k(U_t)\right)\psi_l(t/n)\right|\leqslant \frac{C}{n}\sum_{t=1}^n\left\|\hat{U}_t - U_t\right\|\leqslant C\sqrt{\frac{1}{n}\sum_{t=1}^n\left\|\hat{U}_t - U_t\right\|^2} = o_p(1).
\end{align*}
Note that $U_t$'s are independent, $\pi_k(u)$ is bounded on $[0,1]^d$ and $\psi_l(\tau)$ is bounded on $[0,1]$. Thus, for the second term, we have
\begin{align*}
\mathbb{E}\left|\frac{1}{n}\sum_{t=1}^n\Big(\pi_k(U_t) - \mathbb{E}(\pi_k(U_t))\Big)\psi_l(t/n)\right|^2\leqslant \frac{\|\psi_l\|_\infty^2}{n^2}\sum_{t=1}^n\mathbb{E}(\pi_k^2(U_t)) = O(n^{-1}).
\end{align*}
Then, using Chebyshev's inequality, the second term is of order $o_p(1)$. For the third term, we have
\begin{align*}
\frac{1}{n}\sum_{t=1}^n\mathbb{E}(\pi_k(U_t))\psi_l(t/n) =& \frac{1}{n}\sum_{t=1}^n\int_{[0,1]^d}c(u;t/n)\pi_k(u)\psi_l(t/n)du
\longrightarrow\int_0^1\int_{[0,1]^d}c(u;\tau)\pi_k(u)\psi_l(\tau)dud\tau = \theta_{kl}.
\end{align*}
Therefore, we obtain $\hat{\theta}_{kl}\stackrel{p}\longrightarrow \theta_{kl}$. By \autoref{pro21} and the CMT, it follows that 
\begin{align*}
\hat{\Psi}_n^2(K,L) = \sum_{l=1}^L \hat{\theta}_{n,l}'\hat{\Sigma}_{n,K}^{-1}\hat{\theta}_{n,l}\stackrel{p}\longrightarrow \sum_{l=1}^L\theta_l'\Sigma_K^{-1}\theta_l = \Psi^2(K,L).
\end{align*}
Under $\mathcal{H}_{1a}$ in \eqref{adjH_1}, there exists one $\theta_l\neq 0_K$. Since $\Sigma_K$ is positive definite, we then obtain $\Psi^2(K,L)>0$.
\end{proof}

\begin{proof}[\textit{\textbf{Proof of \autoref{Th4}.}}]
For notational convenience, let $\mathbb{P}_n$ and $\mathbb{Q}_n$ denote the probability measure under $\mathcal{H}_0$ and $\mathcal{H}_{1L}$, respectively. Let $\mathbb{E}_{\mathbb{P}_n}(\cdot)$ and $\mathbb{E}_{\mathbb{Q}_n}(\cdot)$ denote the expectation with respect to $\mathbb{P}_n$ and $\mathbb{Q}_n$. We first establish an important proposition that greatly facilitates our proof. 

\begin{proposition}
\label{pro41}
Suppose Assumptions \ref{A1}--\ref{A6} hold. Let $\{\mathscr{B}_n\}_{n\geqslant 1}$ be a sequence of events such that $\mathbb{P}_n(\mathscr{B}_n)\rightarrow 0$ as $n\rightarrow\infty$. Then, as $n\rightarrow \infty$, $\mathbb{Q}_n(\mathscr{B}_n)\rightarrow 0$.
\end{proposition}

\begin{proof}[\textit{\textbf{Proof of \autoref{pro41}.}}]
Let $L_n$ denote the likelihood ratio
\begin{align*}
L_n = \prod_{t=1}^n\frac{c_n(U_t;t/n)}{\bar{c}(U_t)} = \prod_{t=1}^n\left(1 + n^{-1/2}\zeta(U_t;t/n)\right).
\end{align*}
Under \autoref{A6}(iii), for $t=1,2,\ldots,n$, we have
\begin{align*}
\mathbb{E}_{\mathbb{P}_n}(\zeta(U_t;t/n)) = \int_{[0,1]^d}r(u;t/n)du = 0.
\end{align*}
In addition, \autoref{A6}(i) implies that $\sup_{(u,\tau)\in [0,1]^d\times [0,1]}|\zeta(u;\tau)|\leqslant M$ for some $M>0$. Thus, it follows that
\begin{align*}
\mathbb{E}_{\mathbb{P}_n}(L_n^2) = \prod_{t=1}^n\left(1+ \frac{1}{n}\mathbb{E}_{\mathbb{P}_n}(\zeta^2(U_t;t/n))\right)\leqslant \exp\left(\frac{1}{n}\sum_{t=1}^n\mathbb{E}_{\mathbb{P}_n}(\zeta^2(U_t;t/n))\right)\leqslant \exp\left(M^2\right)<\infty.
\end{align*}
Using the Cauchy–Schwarz inequality, we have
\begin{align*}
\mathbb{Q}_n(\mathscr{B}_n) = \mathbb{E}_{\mathbb{Q}_n}(1(\mathscr{B}_n)) = \mathbb{E}_{\mathbb{P}_n}(L_n1(\mathscr{B}_n))\leqslant \mathbb{E}_{\mathbb{P}_n}^{1/2}(L_n^2)\mathbb{E}_{\mathbb{P}_n}^{1/2}(1(\mathscr{B}_n)) = \mathbb{E}_{\mathbb{P}_n}^{1/2}(L_n^2)\mathbb{P}_n^{1/2}(\mathscr{B}_n) \rightarrow 0
\end{align*}
as $n\rightarrow\infty$, which yields the desired result.
\end{proof}

By \autoref{pro41}, we immediately derive that \eqref{hat theta kl expansion} holds under $\mathcal{H}_{1L}$. Observe that under $\mathcal{H}_{1L}$,
\begin{align*}
\frac{1}{\sqrt{n}}\sum_{t=1}^n\mathbb{E}(\pi_k(U_t))\psi_l(t/n) = \int_{[0,1]^d}\pi_k(u)\bar{c}(u)du\frac{1}{\sqrt{n}}\sum_{t=1}^n\psi_l(t/n) + \frac{1}{n}\sum_{t=1}^n\int_{[0,1]^d}r(u;t/n)\pi_k(u)\psi_l(t/n)du.
\end{align*}
Since 
\begin{align*}
\frac{1}{n}\sum_{t=1}^n\int_{[0,1]^d}r(u;t/n)\pi_k(u)\psi_l(t/n)du\rightarrow \int_0^1\int_{[0,1]^d}r(u;\tau)\pi_k(u)\psi_l(\tau)dud\tau = \Delta_{kl}
\end{align*}
and $n^{-1/2}\sum_{t=1}^n\psi_l(t/n) = o(1)$ by \eqref{bound sum psi}, we obtain
\begin{align}
\frac{1}{\sqrt{n}}\sum_{t=1}^n\mathbb{E}(\pi_k(U_t))\psi_l(t/n) = \Delta_{kl} + o(1).
\label{noncenter}
\end{align}

Define $\Delta = (\Delta_1',\Delta_2',\ldots,\Delta_L')'$ and $Z_{nt} = n^{-1/2}\psi(t/n)\otimes [\pi(U_t) - \mathbb{E}(\pi(U_t))]$. \eqref{hat theta kl expansion} and \eqref{noncenter} imply that 
\begin{align*}
\sqrt{n}\hat{\theta}_n(K,L) = \Delta + \sum_{t=1}^nZ_{nt} + o_p(1).
\end{align*}
Note that $Z_{nt}$'s are mutually independent with $\mathbb{E}(Z_{nt}) = 0_{KL}$. In addition, under $\mathcal{H}_{1L}$,
\begin{align*}
\sum_{t=1}^n\mathbb{E}(Z_{nt}Z_{nt}') = \frac{1}{n}\sum_{t=1}^n[\psi(t/n)\psi'(t/n)]\otimes (\Sigma_K + o(1))\rightarrow I_L\otimes \Sigma_K,
\end{align*}
where the $o(1)$ is uniform in $1\leqslant t\leqslant n$. Furthermore, for some constant $C$, we have $\|Z_{nt}\|\leqslant Cn^{-1/2} $ for $1\leqslant t\leqslant n$. Therefore, for any $\varepsilon>0$, 
\begin{align*}
\sum_{t=1}^n\mathbb{E}\Big[\|Z_{nt}\|^21(\|Z_{nt}\|>\varepsilon)\Big]\rightarrow 0,\,\,\,n\rightarrow\infty.
\end{align*}
Thus, by Proposition 2.27 of van der Vaart (1998) and Slutsky's Theorem, we obtain $\sqrt{n}\hat{\theta}_n(K,L) = \Delta + \sum_{t=1}^nZ_{nt}+o_p(1)\stackrel{d}\longrightarrow \mathcal{N}_{KL}(\Delta,I_L\otimes \Sigma_K)$. Note that under $\mathcal{H}_{1L}$, $\hat{\Sigma}_{n,K}\stackrel{p}\longrightarrow \Sigma_K$ holds by \autoref{pro21} and \autoref{pro41}. Thus, by the CMT, we get
\begin{align*}
n\hat{\Psi}_n^2(K,L) = (\sqrt{n}\hat{\theta}_n(K,L))'(I_L\otimes \hat{\Sigma}_{n,K}^{-1})(\sqrt{n}\hat{\theta}_n(K,L))\stackrel{d}\longrightarrow \chi_{KL}^2\left(\sum_{l=1}^L\Delta_l'\Sigma_K^{-1}\Delta_l\right),
\end{align*}
which completes the proof.
\end{proof}

\begin{proof}[\textit{\textbf{Proof of \autoref{ThD}.}}]
Define $\mathscr{I}(P) = \{\mathbf{k}\in\mathbb{N}_0^d:\|\mathbf{k}\|_0\geqslant 2,\|\mathbf{k}\|_1\leqslant P\}$, $\mathscr{A}(P,L) = \mathscr{I}(P)\times \{1,2,\ldots,L\}$, and $\mathscr{A}_0=\mathscr{A}(2,1)$. Let $\# A$ denote the number of elements of set $A$. Obviously, $\# \mathscr{A}(P,L) = K(P)L$ and $\#\mathscr{A}_0 = d(d-1)/2$. In addition, for any $P\geqslant 2$ and $L\geqslant 1$, $\mathscr{A}_0\subset \mathscr{A}(P,L)$.

We first show that, under the null, $\mathbb{P}\{(\hat{P},\hat{L})\neq (2,1)\}\rightarrow 0$ as $n\rightarrow\infty$. For notational brevity, let 
\begin{align*}
SC_n(P,L) = \sum_{(\mathbf{k},l)\in \mathscr{A}(P,L)}\hat{Z}_{\mathbf{k}l}^2 - K(P)L\log n,\,\,\,\text{where}\,\,\,\hat{Z}_{\mathbf{k}l} = \frac{1}{\sqrt{n}}\sum_{t=1}^n\pi_\mathbf{k}(\hat{U}_t)\psi_l(t/n).
\end{align*}
Let $\mathscr{A}_n = \mathscr{A}(D_P(n),D_L(n))$. If $(\hat{P},\hat{L})\neq (2,1)$, there must exist $2\leqslant P \leqslant D_P(n)$ and $1\leqslant L \leqslant D_L(n)$, such that $SC_n(P,L) - SC_n(2,1)> 0$, which is equivalent to 
\begin{align*}
\sum_{(\mathbf{k},l)\in \mathscr{A}(P,L)\backslash \mathscr{A}_0}\hat{Z}_{\mathbf{k}l}^2> (K(P)L - K(2))\log n.
\end{align*}
Note that the left-hand side of the inequality is a summation over $K(P)L - K(2)$ items. Thus, the above inequality implies that $\max_{(\mathbf{k},l)\in \mathscr{A}(P,L)\backslash \mathscr{A}_0}|\hat{Z}_{\mathbf{k}l}|>\sqrt{\log n}$. As such, under the null hypothesis, we have
\begin{align}
\mathbb{P}\left\{(\hat{P},\hat{L})\neq (2,1)\right\}\leqslant \mathbb{P}\left\{\max_{(\mathbf{k},l)\in \mathscr{A}_n\backslash \mathscr{A}_0}|\hat{Z}_{\mathbf{k}l}|>\sqrt{\log n}\right\}\leqslant \mathbb{P}\left\{\max_{(\mathbf{k},l)\in \mathscr{A}_n}|\hat{Z}_{\mathbf{k}l}|>\sqrt{\log n}\right\}.
\label{bound selection}
\end{align}
Define 
\begin{align*}
\tilde{Z}_{\mathbf{k}l} = \frac{1}{\sqrt{n}}\sum_{t=1}^n\pi_{\mathbf{k}}(U_t)\psi_l(t/n),\,\,\,(\mathbf{k},l)\in\mathscr{A}_n.
\end{align*}
Next, we prove our desired results using the following propositions.

\begin{proposition}
\label{proD0}
Under Assumptions \ref{A5} and \ref{AD1}, \eqref{rate DP DL} implies that $D_P^{\rho_0}(n)D_L^{r_{\psi,0}}(n) = o(\sqrt{n/\log n})$, $D_P^{\rho_2}(n)D_L^{r_{\psi,0}}(n) = o(\sqrt{n\log n})$, $D_P^{\rho_0}(n)D_L^{r_{\psi,1}}(n) = o(\sqrt{n\log n})$, and $D_L^{r_{\psi,0} + r_{\psi,1}}(n) = O(n)$, where $\rho_2 = \max\{2r_{\varphi,1} + (d-2)r_{\varphi,0},r_{\varphi,2} + (d-1)r_{\varphi,0}\}$.
\end{proposition}

\begin{proof}[\textit{\textbf{Proof of \autoref{proD0}.}}] The proof follows directly from simple algebraic calculations and is thus omitted.
\end{proof}

\begin{proposition}
\label{proD1}
Suppose Assumptions \ref{A1}--\ref{A5} and \ref{AD1} hold. Then, under $\mathcal{H}_0$ in \eqref{oriH_0}, as $n\rightarrow\infty$,
\begin{align*}
\max_{(\mathbf{k},l)\in\mathscr{A}_n}\left|\hat{Z}_{\mathbf{k}l} - \tilde{Z}_{\mathbf{k}l}\right| = o_p(\sqrt{\log n}).
\end{align*}
\end{proposition}

\begin{proof}[\textit{\textbf{Proof of \autoref{proD1}.}}]
Gathering \eqref{R decomposition}, \eqref{R1}, and \eqref{R2}, we obtain
\begin{align}
\frac{1}{n}\sum_{t=1}^n\left\|\hat{U}_t - U_t\right\|^2 = O_p(n^{-1})
\label{bound U2}
\end{align}
and 
\begin{align}
\frac{1}{n}\sum_{t=1}^n\left\|\hat{U}_t - U_t\right\| \leqslant \sqrt{\frac{1}{n}\sum_{t=1}^n\left\|\hat{U}_t - U_t\right\|^2}  = O_p(n^{-1/2}).
\label{bound U1}
\end{align}
Using the Taylor expansion in \eqref{Taylor hat theta kl}, we have
\begin{align*}
\hat{Z}_{\mathbf{k}l} - \tilde{Z}_{\mathbf{k}l} = \frac{1}{\sqrt{n}}\sum_{t=1}^n\frac{\partial \pi_\mathbf{k}}{\partial u'}(U_t)(\hat{U}_t - U_t)\psi_l(t/n) + \frac{1}{2\sqrt{n}}\sum_{t=1}^n(\hat{U}_t - U_t)'\frac{\partial^2 \pi_\mathbf{k}}{\partial u \partial u'}(\tilde U_t)(\hat{U}_t - U_t)\psi_l(t/n).
\end{align*}
Define 
\begin{align*}
B_{\pi,1}(P) = \max_{\mathbf{k}\in\mathscr{I}(P)}\sup_{u\in[0,1]^d}\left\|\frac{\partial \pi_\mathbf{k}}{\partial u}(u)\right\|,\,\,\, B_{\pi,2}(P) = \max_{\mathbf{k}\in\mathscr{I}(P)}\sup_{u\in[0,1]^d}\left\|\frac{\partial^2 \pi_\mathbf{k}}{\partial u\partial u'}(u)\right\|,\,\,\, B_{\psi,0}(L) = \max_{1\leqslant l\leqslant L}\|\psi_l\|_\infty.
\end{align*}
Then, gathering \eqref{bound U1} and \eqref{bound U2}, we get
\begin{align*}
\max_{(\mathbf{k},l)\in\mathscr{A}_n}\left|\hat{Z}_{\mathbf{k}l} - \tilde{Z}_{\mathbf{k}l}\right| = O_p\left(B_{\pi,1}(D_P(n))B_{\psi,0}(D_L(n)) + n^{-1/2}B_{\pi,2}(D_P(n))B_{\psi,0}(D_L(n))\right).
\end{align*}
Under \autoref{AD1}(ii), we have $B_{\pi,1}(D_P(n)) = O(D_P^{\rho_1}(n))$, $B_{\pi,2}(D_P(n)) = O(D_P^{\rho_2}(n))$, and $B_{\psi,0}(D_L(n)) = O(D_L^{r_{\psi,0}}(n))$. Thus, it follows that
\begin{align*}
\max_{(\mathbf{k},l)\in\mathscr{A}_n}\left|\hat{Z}_{\mathbf{k}l} - \tilde{Z}_{\mathbf{k}l}\right| = O_p\left(D_P^{\rho_1}(n)D_L^{r_{\psi,0}}(n) + n^{-1/2}D_P^{\rho_2}(n)D_L^{r_{\psi,0}}(n)\right) = o_p(\sqrt{\log n})
\end{align*}
by \autoref{AD1}(iii) and \autoref{proD0}, which completes the proof of \autoref{proD1}.
\end{proof}

Let $\mu_\mathbf{k} = \mathbb{E}(\pi_\mathbf{k}(U_t))$ and $\xi_{nt,\mathbf{k},l} = n^{-1/2}(\pi_\mathbf{k}(U_t) - \mu_\mathbf{k})\psi_l(t/n)$. We decompose $\tilde{Z}_{\mathbf{k}l}$ as
\begin{align*}
\tilde{Z}_{\mathbf{k}l} = \sum_{t=1}^n\xi_{nt,\mathbf{k},l} + \frac{\mu_\mathbf{k}}{\sqrt{n}}\sum_{t=1}^n\psi_l(t/n).
\end{align*}

\begin{proposition}
\label{proD2}
Suppose Assumptions \ref{A1}--\ref{A5} and \ref{AD1} hold. Then, under $\mathcal{H}_0$ in \eqref{oriH_0}, as $n\rightarrow\infty$,
\begin{align*}
\max_{(\mathbf{k},l)\in\mathscr{A}_n}\left|\frac{\mu_\mathbf{k}}{\sqrt{n}}\sum_{t=1}^n\psi_l(t/n)\right| = o_p(\sqrt{\log n}).
\end{align*}
\end{proposition}
\begin{proof}[\textit{\textbf{Proof of \autoref{proD2}.}}]
Note that
\begin{align*}
|\mu_{\mathbf{k}}|
=\left|\int_{[0,1]^d}\pi_{\mathbf{k}}(u)\bar c(u)\,du\right|
\leqslant B_{\pi,0}(D_P(n)),\,\,\,\mathbf{k}\in\mathscr{I}(D_P(n)).
\end{align*}
Let 
\begin{align*}
B_{\pi,0}(P) = \max_{\mathbf{k}\in\mathscr{I}(P)}\sup_{u\in[0,1]^d}|\pi_\mathbf{k}(u)|,\,\,\, B_{\psi,1}(L) = \max_{1\leqslant l\leqslant L}\|\dot\psi_l\|_\infty.
\end{align*}
Obviously, as $n\rightarrow\infty$, $B_{\pi,0}(D_P(n)) = O(D_P^{\rho_0}(n))$ and $B_{\psi,1}(D_L(n)) = O(D_L^{r_{\psi,1}}(n))$. Then, gathering \eqref{bound sum psi} and \autoref{proD0}, we have
\begin{align*}
\max_{(\mathbf{k},l)\in\mathscr{A}_n}\left|\frac{\mu_\mathbf{k}}{\sqrt{n}}\sum_{t=1}^n\psi_l(t/n)\right|\leqslant \frac{B_{\pi,0}(D_P(n))B_{\psi,1}(D_L(n))}{\sqrt{n}} = O_p\left(n^{-1/2}D_P^{\rho_0}(n)D_L^{r_{\psi,1}}(n)\right) = o_p(\sqrt{\log n}),
\end{align*}
which yields the desired results.
\end{proof}

\begin{proposition}
\label{proD3}
Suppose Assumptions \ref{A1}--\ref{A5} and \ref{AD1} hold. Then, under $\mathcal{H}_0$ in \eqref{oriH_0}, as $n\rightarrow\infty$,
\begin{align*}
\max_{(\mathbf{k},l)\in\mathscr{A}_n}\left|\sum_{t=1}^n\xi_{nt,\mathbf{k},l}\right| = o_p(\sqrt{\log n}).
\end{align*}
\end{proposition}

\begin{proof}[\textit{\textbf{Proof of \autoref{proD3}.}}]
Note that $\mathbb{E}(\xi_{nt,\mathbf{k},l}) = 0$ and $|\xi_{nt,\mathbf{k},l}|\leqslant 2n^{-1/2}B_{\pi,0}(D_P(n))B_{\psi,0}(D_L(n))$ almost surely. As $n\rightarrow\infty$, we have $B_{\pi,0}(D_P(n)) = O(D_P^{\rho_0}(n))$. Using the same technique in proving \eqref{bound sum psi}, we have
\begin{align*}
\left|\frac{1}{n}\sum_{t=1}^n\psi^2_l(t/n) -1\right|\leqslant \frac{2\|\psi_l\dot\psi_l\|_\infty}{n}\leqslant \frac{2B_{\psi,0}(D_L(n))B_{\psi,1}(D_L(n))}{n}.
\end{align*}
Under the null hypothesis, $\xi_{nt,\mathbf{k},l}$ are mutually independent, and 
\begin{align*}
\sum_{t=1}^n\mathbb{E}(\xi_{nt,\mathbf{k},l}^2) \leqslant \frac{\mathbb{E}(\pi_\mathbf{k}^2(U_t))}{n}\sum_{t=1}^n\psi^2_l(t/n)\leqslant B_{\pi,0}^2(D_P(n))\left(1 + \frac{2B_{\psi,0}(D_L(n))B_{\psi,1}(D_L(n))}{n}\right),\,\,\,(\mathbf{k},l)\in\mathscr{A}_n.
\end{align*}
By \autoref{proD0}, we have $B_{\psi,0}(D_L(n))B_{\psi,1}(D_L(n)) = O(D_L^{r_{\psi,0}+r_{\psi,1}}(n)) = O(n)$. Thus, there exists some positive constant $V_0$ such that 
\begin{align*}
\sum_{t=1}^n\mathbb{E}(\xi_{nt,\mathbf{k},l}^2)\leqslant V_0B_{\pi,0}^2(D_P(n))\,\,\,\text{for all}\,\,\,n\geqslant 1.
\end{align*}

Using Bernstein's inequality (Pollard, 2012, p.193), for any $\varepsilon>0$, we have
\begin{align*}
\mathbb{P}\left\{\left|\sum_{t=1}^n\xi_{nt,\mathbf{k},l}\right|>\varepsilon\sqrt{\log n}\right\}\leqslant 2\exp\left(-\frac{\varepsilon^2\log n}{2\left(V_0B_{\pi,0}^2(D_P(n)) + \frac{2\varepsilon}{3}\sqrt{\frac{\log n}{n}}B_{\pi,0}(D_P(n))B_{\psi,0}(D_L(n))\right)}\right)
\end{align*}
for all $(\mathbf{k},l)\in\mathscr{A}_n$. By \autoref{proD0}, we have
\begin{align*}
\sqrt{\frac{\log n}{n}}B_{\pi,0}(D_P(n))B_{\psi,0}(D_L(n)) = O\left(\sqrt\frac{\log n}{n}D_P^{\rho_0}(n)D_L^{r_{\psi,0}}(n)\right) = o(1).
\end{align*}
Hence, there exists $c_\varepsilon>0$, such that for sufficiently large $n$,
\begin{align*}
\mathbb{P}\left\{\left|\sum_{t=1}^n\xi_{nt,\mathbf{k},l}\right|>\varepsilon\sqrt{\log n}\right\}\leqslant 2\exp\left(-\frac{c_\varepsilon\log n}{B_{\pi,0}^2(D_P(n))}\right),\,\,\,(\mathbf{k},l)\in\mathscr{A}_n.
\end{align*}
Consequently, we obtain
\begin{align*}
\mathbb{P}\left\{\max_{(\mathbf{k},l)\in\mathscr{A}_n}\left|\sum_{t=1}^n\xi_{nt,\mathbf{k},l}\right|>\varepsilon\sqrt{\log n}\right\}\leqslant &\sum_{(\mathbf{k},l)\in\mathscr{A}_n}\mathbb{P}\left\{\left|\sum_{t=1}^n\xi_{nt,\mathbf{k},l}\right|>\varepsilon\sqrt{\log n}\right\}\\
\leqslant &2\exp\left(\log [K(D_P(n))D_L(n)]-\frac{c_\varepsilon\log n}{B_{\pi,0}^2(D_P(n))}\right).
\end{align*}
By \autoref{AD1}(iii), it follows that
\begin{align*}
B_{\pi,0}^2(D_P(n))\log [K(D_P(n))D_L(n)] = O\left(D_P^{2\rho_0}(n)\log [D_P^d(n)D_L(n)]\right) = o(\log n).
\end{align*}
Thus, we further derive that 
\begin{align*}
\mathbb{P}\left\{\max_{(\mathbf{k},l)\in\mathscr{A}_n}\left|\sum_{t=1}^n\xi_{nt,\mathbf{k},l}\right|>\varepsilon\sqrt{\log n}\right\}\leqslant 2\exp\left(\log [K(D_P(n))D_L(n)]-\frac{c_\varepsilon\log n}{B_{\pi,0}^2(D_P(n))}\right)\rightarrow 0,\,\,\, n\rightarrow\infty,
\end{align*}
which ends the proof of \autoref{proD3}.
\end{proof}

Gathering Propositions \ref{proD1}--\ref{proD3}, we obtain $\max_{(\mathbf{k},l)\in\mathscr{A}_n}|\hat{Z}_{\mathbf{k}l}| = o_p(\sqrt{\log n})$. Then, by \eqref{bound selection}, we have $\mathbb{P}\{(\hat{P},\hat{L})\neq (2,1)\}\rightarrow 0$, which implies $\mathbb{P}\{(\hat{P},\hat{L})=(2,1)\}\rightarrow 1$.

Note that for any $\varepsilon> 0$, 
\begin{align*}
\mathbb{P}\left\{\left|n\hat{\Psi}_n^2(K(\hat{P}),\hat{L}) - n\hat{\Psi}_n^2(K(2),1)\right|>\varepsilon\right\}\leqslant \mathbb{P}\{(\hat{P},\hat{L})\neq (2,1)\}\rightarrow 0,\,\,\,n\rightarrow\infty.
\end{align*}
Thus, by \autoref{Th2} and Slutsky's theorem, we have
\begin{align*}
n\hat{\Psi}_n^2(K(\hat{P}),\hat{L}) = n\hat{\Psi}_n^2(K(2),1) + o_p(1)\stackrel{d}\longrightarrow \chi_{d(d-1)/2}^2
\end{align*}
under the null hypothesis.

Let 
\begin{align*}
\theta_{\mathbf{k}l} = \int_0^1\int_{[0,1]^d}c(u;\tau)\pi_{\mathbf{k}}(u)\psi_l(\tau)dud\tau\,\,\, \text{and} \,\,\, \hat{\theta}_{\mathbf{k}l} = \frac{1}{n}\sum_{t=1}^n\pi_{\mathbf{k}}(\hat{U}_t)\psi_{l}(t/n).
\end{align*}
Under the fixed alternative $\mathcal{H}_1$ in \eqref{oriH_1}, there exist $\mathbf{k}_0$ and $l_0$ such that $\theta_{\mathbf{k}_0l_0}\neq 0$. Let $P_0 = \|\mathbf{k}_0\|_1$ and $L_0 = l_0$. Since $D_P(n)$ and $D_L(n)$ diverge to infinity, we have $P_0\leqslant D_P(n)$ and $L_0\leqslant D_L(n)$ for $n$ sufficiently large.

According to the proof of \autoref{Th3}, we have $\hat{\theta}_{\mathbf{k}l}\stackrel{p}\longrightarrow \theta_{\mathbf{k}l}$ for $(\mathbf{k},l)\in\mathscr{A}(P_0,L_0)$. Hence, there exists $q_0>0$, such that $\|\hat\theta_n(K(P_0),L_0)\|^2\stackrel{p}\longrightarrow q_0$. It follows that
\begin{align*}
\frac{SC_n(P_0,L_0)}{n} = \left\|\hat\theta_n(K(P_0),L_0)\right\|^2 - \frac{K(P_0)L_0\log n}{n}\stackrel{p}\longrightarrow q_0
\end{align*}
by Slutsky's theorem. Therefore, we obtain $\mathbb{P}\{SC_n(P_0,L_0)>nq_0/2\}\rightarrow 1$ as $n\rightarrow\infty$. Using the selection rule of the data-driven method, we have
\begin{align*}
n\left\|\hat\theta_n(K(\hat{P}),\hat{L})\right\|^2\geqslant SC_n(\hat{P},\hat{L})\geqslant SC_n(P_0,L_0)
\end{align*}
for $n$ sufficiently large. Thus, $\mathbb{P}\{\|\hat\theta_n(K(\hat{P}),\hat{L})\|^2>q_0/2\}\rightarrow 1$ as $n\rightarrow\infty$.

Consider the events 
\begin{align*}
\mathcal{E}_{n,1} = \left\{\min_{2\leqslant P\leqslant D_P(n)}
\lambda_{\min}(\hat{\Sigma}_{n,K(P)})>0\right\}\,\,\,\text{and}\,\,\, \mathcal{E}_{n,2} = \left\{\left\|\hat\theta_n(K(\hat{P}),\hat{L})\right\|^2>\frac{q_0}{2}\right\}.
\end{align*}
Note that when $\mathcal{E}_{n,1}$ occurs, 
\begin{align*}
\hat{\Psi}_n^2(K(\hat{P}),\hat{L})\geqslant \frac{\left\|\hat{\theta}_n(K(\hat{P}),\hat{L})\right\|^2}{\lambda_{\max} \left(\hat{\Sigma}_{n,K(\hat{P})}\right)}.
\end{align*}
In addition,
\begin{align*}
\lambda_{\max}(\hat{\Sigma}_{n,K(\hat{P})}) \leqslant \text{tr}\left(\hat{\Sigma}_{n,K(\hat{P})}\right)\leqslant \frac{1}{n}\sum_{t=1}^n\|\pi(\hat{U}_t)\|^2\leqslant K(\hat{P})B^2_{\pi,0}(\hat{P})\leqslant K(D_P(n))B^2_{\pi,0}(D_P(n)).
\end{align*}
Therefore, if $\mathcal{E}_{n,1}\cap \mathcal{E}_{n,2}$ occurs, we obtain
\begin{align*}
n\hat{\Psi}_n^2(K(\hat{P}),\hat{L})\geqslant \frac{nq_0}{2K(D_P(n))B^2_{\pi,0}(D_P(n))}.
\end{align*}
Under \autoref{AD1}(iii), we have
\begin{align*}
\frac{K(D_P(n))B^2_{\pi,0}(D_P(n))}{n} = O(n^{-1}D_P^{d+2\rho_0}(n)) = o(1).
\end{align*}
Hence, $nq_0/[2K(D_P(n))B^2_{\pi,0}(D_P(n))]\rightarrow\infty$ as $n\rightarrow\infty$. For any $x>0$, as long as $n$ is large enough, we have $nq_0/[2K(D_P(n))B^2_{\pi,0}(D_P(n))]>x$. It then follows that $\mathcal{E}_{n,1}\cap \mathcal{E}_{n,2}\subset \{n\hat{\Psi}_n^2(K(\hat{P}),\hat{L})>x\}$. Thus, 
\begin{align*}
\mathbb{P}\left\{n\hat{\Psi}_n^2(K(\hat{P}),\hat{L})\leqslant x\right\}\leqslant \mathbb{P}\{\mathcal{E}_{n,1}^c\} + \mathbb{P}\{\mathcal{E}_{n,2}^c\}
\rightarrow 0,\,\,\,n\rightarrow\infty,
\end{align*}
which completes the proof.
\end{proof}

\section*{References}
\spacingset{1}

\begingroup
\setlength{\leftskip}{0.5em}
\setlength{\parindent}{-0.5em}
\setlength{\parskip}{0.2\baselineskip}

\noindent\hspace*{-0.5em}Bae, J., D. Jun, and S. Levental (2010). The uniform CLT for martingale difference arrays under the uniformly integrable entropy. \textit{Bull. Korean Math. Soc} 47 (1), 39–51.

Massart, P. (1990). The tight constant in the Dvoretzky--Kiefer--Wolfowitz inequality. \textit{The annals of Probability}, 1269–1283.

Pollard, D. (2012). \textit{Convergence of stochastic processes}. Springer Science \& Business Media.

van der Vaart, A. W. (1998). \textit{Asymptotic Statistics}. Cambridge Series in Statistical and Probabilistic Mathematics. Cambridge University Press.

van der Vaart, A. W. and J. A. Wellner (1996). \textit{Weak Convergence and Empirical Processes: With Applications to Statistics}. Springer Series in Statistics. Springer.

\par
\endgroup

\end{document}